\documentclass[11pt]{article}

\usepackage[margin=1in]{geometry}
\usepackage{amsmath,amssymb,mathtools,amsthm}
\usepackage{booktabs,tabularx,array}
\usepackage{enumitem}
\usepackage{microtype}
\usepackage{setspace}
\usepackage[authoryear,round]{natbib}
\usepackage{tikz}
\usetikzlibrary{arrows.meta}

\setlist[itemize]{leftmargin=1.5em,itemsep=0.2em,topsep=0.25em}
\setlist[enumerate]{leftmargin=1.7em,itemsep=0.2em,topsep=0.25em}

\theoremstyle{plain}
\newtheorem{theorem}{Theorem}
\newtheorem{proposition}{Proposition}

\newtheorem{corollary}{Corollary}

\theoremstyle{definition}
\newtheorem{definition}{Definition}

\theoremstyle{remark}
\newtheorem{remark}{Remark}

\newcommand{\E}{\mathbb{E}}

\newcommand{\R}{\mathbb{R}}

\begin{document}

\title{Dynamic Mechanism and Information Design in AI-Agent Environments:
 A General Analytical Framework}

\author{Guoqiang Tian\thanks{Department of Economics, Texas A\&M University, USA. Email: gtian@tamu.edu.}}

\date{First draft: September, 2026; Revised: October 27, 2026}

\maketitle

\begin{abstract}

This paper develops a general framework for dynamic mechanism and information design in AI-agent principal--agent environments, encompassing optimal contracting and integrating
evolving hidden information and hidden action with partial verification,
technological authorization, and multi-agent interaction. It nests conventional
dynamic screening, hidden-action contracting, and conditional dynamic
information design, while allowing reporting, action, and information-use
decisions to interact through a common continuation mechanism. We establish a dynamic revelation principle and a recursive characterization of implementability in which truthful reporting and post-report obedience are jointly determined through continuation
incentives. In quasilinear optimal-contracting environments, dynamic envelope and
integral-monotonicity methods yield a verification-adjusted dynamic
virtual-surplus representation and identify three distinct margins: dynamic
information rents, verification rents, and the implementation cost of hidden
action. These margins generate a branched hierarchy separating conventional
dynamic screening, hidden-action contracting, and their joint problem. We also
study the joint design of monitoring, information use, and incentives: more
informative monitoring capacity expands the designer's opportunity set, but
full disclosure need not be optimal, and information use can itself alter
hidden-action incentives. With multiple agents, correlated information, peer
discipline, coupled incentives, joint feasibility, and continuation
opportunities create additional strategic interactions. Although motivated by AI-agent environments, the framework applies more broadly to dynamic agency settings with similar economic features.

\medskip

\noindent
\textbf{JEL Classification:} D82, D86, C73, D83.

\noindent
\textbf{Keywords:} AI agents, dynamic mechanism design, dynamic information
design, hidden information, hidden action, partial verification,
authorization, monitoring, multi-agent mechanisms.
\end{abstract}

\pagebreak

\setcounter{tocdepth}{2}
\tableofcontents

\section{Introduction}

AI systems are increasingly deployed as agents that act over time rather than
as devices making isolated predictions or recommendations. They search for
information, write and execute code, call tools and external systems, interact
with other agents, and carry out sequences of tasks in which later decisions
depend on earlier actions and observations. As these systems become capable of
longer and more consequential tasks, the relevant economic problem is
increasingly one of delegation and control over an evolving relationship.
Monitoring today may determine the permissions granted tomorrow; current
behavior may change what the designer learns about the system; and reports or
actions may alter future rewards, deployment opportunities, or interaction
with other agents. The problem is therefore not simply how to induce a
desirable action at a particular date, but how to organize information,
incentives, and control over time.

Recent work on AI behavior provides a direct motivation for treating this
problem as strategic. AI systems have been studied in environments involving
goal misgeneralization, learned or emergent objectives, deceptive behavior,
alignment faking, awareness of evaluation, strategic underperformance, and
interaction among multiple artificial agents
\citep{ShahEtAl2022,CasperEtAl2023,ChenEtAl2023,HubingerEtAl2024,
GreenblattEtAl2024,VanDerWeijEtAl2025,NeedhamEtAl2025,HammondEtAl2025}.
Other work emphasizes both the growing ability of AI systems to complete
extended tasks and the importance and fragility of monitoring their internal
or external behavior
\citep{KwaEtAl2025,KorbakEtAl2025,GuanEtAl2025,MazeikaEtAl2025}.
These developments suggest that neither the agent's private characteristics
nor its subsequent behavior should automatically be treated as transparent to
the designer. What an AI agent can do, what it knows, what objective governs
its behavior, and what it chooses to reveal may be private information, while
the action it ultimately takes may itself be imperfectly observed.

Two features make the resulting mechanism-design problem particularly natural
for AI agents. First, private information can coexist with technologically
enforceable control. An agent may know more than its designer about its
behavioral objective, execution capability, beliefs, or information
capability, while the designer may be able to restrict access to tools,
external systems, communication channels, computational resources,
transactions, or deployment. Execution capability and authorization are
therefore distinct economic objects: capability describes what the agent can
do, whereas authorization describes what the mechanism permits it to do.
Second, communication need not determine subsequent behavior. After reporting,
the agent may receive additional private information and then choose its own
action. That action need not itself be publicly observed. A mechanism must
therefore induce not only truthful communication but also desirable behavior
after communication.

Dynamics make these information and control margins interact. Reports,
privately chosen actions, monitoring evidence, and public outcomes today can
change the mechanism an agent faces tomorrow. Future authorization can reward
or discipline current behavior; monitoring can affect incentives because the
information it generates changes subsequent treatment; and future
opportunities can create strategic interaction among agents even when their
current primitive payoffs are separable. The continuation mechanism is
therefore itself an incentive instrument. This feedback from current
information and behavior to future control, and from future control back to
current reporting and action incentives, is the central economic force studied
in this paper.

A growing economics literature studies related problems of delegation,
alignment, and control with artificial agents. Hadfield-Menell and Hadfield
\citeyearpar{HadfieldMenellHadfield2019} formulate AI alignment as an
incomplete-contracting problem, while Fudenberg and Liang
\citeyearpar{FudenbergLiang2026} study delegation to an AI whose alignment is
unknown. Dworczak and Smolin \citeyearpar{DworczakSmolin2026} study decision
making when advice may come from a potentially misaligned source. Related work
considers environments in which the principal does not fully know the agent's
feasible choice set
\citep{HurwiczShapiro1978,Carroll2015,GanLi2026}, and competition among
artificial agents as a means of disciplining behavior
\citep{CollinaEtAl2025}. These approaches address different manifestations of
a common problem: the designer may not fully know the objectives,
capabilities, information, or behavior of the system to which decisions are
delegated.

The closest AI-agent mechanism-design benchmark is Bergemann, Koh, and Morris
\citeyearpar{BergemannKohMorris2026}. They develop a static framework in which
an AI agent may privately know its alignment, feasible actions, information,
and beliefs. Because communication is followed by an action chosen by the
agent, implementation must induce both honesty and obedience. Their one-sided
imitation structure captures an important asymmetry in capability revelation:
a more capable agent may conceal capability or imitate a less capable one,
while a less capable agent cannot necessarily counterfeit capabilities it does
not possess. This structure is particularly relevant for AI evaluation, where
observed performance may reflect strategic underperformance rather than
technological limits.

A first broader theoretical benchmark comes from dynamic mechanism design,
where private information and incentives evolve over time. Baron and Besanko
\citeyearpar{BaronBesanko1984} study regulation in a continuing relationship,
Courty and Li \citeyearpar{CourtyLi2000} analyze sequential screening, and
Athey and Segal \citeyearpar{AtheySegal2013} construct efficient dynamic
mechanisms with evolving private information. Pavan, Segal, and Toikka
\citeyearpar{PavanSegalToikka2014} develop a Myersonian approach in which
impulse responses measure how current private information propagates into
future types and dynamic envelope methods translate this persistence into
information rents, virtual surplus, and intertemporal distortions. Bergemann
and V\"alim\"aki \citeyearpar{BergemannValimaki2019} and Pavan
\citeyearpar{Pavan2026} provide broader treatments of dynamic mechanism
design.

A second broader benchmark comes from principal--agent theory with hidden
action. The joint presence of adverse selection and moral hazard is analyzed,
among others, by Laffont and Tirole \citeyearpar{LaffontTirole1986}, Picard
\citeyearpar{Picard1987}, Guesnerie, Picard, and Rey
\citeyearpar{GuesneriePicardRey1989}, and Meng and Tian
\citeyearpar{MengTian2013}. In particular, Meng and Tian
\citeyearpar{MengTian2013} study hidden action together with multidimensional
private characteristics in a static multitask contracting environment.
Dynamic environments combining evolving hidden types and hidden actions
include Es\H{o} and Szentes \citeyearpar{EsoSzentes2017}. These literatures
make clear that hidden information and hidden action are distinct incentive
problems even when they interact within the same contract.

A third broader benchmark comes from information design. Bayesian persuasion
treats the information structure itself as a design instrument
\citep{KamenicaGentzkow2011}, while the comparison and design of information
structures are developed more generally by
\citet{BergemannMorris2016,BergemannMorris2019}. Dynamic and sequential
information-design problems study how information released today changes
future beliefs, actions, and continuation decisions
\citep{DovalEly2020,ElyGeorgiadisRayo2025}, while recent work brings
information design and mechanism design into a common framework
\citep{BergemannHeumannMorris2026,BergemannHeumannMorrisJPE2026}. In these
problems, information is not merely an exogenous input into a mechanism: what
is revealed, when it is revealed, and how it is used can itself be a dynamic
design choice.

Three complementary foundations guide the present analysis. The static
AI-agent framework of Bergemann, Koh, and Morris
\citeyearpar{BergemannKohMorris2026} provides the capability,
one-sided-imitation, and joint honesty--obedience structure. Pavan, Segal,
and Toikka \citeyearpar{PavanSegalToikka2014} provide dynamic screening,
controlled impulse responses, envelope methods, and virtual surplus.
Tian \citeyearpar{TianGeneralDynamic2026} develops a general theory of
dynamic mechanism design with endogenous information and hidden actions,
including joint dynamic implementation, incentive-friction benchmarks,
dynamic rents, and hidden-action implementation costs. The present paper
builds on these three foundations without adopting their models as an
identical timing or a nested set of primitives: its AI-agent timing,
post-report information, verification correspondence, authorization,
monitoring, and continuation control define a distinct economic environment.

The gap among these benchmarks is substantive. Standard dynamic screening
allows private information to evolve, but the payoff-relevant allocation is
typically directly assigned and the principal strategic problem is reporting.
Hidden-action contracting allows behavior to remain privately chosen, but
generally abstracts from the full conjunction of evolving private information,
one-sided verification, technological authorization, and endogenous
information use. Dynamic information design makes information itself a policy
instrument, but typically does not embed that instrument in a dynamic agency
problem in which the agent both privately reports evolving information and
subsequently chooses a potentially hidden action. The static AI-agent
benchmark allows rich private characteristics together with post-report
action, but does not address how those characteristics evolve or how current
reports, actions, and monitoring information reshape future mechanism choices.

The objective here is therefore not to introduce joint dynamic incentives as an entirely new general problem, but to develop their consequences for AI-agent delegation in an environment with its own information and decision timing, partial verification, and technologically enforceable controls. AI-agent environments naturally generate precisely this interaction. Private characteristics can evolve; reporting can be constrained by verification; additional private information can arrive after communication; the agent can retain hidden post-report action discretion; monitoring can generate evidence about states, types, or behavior; and the designer can use authorization, rewards, information disclosure, and future opportunities as dynamic instruments. The resulting combination goes beyond each foundation in a different direction: it makes the static AI-agent reporting and obedience problem dynamic, adapts dynamic screening tools to private post-report action and restricted claims, and develops the integrated contracting perspective under AI-specific verification and control instruments.

The paper contributes at two levels. First, at the framework level, we develop a general dynamic mechanism and information-design framework for AI-agent principal--agent environments that integrates evolving hidden information and hidden action with partial verification, technological authorization, and multi-agent interaction, while encompassing optimal contracting as a central design problem. At each date, the agent has a private type comprising its behavioral objective, execution capability, subjective belief, and information capability, and this private type may evolve over time. After communicating with the mechanism, the agent may receive additional private information and then privately choose its realized action. The action is distinct from the private-type profile and need not be publicly observed, although it may affect public outcomes, monitoring evidence, physical transitions, and future private information. A primitive monitoring technology generates evidence, while the mechanism chooses how that evidence is used in the public continuation state. The designer commits to a history-dependent mechanism governing recommendations, authorization, rewards, information use, and continuation decisions. The framework therefore nests conventional dynamic screening, hidden-action contracting, and conditional dynamic information design, while allowing reporting, action, and information-use incentives to interact through a common continuation mechanism. Although motivated by AI-agent environments, the framework is not specific to artificial agents and applies more broadly to dynamic agency settings with similar economic features.

Two informational objects play distinct roles. We use
\textbf{verification} for restrictions on which reports or claims can be
substantiated and \textbf{monitoring} for technologies that generate evidence
about latent states, private characteristics, or behavior. An AI evaluation
may serve either role in practice, but economically the distinction matters:
verification restricts the feasible report domain, whereas monitoring
determines what evidence is available for subsequent information and control
decisions. The primitive monitoring technology is fixed in the baseline
model, but the mechanism chooses how the resulting evidence is processed and
incorporated into the public continuation state. When behavior is hidden, the
monitoring technology may depend on the latent action even though the
information-use policy cannot condition directly on that action unless it is
publicly observed.

The present model has its own timing and probability structure, rather
than directly adopting the ordered information-production and disclosure
kernels of \citet{TianGeneralDynamic2026}. It carries forward the joint
incentive and institutional-design perspective while taking AI-specific
monitoring technologies as primitive and making their strategic use part of
the committed mechanism.

This distinction places the model between conventional dynamic information
design and a fully endogenous informational system. Information production is
taken as primitive, but information use is endogenous and dynamically linked
to incentives, behavior, beliefs, and future control. The model therefore
endogenizes the informational adjustment generated by the mechanism while
stopping short of endogenizing information production itself. Endogenous
information acquisition, experimentation, and the joint evolution of
information production and feasibility lead to the broader system-level
problem studied in \citet{Tian2026}.

Second, at the \textbf{theoretical level}, the framework yields four main
sets of results on dynamic implementability, optimal contracting, information
design, and multi-agent incentives. These results can be summarized as
follows.

First, we establish the revelation and implementability foundations of the
framework. A dynamic revelation principle shows that, under partial
verification, indirect mechanisms can be replaced by direct mechanisms in
which truthful reporting and post-report obedience are jointly incentive
compatible. Because the agent receives additional private information and
chooses its own action after reporting, a deviation may combine a false report
with subsequent private action reoptimization. Dynamic incentive compatibility
therefore has a nested structure: each feasible report must be evaluated
together with the best post-report action available to the true type, while
the truthful-report comparison must itself sustain the recommended action. We
give an equivalent recursive characterization that separates honesty,
obedience, and promise consistency while linking all three through
continuation values. When the action is directly assigned, the inner action
problem disappears and the recursion reduces to the standard dynamic
reporting problem.

Second, we characterize optimal dynamic mechanism design when hidden
information, hidden action, and one-sided verification coexist. We focus on a
regular scalar quasilinear environment for the hidden-information dimension
while retaining post-report private action choice. Dynamic envelope and
integral-monotonicity methods remain the analytical foundation, but the
additional margins modify their implications in distinct ways. One-sided
verification can place truthful reporting at the boundary of the feasible
report set, so dynamic incentive compatibility need not imply the usual
two-sided envelope equality. We capture this departure through a nonnegative
\textbf{verification wedge}, which vanishes and restores the standard
envelope implication when verification is locally nonbinding. Hidden action
operates differently: because the true type can reoptimize its action after
every feasible report, it changes the value associated with each report and
restricts the set of mechanisms capable of implementing the desired
truthful-path behavior.

Optimal design consequently reflects three distinct margins: the familiar
\textbf{dynamic information rent}, an additional \textbf{verification rent},
and the \textbf{implementation cost of hidden action}. The first two enter the
rent representation, whereas the third generally operates through the
implementable set rather than as another additive rent. This yields a
verification-adjusted dynamic virtual-surplus representation, a
virtual-surplus upper bound, and a relaxed virtual-surplus candidate. The
candidate solves the full problem when it satisfies the nested report--action
incentive condition together with the conditions under which the rent
representation is tight.

The resulting benchmark structure is branched rather than linear. The full
\(4B\) environment combines hidden information, one-sided verification, and
hidden action. The two second-best branches isolate the classical frictions:
\(2B^{\mathrm{HI}}\) is conventional dynamic hidden information with directly
assignable action, whereas \(2B^{\mathrm{HA}}\) is the complete-information
hidden-action problem. Adding one-sided verification to the former gives
\(3B^{\mathrm{HI+V}}\), while adding hidden action to dynamic screening under
unrestricted reporting gives \(3B^{\mathrm{HI+HA}}\). First best requires
both private-information and action-implementation frictions to be absent.
Thus eliminating private information does not in general restore first best
when hidden action remains, just as eliminating hidden action does not restore
first best when private information remains. Value comparisons follow from
genuine feasible-set inclusion along a given branch, but there is no general
ranking between the two second-best branches, and comparisons across different
verification technologies require the corresponding report domains to be
taken into account.

The virtual-surplus representation and hidden-action implementation condition
also yield a closed-form distortion decomposition. Under type-linear payoffs,
hidden information is summarized by a \textbf{dynamic virtual type} that
adjusts the current private type for the information rent propagated through
the impulse response. Hidden action does not create a second virtual-type
term; instead, it restricts the relaxed control to what can be supported by
the available action incentives. In a linear-quadratic benchmark, we
characterize all six benchmark controls and separate the distortion due to
dynamic hidden information from the additional obedience distortion generated
by bounded marginal action incentives. Under an AR(1) type process, the
informational component decays geometrically with the horizon, while the
hidden-action component depends on whether the required action incentive
binds. The highest initial type therefore has no information-rent distortion
under the usual regularity condition but may still face a hidden-action
distortion, and the total \(4B\) distortion need not decline monotonically over
calendar time. When control represents autonomy, tool access, compute, or
another delegated resource, these results provide a mechanism-design
rationale for staged control without requiring the realized control path
itself to be monotone.

Third, we develop the dynamic information-design component of the framework.
We distinguish \textbf{monitoring capacity}, the informativeness of the
primitive monitoring technology, from the \textbf{information-use policy}
that determines how the resulting evidence is processed, disclosed, and
incorporated into the continuation mechanism. The primitive experiment
determines the informational frontier available to the designer, while the
information-use policy selects a public experiment from that frontier.
Because public information changes future authorization, rewards, deployment,
and other mechanism choices, information design is intrinsically dynamic even
when the primitive monitoring technology itself is fixed.

The interaction with mechanism design becomes sharper when behavior is hidden.
The raw monitoring experiment may depend on the privately chosen action, so
different actions generate different distributions of evidence. Although the
information-use policy cannot condition directly on the latent action, it
determines how strongly those behavioral differences are transmitted into the
public continuation state. Information use can therefore change both what the
designer learns and the continuation incentives supporting current behavior.
Information design and hidden-action implementation are thus jointly
determined through the same continuation mechanism.

We characterize the feasible public experiments as the Blackwell garblings of
the primitive monitoring technology. Under free informational disposal, more
informative monitoring capacity weakly expands the designer's opportunity set
because every policy available under a less informative experiment can be
replicated. The replication argument holds conditional on the latent state,
private type, and hidden action and therefore preserves both truthful behavior
and feasible report--action deviations. Greater monitoring capacity, however,
does not imply that greater public disclosure is optimal. Information
revealed today changes future control and can thereby alter current reporting
and action incentives, so full disclosure may be dominated by selective
disclosure even when additional information has positive direct decision
value. We characterize the boundary between information and incentive design:
separation obtains precisely when the relevant continuation-value differences
between truthful-and-obedient behavior and feasible report--action deviations
are invariant to the information-use policy. Otherwise, information use
changes the incentive constraints themselves and becomes an intertemporal
incentive instrument.

Fourth, we extend the framework to multiple AI agents. The dynamic revelation
principle and recursive incentive characterization extend agent by agent while
allowing correlated private information, higher-order beliefs, latent actions,
and a common continuation mechanism. Under an independence benchmark, the
dynamic rent and virtual-surplus analysis extends agent by agent, as does the
branched hidden-information--hidden-action structure. Hidden action remains
embedded in each agent's nested obedience condition, while joint feasibility
adds a distinct allocation margin. With a shared resource, for example,
optimal allocation depends on marginal dynamic virtual values rather than
current private values alone.

Correlation creates another incentive instrument. An agent's private
information may predict the private information of other agents, allowing peer
information to discipline false reporting. With finite report support and
sufficiently distinct predictive distributions, sufficiently strong proper
peer scoring can eliminate profitable false reports. Strict
affiliation of ordered finite-support types provides a primitive sufficient
condition for this predictive separation. This logic has a sharp limitation
in smooth continuous environments: a nearby false report may generate a
first-order gain while a smooth quadratic peer score imposes only a
second-order loss. Peer information therefore need not eliminate the
underlying reporting and verification friction, nor does it by itself resolve
the hidden-action obedience problem.

Multi-agent incentives also operate through current observable consequences
and continuation opportunities. A privately chosen action may affect the
distribution of public outcomes or other contractible signals and thereby the
agent's expected current reward, while also changing future authorization,
deployment, task assignments, or other opportunities. The marginal incentive
for obedience therefore contains both an induced current-reward component and
a continuation-value component. When these consequences depend jointly on the
actions of several agents, the agentwise obedience conditions become
strategically coupled even when primitive current payoffs are separable.
Direct reward dependence on the full action profile is only the
observable-action benchmark; the general result requires only that actions
affect publicly contractible outcomes, monitoring evidence, physical
transitions, or continuation states.

Beyond these principal benchmarks, the revelation and verification structure
builds on work on communication, partial verification, and evidence
\citep{Myerson1982,Forges1986,GreenLaffont1986,
BenPorathDekelLipman2019}, while the dynamic incentive analysis relates to
broader work on dynamic screening and multistage revelation
\citep{Battaglini2005,BergemannValimaki2010,GarrettPavan2015,
SugayaWolitzky2021}. The present framework combines evolving private
information with verification-restricted communication and post-report action
discretion, so deviations may involve both communication and private action
reoptimization and may alter the continuation mechanism. On the
information-design side, the distinctive feature relative to the benchmark
literatures discussed above is that information use is embedded in this same
dynamic agency problem: monitoring evidence enters the continuation mechanism,
may itself be generated by hidden behavior, and therefore affects both future
decisions and current report--action incentives.

The multi-agent analysis is related to work on dynamic mechanisms with network
interactions \citep{MengTian2021,MengSunTian2022} and peer prediction and weak
supervision \citep{QiuCarrollAllen2026}. The distinction between execution
capability and authorization also connects the analysis to work on AI
oversight, interpretability, adaptive deployment, and technological
architectures for controlling tool-using agents
\citep{WenEtAl2024,HerrmannLevinstein2026,WEF2026,Microsoft2026,
Saha2026,SurapaniEtAl2026}. These literatures isolate different
informational, technological, or behavioral components of the control
problem; the present paper studies how those components interact within a
common dynamic mechanism and information-design environment.

The remainder of the paper is organized as follows. Section~2 introduces the
single-agent dynamic environment, hidden action, and partial verification.
Section~3 develops indirect and direct mechanisms, establishes the dynamic
revelation principle, and provides the recursive characterization of joint
reporting and action incentives. Section~4 studies optimal dynamic mechanism
design with hidden information and hidden action, derives the
verification-adjusted dynamic virtual-surplus representation, and
characterizes the branched distortion hierarchy. Section~5 develops the
dynamic monitoring and information-design problem and studies feasible public
experiments, selective disclosure, and the joint design of information and
incentives. Section~6 extends the framework to multiple agents and studies
dynamic virtual surplus, hidden-action implementation, peer information,
coupled incentives, joint feasibility, and continuation externalities.
Section~7 identifies the inherited, modified, and additional features of the
framework relative to standard dynamic contracting. Section~8 develops
applications to staged autonomy, strategic evaluation, and competition through
continuation opportunities. Section~9 concludes.

\section{A Dynamic Model with a Single AI Agent}

This section specifies the dynamic single-agent environment that underlies the mechanism-design analysis. The representation of the agent's private characteristics draws on the static AI-agent framework of Bergemann, Koh, and Morris \citeyearpar{BergemannKohMorris2026}; the controlled evolution of private information and subsequent rent analysis use methods associated with Pavan, Segal, and Toikka \citeyearpar{PavanSegalToikka2014}; and the joint treatment of reporting, hidden actions, and continuation implementation builds on the general dynamic framework of Tian \citeyearpar{TianGeneralDynamic2026}. The timing and primitive kernels here are specified separately for AI delegation, including post-report private information, partial verification, and authorization. We first specify the primitives of the environment, then the dynamic evolution and within-period timing, and finally the partial-verification structure.

\subsection{Environment}

Time is discrete, with dates \(t=1,2,\ldots,T\), where \(T\) may be finite or infinite. There is one designer (the principal) and one AI agent. At each date \(t\), the environment includes an unobserved common payoff-relevant state \(\theta_t\in\Theta_t\), an action \(a_t\in\mathcal A_t\) chosen by the agent, and a publicly observed outcome \(o_t\in\mathcal O_t\). The action \(a_t\) need not itself be publicly observed; its observability is determined by the public-outcome technology introduced below. Let \(\mathcal S_t\) denote the space of possible realizations of a private signal about the payoff-relevant state \(\theta_t\).

All primitive spaces are measurable spaces, all functions and stochastic kernels defined on them are understood to be measurable, and \(\Delta(X)\) denotes the set of probability measures on \(X\). Function spaces and collections of subsets introduced below are endowed with maintained measurable structures under which the relevant maps obtained by applying these functions and correspondences are measurable. Random variables are distinguished from their realizations by a tilde; thus \(\widetilde\theta_t\) denotes the random state and \(\theta_t\) its realization, with the same convention used for types, signals, actions, and outcomes.

Let \(H_t\) denote the space of public histories immediately before date-\(t\) communication, with \(H_1=\{\varnothing\}\), and let \(h_t\in H_t\) denote a realized public history. For \(t\geq2\), \(h_t\) records the publicly observed history generated at earlier dates. We leave its coordinate representation abstract at this stage; Section~3 specifies the public record induced by the mechanism once the relevant communication and control objects have been introduced.

Let
\[
\mathcal U_t
=
\left\{
u:\mathcal A_t\times\Theta_t\rightarrow\mathbb R
\right\}
\]
denote the space of the agent's period-\(t\) payoff representations, and let
\[
\mathcal V_t
=
\left\{
v:\mathcal A_t\times\Theta_t\times\mathcal O_t
\rightarrow\mathbb R
\right\}
\]
denote the space of the designer's period-\(t\) payoff functions.

The agent's payoff representation is interpreted in reduced form. One interpretation, following the literature, including Bergemann, Koh, and Morris \citeyearpar{BergemannKohMorris2026}, is that it represents a behavioral objective induced by training, reinforcement, or other features of system design. Alternatively, it may be understood as a utility representation of coherent observed choice behavior, in the tradition of Samuelson's revealed-preference approach and the classical Hurwicz--Uzawa integrability framework; see Samuelson \citeyearpar{Samuelson1938} and Hurwicz and Uzawa \citeyearpar{HurwiczUzawa1971}.

Let
\[
\mathcal B_t
=
\left\{
b:H_t\rightarrow\Delta(\Theta_t)
\right\}
\]
denote the space of \textbf{subjective belief kernels} over the payoff-relevant state \(\theta_t\) conditional on public history. Thus \(b(\cdot\mid h_t)\in\Delta(\Theta_t)\) is a probability distribution over \(\theta_t\) after public history \(h_t\).

Let
\[
\Pi_t
=
\left\{
\pi:\Theta_t\rightarrow\Delta(\mathcal S_t)
\right\}
\]
denote the space of \textbf{information experiments} from payoff-relevant states to private signals. An experiment \(\pi\in\Pi_t\) specifies, for each \(\theta_t\in\Theta_t\), a conditional distribution over private-signal realizations \(s_t\in\mathcal S_t\).

The \textbf{private type space} satisfies
\[
\mathcal T_t
\subseteq
\left(2^{\mathcal A_t}\setminus\{\emptyset\}\right)
\times
\mathcal U_t
\times
\mathcal B_t
\times
\Pi_t .
\]
Thus a private type may contain heterogeneity in execution capability, behavioral objective, subjective beliefs about the payoff-relevant state, and information capability.

At date \(t\), the agent privately observes
\[
\tau_t=(A_t,u_t,b_t,\pi_t)\in\mathcal T_t,
\]
where \(A_t\subseteq\mathcal A_t\) is its \textbf{execution capability}, the set of actions it can actually take; \(u_t\in\mathcal U_t\) is its period-\(t\) \textbf{behavioral objective}, assigning payoffs to action--state pairs \((a_t,\theta_t)\in\mathcal A_t\times\Theta_t\); \(b_t\in\mathcal B_t\) is its \textbf{subjective belief kernel}, with \(b_t(\cdot\mid h_t)\in\Delta(\Theta_t)\) describing its belief about the payoff-relevant state conditional on the public history; and \(\pi_t\in\Pi_t\) is its \textbf{information capability}, with \(\pi_t(\cdot\mid\theta_t)\in\Delta(\mathcal S_t)\) describing the distribution of the private signal it receives conditional on the payoff-relevant state. Accordingly,
\[
\widetilde\tau_t
=
(\widetilde A_t,\widetilde u_t,\widetilde b_t,\widetilde\pi_t)
\]
denotes the corresponding random type.

The distinction between \(b_t\) and \(\pi_t\) is substantive. The belief kernel \(b_t\) describes the agent's assessment of the unobserved state before its period-\(t\) private signal is realized, whereas the experiment \(\pi_t\) describes how that state generates additional private information. Conditional on the realized state \(\theta_t\), the realized \textbf{private signal} \(s_t\in\mathcal S_t\) is drawn according to \(\pi_t(\cdot\mid\theta_t)\) and is observed by the agent after communication, as specified in the timing below.

\begin{remark}[Nested agency environments]
\label{rem:nested_type_environments}
The product representation of \(\tau_t\), together with the distinction between the agent's execution capability, endogenous action, and publicly observed outcome, nests a number of familiar agency environments. If only the payoff component \(u_t\) varies privately and the utility-relevant action is directly assigned, the model reduces to conventional adverse selection or screening. If the execution capability \(A_t\) is common knowledge but the agent chooses \(a_t\in A_t\) and the public outcome \(o_t\) does not fully reveal that action, the model contains the standard hidden-action or moral-hazard problem. If \(A_t\) itself is private, the designer additionally faces private information about the agent's execution capability. When the action is publicly observed, or equivalently when the public outcome reveals it, the same formulation yields post-report action discretion with observable behavior rather than hidden action.

If only \(b_t\) varies privately, private information concerns the agent's subjective beliefs about the payoff-relevant state, whereas if only \(\pi_t\) varies privately, agents differ in the information technology through which they learn about that state. These dimensions may also vary jointly. The general formulation therefore accommodates adverse selection, hidden action, private execution capability, heterogeneous beliefs, private information capabilities, and post-report action discretion within a common dynamic agency framework. Section~4 focuses on a scalar quasilinear screening environment in which execution capability, beliefs, and information capability are held fixed while only the payoff component varies privately.
\end{remark}

The designer's period payoff function \(v_t\in\mathcal V_t\) is commonly known. If the mechanism provides a \textbf{reward} \(r_t\in\mathbb R\), the period payoffs are
\[
u_t(a_t,\theta_t)+r_t,
\qquad
v_t(a_t,\theta_t,o_t)-c_t(r_t),
\]
for the agent and the designer, respectively, where \(c_t:\mathbb R\rightarrow\mathbb R\) is the cost of providing the reward. The standard monetary-transfer case is \(c_t(r)=r\).

Both parties discount future payoffs by \(\delta\in(0,1]\). Their realized intertemporal payoffs are
\[
U^A
=
\sum_{t=1}^{T}
\delta^{t-1}
\left[
u_t(a_t,\theta_t)+r_t
\right]
\]
and
\[
U^P
=
\sum_{t=1}^{T}
\delta^{t-1}
\left[
v_t(a_t,\theta_t,o_t)-c_t(r_t)
\right].
\]
When \(T=\infty\), we restrict attention to environments and admissible mechanisms for which the relevant expected discounted payoffs are well defined and finite.

At date \(0\), the designer commits to a mechanism governing the relationship over the horizon. When participation is imposed in the optimal-design problem below, the agent observes its initial private type before deciding whether to participate; its outside option is normalized to zero.

\subsection{Dynamic evolution and timing}

In addition to the standard dynamic evolution of private information, the present environment incorporates post-report action discretion. After communicating with the mechanism, the agent may receive additional private information and then choose a realized action that need not itself be publicly observed. That action can affect current observable outcomes as well as the evolution of future payoff-relevant states and private characteristics. We first specify the outcome and transition technologies and then summarize the resulting within-period interaction.

Public outcomes are generated according to the stochastic kernel
\[
Q_t:
\Theta_t\times\mathcal A_t
\longrightarrow
\Delta(\mathcal O_t),
\]
so that, conditional on the current state and realized action,
\[
\widetilde o_t
\sim
Q_t(\cdot\mid\theta_t,a_t).
\]
Thus \(Q_t(\cdot\mid\theta_t,a_t)\) is the conditional distribution of the public outcome generated by the realized action in state \(\theta_t\). The action \(a_t\) itself is not assumed to be publicly observed; public information about the action is mediated through the outcome kernel,
\[
a_t
\longrightarrow
Q_t(\cdot\mid\theta_t,a_t)
\longrightarrow
o_t.
\]
If the induced public outcome does not reveal \(a_t\), the action is hidden and incentives operate through its observable consequences. If the outcome reveals the action perfectly, for example when
\[
o_t=(a_t,\bar o_t),
\]
the model specializes to the observable-action environment. The agent's interim private signal is generated according to its realized information capability,
\[
\widetilde s_t
\sim
\pi_t(\cdot\mid\theta_t).
\]
For each \(t<T\), let
\[
K_{t+1}:
H_t
\times
\Theta_t
\times
\mathcal T_t
\times
\mathcal A_t
\times
\mathcal O_t
\longrightarrow
\Delta(\Theta_{t+1}\times\mathcal T_{t+1})
\]
be the transition kernel for the payoff-relevant state and private type. Recall that the type \(\tau_t\) contains the execution capability \(A_t\subseteq\mathcal A_t\), whereas \(a_t\in A_t\) is an action actually executable by the agent. Conditional on the realized public history, current state and type, realized action, and public outcome,
\[
(\widetilde\theta_{t+1},\widetilde\tau_{t+1})
\sim
K_{t+1}
\bigl(
\cdot
\mid
h_t,\theta_t,\tau_t,a_t,o_t
\bigr).
\]
The transition process may therefore be exogenous or decision dependent. Current actions and outcomes may affect future payoff-relevant states, behavioral objectives, execution capabilities, beliefs, or information capabilities. This physical dependence of the transition on the realized action does not require the action itself to be publicly observed.

The designer's beliefs need not coincide with the agent's subjective beliefs. Let \(P_1\in\Delta(\Theta_1\times\mathcal T_1)\) denote the designer's initial joint distribution over the state and private type, with $
(\widetilde\theta_1,\widetilde\tau_1)\sim P_1.$
Let \(P_t^P(\cdot\mid h_t,\tau_t)\in\Delta(\Theta_t)\) denote the designer's conditional distribution of the current state given the public history and current type whenever this conditional distribution is defined. This conditional law is an analytical object and does not require the designer to observe \(\tau_t\). We impose the absolute-continuity condition
\[
P_t^P(\cdot\mid h_t,\tau_t)
\ll
b_t(\cdot\mid h_t),
\]
so that any event assigned zero probability by the agent is also assigned zero probability by the designer. Thus the designer and the agent may disagree about the likelihood of states, but a state regarded as possible by the designer is not ruled out by the agent. The space of admissible information experiments and the signal-generation rule conditional on the realized information capability, together with the transition kernels and public-outcome kernels, are common knowledge; the realized \(\pi_t\) remains private as a component of \(\tau_t\).

The within-period interaction follows the sequence below.

\begin{enumerate}
\item The state \(\theta_t\) and private type
\(\tau_t=(A_t,u_t,b_t,\pi_t)\) are realized, and the agent observes
\(\tau_t\).

\item The agent communicates with the mechanism, after which the mechanism
selects its current response, as specified formally in Section~3.

\item The agent observes its private signal \(s_t\) and then chooses a realized
action that is technologically feasible under its true execution capability
\(A_t\) and, once the mechanism is specified, satisfies any authorization
restriction imposed by the mechanism.

\item The realized action generates the public outcome and affects the realized
reward through publicly contractible variables. The action itself need not be
publicly observed and may also affect the transition to the next state and
private type.

\item The publicly observed current-period events are incorporated into the
public history for the next date.
\end{enumerate}

The state \(\theta_t\), private type \(\tau_t\), interim signal \(s_t\), and
realized action \(a_t\) may therefore remain private even though they affect
observable outcomes and continuation dynamics. We take \((h_t,\tau_t)\) to be
a sufficient private state for the agent at date \(t\): conditional on the
public history \(h_t\) and current private type \(\tau_t\), earlier private
types and signals contain no additional information relevant for the agent's
current payoff and feasible actions, its beliefs and information technology,
or its continuation problem. Thus \(\tau_t\) summarizes all payoff- and
continuation-relevant private history at date \(t\), while the model may remain
fully history dependent through the public history \(h_t\).

This timing separates reporting from subsequent action choice: communication
occurs before \(s_t\) is realized, whereas the realized action is chosen
afterward. A mechanism must therefore induce both truthful communication and
the desired post-report action, giving rise to the honesty-and-obedience
distinction developed in Section~3.

\subsection{Partial verification}

The final ingredient of the baseline environment is partial verification of the agent's private characteristics. The designer may obtain verifiable evidence about the agent's realized private type. Such evidence can be asymmetric: an agent may conceal an existing capability, information source, or other verifiable component of its type without being able to counterfeit an attribute it does not possess. We represent this asymmetry through \textbf{partial verification}. Throughout, \textbf{one-sided verification} refers to the directional feature of this structure in which some verifiable characteristics can be concealed or understated without being symmetrically counterfeited or overstated. The general formulation below is represented by a verification preorder; the scalar withholding structure introduced in Section~4 is a particular one-sided specialization. This formulation builds on Green and Laffont \citeyearpar{GreenLaffont1986}, one-sided endowment-reporting mechanisms \citep{Tian1989}, withholding mechanisms \citep{Tian1993}, the evidence literature \citep{BenPorathDekelLipman2019}, and the application to AI agents in Bergemann, Koh, and Morris \citeyearpar{BergemannKohMorris2026}.

\begin{definition}[Verification preorder]
\label{def:verification_preorder}
For each date \(t\), a \textbf{verification preorder} \(\succeq_t\) is a reflexive and transitive relation on \(\mathcal T_t\). For any \(\tau_t,\widehat\tau_t\in\mathcal T_t\), we write \(\tau_t\succeq_t\widehat\tau_t\) when type \(\tau_t\) can substantiate every claim or pass every verification test available to type \(\widehat\tau_t\). The corresponding \textbf{claim set} is
\[
C_t(\tau_t)
=
\left\{
\widehat\tau_t\in\mathcal T_t:
\tau_t\succeq_t\widehat\tau_t
\right\}.
\]
\end{definition}

The verification preorder therefore determines which private types can imitate which others. A more capable agent may, for example, conceal capability through strategic underperformance, whereas an agent lacking that capability cannot necessarily counterfeit it. Thus \(C_t(\tau_t)\) collects the verification-feasible claims, or reportable counterparts, that type \(\tau_t\) can support. Reflexivity ensures that a type can always report itself, while transitivity captures the idea that evidence may be withheld but not fabricated.

We impose \textbf{capability consistency}: whenever \(\tau_t\succeq_t\widehat\tau_t\), the reported execution capability must satisfy \(\widehat A_t\subseteq A_t\). When information capability is also verifiable, the preorder may additionally require \(\pi_t\succeq_B\widehat\pi_t\), where \(\succeq_B\) denotes the Blackwell order. More generally, \(\succeq_t\) may restrict whichever components of the private type can be concealed but not counterfeited; behavioral objectives and subjective beliefs need not themselves be verifiable.

\section{Mechanisms, Incentives, and Revelation}

\subsection{Indirect mechanisms and attainable outcomes}

Section~2 specifies the economic environment and within-period timing. We now specify how a mechanism governs the communication and control stages of that interaction.

Fix the environment specified in Section~2. For each date \(t\), let \(M_t\) be the \textbf{ambient message space}, equipped with its maintained Borel structure, and let \(Y_t\) denote the space of mechanism responses. Let \(\mathcal D_t=2^{\mathcal A_t}\setminus\{\varnothing\}\) denote the collection of nonempty \textbf{authorization sets}, and let
\[
\mathcal R_t
:=
\left\{
R:Y_t\times\mathcal D_t\times\mathcal O_t\to\R
\right\}
\]
denote the space of period-\(t\) reward schedules. These spaces inherit the maintained measurable structures of Section~2.

Partial verification may restrict which messages in \(M_t\) are available to a given private type. Under an indirect mechanism, the public history \(h_t\) records past messages, realized mechanism responses and authorization sets, rewards, public outcomes, and any other publicly observed signals or control choices generated at earlier dates. The realized action is not a separate component of the public history unless it is revealed through the public outcome. Let \(\mathcal M_t:H_t\times\mathcal T_t\rightrightarrows M_t\) be the \textbf{feasible-message correspondence}, so that \(\mathcal M_t(h_t,\tau_t)\subseteq M_t\) is the nonempty set of messages available to type \(\tau_t\) after public history \(h_t\). We assume that
\[
\mathcal M_t(h_t,\widehat\tau_t)
\subseteq
\mathcal M_t(h_t,\tau_t)
\]
whenever \(\tau_t\succeq_t\widehat\tau_t\). Thus a type that can imitate another type can send every message available to the type it can imitate.

For each date \(t\), let
\[
\varphi_t:
H_t\times M_t
\longrightarrow
\Delta(Y_t\times\mathcal D_t)\times\mathcal R_t
\]
be the \textbf{outcome rule}. For each \((h_t,m_t)\), write
\[
\varphi_t(h_t,m_t)
=
\left(
\lambda_t(\cdot\mid h_t,m_t),
R_t(h_t,m_t,\cdot,\cdot,\cdot)
\right),
\]
where \(\lambda_t(\cdot\mid h_t,m_t)\) is the distribution of the mechanism response and authorization set, and \(R_t(h_t,m_t,\cdot,\cdot,\cdot)\in\mathcal R_t\) is the reward schedule. Thus the reward may condition on the realized mechanism response and authorization set as well as the public outcome:
\[
r_t=R_t(h_t,m_t,y_t,D_t,o_t).
\]
It does not condition directly on the realized action. When the public outcome reveals the action, this formulation nests action-contingent rewards through the outcome \(o_t\).

A standard mechanism specifies message spaces and an outcome rule. Here the additional feasible-message correspondence records the restrictions generated by partial verification. This parallels the role of constraint correspondences in generalized games and abstract economies, where feasible choices need not coincide with the ambient strategy set \citep{Debreu1952,ShaferSonnenschein1975}.

\begin{definition}[Indirect mechanism]
An \textbf{indirect mechanism} is a triple
\(\Gamma=(M,\mathcal M,\varphi)\), where
\(M=(M_t)_{t=1}^{T}\) is the profile of ambient message spaces,
\(\mathcal M=(\mathcal M_t)_{t=1}^{T}\) is the profile of feasible-message
correspondences, and \(\varphi=(\varphi_t)_{t=1}^{T}\) is the profile of
outcome rules. The unrestricted-message formulation is the special case
\(\mathcal M_t(h_t,\tau_t)=M_t\) for every \((h_t,\tau_t)\).
\end{definition}

Under an indirect mechanism, the abstract public history \(h_t\) introduced in Section~2 records past messages, realized mechanism responses and authorization sets, rewards, public outcomes, and any additional publicly observed signals or mechanism choices generated at earlier dates. Any information about past actions contained in the public history is therefore mediated through public outcomes or other publicly observed signals.

The mechanism governs the communication and control stages of the timing specified in Section~2. Under an indirect mechanism, the within-period interaction proceeds as follows.

\begin{enumerate}
\item After observing \((h_t,\tau_t)\), the agent chooses a feasible message
\[
m_t\in\mathcal M_t(h_t,\tau_t).
\]

\item Conditional on \((h_t,m_t)\), the mechanism draws its response and authorization set according to
\[
(\widetilde y_t,\widetilde D_t)
\sim
\lambda_t(\cdot\mid h_t,m_t).
\]

\item After observing the realized response \(y_t\) and authorization set \(D_t\), the agent observes its private signal \(s_t\) and chooses
\[
a_t\in A_t\cap D_t.
\]

\item The realized action generates the public outcome according to
\[
\widetilde o_t
\sim
Q_t(\cdot\mid\theta_t,a_t),
\]
and the realized reward is
\[
r_t=R_t(h_t,m_t,y_t,D_t,o_t).
\]
The action itself need not be publicly observed.

\item The publicly observed current-period events are incorporated into the next public history and, if \(t<T\), the state and private type evolve according to \(K_{t+1}\), which may depend physically on the realized action even when that action is not publicly observed.
\end{enumerate}

The distinction between \(A_t\) and \(D_t\) separates the agent's \textbf{execution capability} from the mechanism's \textbf{authorization}. The set \(A_t\), which is a component of the private type, specifies the actions the agent is capable of taking. The set \(D_t\) specifies the actions permitted by the mechanism and may depend on past communication, realized mechanism choices, rewards, public outcomes, and other publicly observed information.

A pure strategy for the agent under \(\Gamma\) is a pair \((\sigma,\beta)\), where
\[
\sigma_t(\tau_t,h_t)\in\mathcal M_t(h_t,\tau_t)
\]
specifies the date-\(t\) message and
\[
\beta_t(\tau_t,h_t,m_t,y_t,D_t,s_t)\in A_t\cap D_t
\]
specifies the subsequent action. 

An \textbf{equilibrium} is a strategy \((\sigma^\ast,\beta^\ast)\) that is sequentially optimal after every feasible \((h_t,\tau_t)\) under the agent's subjective beliefs, true information capability, transition process, and continuation mechanism.

Let
\[
\mathcal Z
:=
\prod_{t=1}^{T}
\left(
\Theta_t\times\mathcal T_t\times\mathcal A_t\times
\mathcal O_t\times\mathbb R
\right)
\]
denote the path space of payoff-relevant states, private types, actions, public outcomes, and rewards. For an indirect mechanism \(\Gamma\), let \(\mathcal O(\Gamma)\subseteq\Delta(\mathcal Z)\) denote the set of joint laws induced by its equilibria, and let
\[
\mathcal O:=\bigcup_{\Gamma}\mathcal O(\Gamma)
\]
denote the set of outcome laws attainable through indirect mechanisms. Thus attainability is defined in the sense of partial implementation: an outcome law is attainable if it is induced by an equilibrium of some indirect mechanism.

\subsection{Direct mechanisms and incentive compatibility}

A direct mechanism replaces the agent's arbitrary message with a report of its current private type. Directness concerns the communication stage only: after reporting, the agent still receives its private signal and chooses its realized action. The mechanism must therefore govern not only reporting but also authorization and the action it recommends after each private signal.

Thus the ambient message space at date \(t\) is \(\mathcal T_t\), while partial verification restricts a true type \(\tau_t\) to the feasible report set \(C_t(\tau_t)\subseteq\mathcal T_t\). Equivalently, in the direct mechanism the feasible-message correspondence is
\[
\mathcal M_t(h_t,\tau_t)=C_t(\tau_t).
\]
Truthful reporting is feasible because the verification preorder is reflexive, so \(\tau_t\in C_t(\tau_t)\).

Because the agent receives an additional private signal after reporting and chooses an action afterward, a direct mechanism specifies both authorization and a signal-contingent recommended action. Let
\[
\mathcal L_t
:=
\left\{
(D,\alpha):
D\in\mathcal D_t,\;
\alpha:\mathcal S_t\to D\text{ measurable}
\right\}
\]
denote the space of \textbf{authorization--plan pairs}, where \(D\) is the authorization set and \(\alpha(s)\) is the recommended action following signal \(s\). Thus \(D\) is the mechanism's authorization choice, specifying the actions it permits, while \(\alpha\) selects a recommended action within that set. The agent's actual action remains constrained by its true execution capability, so after authorization \(D_t\) its feasible action set is \(A_t\cap D_t\). We endow \(\mathcal L_t\) with the maintained measurable structure required for the mechanism kernels below. Let
\[
\mathcal R_t^D
:=
\left\{
R:Y_t\times\mathcal L_t\times\mathcal O_t\to\R
\right\}
\]
denote the corresponding space of direct-mechanism reward schedules.

For each date \(t\), let
\[
\psi_t:
H_t\times\mathcal T_t
\longrightarrow
\Delta(Y_t\times\mathcal L_t)\times\mathcal R_t^D
\]
be the direct outcome rule. For each \((h_t,\widehat\tau_t)\), write
\[
\psi_t(h_t,\widehat\tau_t)
=
\left(
\lambda_t^D(\cdot\mid h_t,\widehat\tau_t),
R_t^D(h_t,\widehat\tau_t)
\right),
\]
where \(\lambda_t^D\) determines the joint distribution of the mechanism response and authorization--plan pair, while \(R_t^D(h_t,\widehat\tau_t)\in\mathcal R_t^D\) is the corresponding reward schedule. Thus rewards may depend on the realized mechanism response, authorization--plan pair, and public outcome, but not directly on an unobserved realized action.

\begin{definition}[Direct mechanism]
\label{def:direct_mechanism}
A \textbf{direct mechanism} is
\(\Gamma^D=(\mathcal T,\psi)\), where
\(\mathcal T=(\mathcal T_t)_{t=1}^{T}\) is the profile of direct message
spaces and \(\psi=(\psi_t)_{t=1}^{T}\) is the profile of direct outcome rules.
Because the direct message spaces are the type spaces, when these spaces are
understood we identify the direct mechanism with its outcome-rule profile
\(\psi\) and refer to \(\psi\) simply as the mechanism.
\end{definition}

\begin{definition}[Action feasibility]
\label{def:action_feasibility}
A direct mechanism \(\psi\) is \textbf{action feasible} if, for every date
\(t\), public history \(h_t\), reported type \(\widehat\tau_t\), and every
\((y_t,D_t,\alpha_t)\) in the support of
\(\lambda_t^D(\cdot\mid h_t,\widehat\tau_t)\),
\[
\alpha_t(\mathcal S_t)\subseteq\widehat A_t\cap D_t,
\]
where \(\alpha_t(\mathcal S_t)\) denotes the set of actions prescribed
by the recommendation plan across all private signals.
\end{definition}

Action feasibility requires every action prescribed by the
recommendation plan to be both executable by the reported type
and authorized by the mechanism; it does not require
\(D_t\subseteq\widehat A_t\). Thus, execution capability \(A_t\)
specifies what the agent can do, authorization \(D_t\) specifies
what the mechanism permits, and the recommendation plan
\(\alpha_t\) selects an action from the authorized set for each
private signal. The agent's actual action remains a choice from
\(A_t\cap D_t\), which may include actions beyond those recommended
by the mechanism. Whether the agent follows the recommendation
is determined by obedience, not action feasibility.

Under a feasible direct mechanism, the within-period interaction proceeds as follows.
\begin{enumerate}
\item The agent submits a feasible report
\(\widehat\tau_t\in C_t(\tau_t)\).

\item Conditional on \((h_t,\widehat\tau_t)\), the mechanism draws its response, authorization set, and recommended plan according to
\[
(\widetilde y_t,\widetilde D_t,\widetilde\alpha_t)
\sim
\lambda_t^D(\cdot\mid h_t,\widehat\tau_t).
\]

\item The agent then observes its true private signal
\(\widetilde s_t\sim\pi_t(\cdot\mid\theta_t)\) and chooses an action
\(\widetilde a_t\in A_t\cap\widetilde D_t\).

\item The public outcome is generated according to
\[
\widetilde o_t
\sim
Q_t(\cdot\mid\widetilde\theta_t,\widetilde a_t),
\]
and the corresponding reward is
\[
\widetilde r_t
=
R_t^D(h_t,\widehat\tau_t)
\left(
\widetilde y_t,
(\widetilde D_t,\widetilde\alpha_t),
\widetilde o_t
\right).
\]

\item The publicly observed current-period events are incorporated into the next public history, and the continuation proceeds according to the dynamic environment of Section~2.
\end{enumerate}

The recommendation \(\widetilde\alpha_t\) does not replace the agent's action choice: the realized action remains privately chosen from \(A_t\cap\widetilde D_t\). It affects the reward through observable consequences whenever the reward is outcome contingent, but the reward schedule does not condition directly on a latent action.

We next specify the agent's reporting and post-report action strategies.

\begin{definition}[Admissible dynamic report--action strategy]
\label{def:admissible_strategy}
An \textbf{admissible dynamic report--action strategy} is a pair
\((\sigma,d)\) of measurable strategy profiles such that, at every
date \(t\) and every realization of their arguments,
\[
\sigma_t(\tau_t,h_t)\in C_t(\tau_t)
\]
and
\[
d_t(\tau_t,h_t,\widehat\tau_t,y_t,D_t,\alpha_t,s_t)
\in A_t\cap D_t.
\]
\end{definition}

The authorization set \(D_t\) specifies permitted actions,
whereas the recommendation plan \(\alpha_t\) assigns a single
recommended action \(\alpha_t(s_t)\) to each private signal \(s_t\).
The agent's actual action strategy \(d_t\) may select any action
in \(A_t\cap D_t\), not necessarily the recommended action.

Let \(\sigma^{\mathrm{tr}}\) denote \textbf{truthful reporting},
\[
\sigma_t^{\mathrm{tr}}(\tau_t,h_t)=\tau_t,
\]
and let \(d^{\mathrm{ob}}\) denote \textbf{obedience},
\[
d_t^{\mathrm{ob}}
(\tau_t,h_t,\widehat\tau_t,y_t,D_t,\alpha_t,s_t)
=
\alpha_t(s_t).
\]
The obedience strategy follows the recommendation after any feasible
report. Reflexivity of the verification preorder makes truthful
reporting admissible, while action feasibility and capability
consistency make the obedience strategy admissible after every
feasible report. Dynamic incentive compatibility requires obedience
following truthful reporting, but does not impose obedience to the
original recommendation following a false report. Such deviations
allow the true type to reoptimize its post-report action.

Fix a direct mechanism \(\psi\). For current type \(\tau_t\), public history
\(h_t\), and admissible continuation strategy \((\sigma,d)\), the agent's
expected discounted continuation payoff under its subjective beliefs is
\[
V_t(\sigma,d\mid\tau_t,h_t;\psi)
=
\int
\sum_{\ell=t}^{T}
\delta^{\ell-t}
\left[
\widetilde u_\ell(
\widetilde a_\ell,\widetilde\theta_\ell)
+
\widetilde r_\ell
\right]
\,d\mathbb P_t
(\cdot\mid\tau_t,h_t;\psi,\sigma,d).
\]

\begin{definition}[Dynamic incentive compatibility]
\label{def:dynamic_IC}
An action-feasible direct mechanism \(\psi\) is
\textbf{dynamically incentive compatible} if, for every date \(t\), every
feasible pair \((h_t,\tau_t)\), and every admissible continuation
report--action strategy \((\sigma,d)\),
\[
V_t
\left(
\sigma^{\mathrm{tr}},
d^{\mathrm{ob}}
\mid\tau_t,h_t;\psi
\right)
\geq
V_t
\left(
\sigma,d
\mid\tau_t,h_t;\psi
\right),
\]
and, following a truthful report, the recommended action is
sequentially optimal at every feasible post-report information
history, conditional on truthful-and-obedient continuation play.
\end{definition}

Dynamic incentive compatibility restricts the entire continuation
report--action strategy. It requires truthful reporting and obedient
post-report behavior along the truthful continuation, while ruling out
joint deviations that combine a false report, private action
reoptimization after new private information is realized, and
subsequent report or action deviations. Following a false report,
the agent may choose any feasible post-report action, and sequential
optimality does not require following the recommendation associated
with that report. A current action deviation can change the
distribution of the observable outcome and, through the physical
transition law, future states and private types even when the action
itself never enters the public record. Section~3.4 shows that, under
the maintained finite-horizon or continuity-at-infinity conditions,
dynamic incentive compatibility admits an equivalent recursive
representation in terms of honesty, obedience, and promise consistency.

An outcome law \(O\in\Delta(\mathcal Z)\) is \textbf{implementable}
if it is induced by truthful and obedient play under an action-feasible,
dynamically incentive-compatible direct mechanism. Let
\[
\mathcal O^{IC}\subseteq\Delta(\mathcal Z)
\]
denote the set of implementable outcome laws.

\subsection{Revelation principle}

The direct formulation is without loss for the outcome laws attainable in the
indirect environment, provided that feasible reports respect the verification
structure and deviations are allowed to combine reporting with subsequent
private action choices.

\begin{theorem}[Revelation principle under partial verification]
\label{prop:dynamic_revelation}
An outcome law is attainable in equilibrium through an indirect mechanism if
and only if it is induced by truthful and obedient play under an action-feasible,
dynamically incentive-compatible direct mechanism \(\psi\) with feasible report
correspondence \(C_t(\tau_t)\). That is,
\[
\mathcal O=\mathcal O^{IC}.
\]
\end{theorem}

\begin{proof}
Fix an indirect mechanism
\(\Gamma=(M,\mathcal M,\varphi)\) and an equilibrium
\((\sigma^I,\beta^I)\). We construct a direct mechanism \(\psi\) that
reproduces the same outcome law.

Because the direct mechanism records type reports whereas the indirect
mechanism records messages, write \(h_t^D\) for a direct public history and
\(h_t^I\) for its corresponding simulated indirect public history. We
construct \(h_t^I\) recursively from \(h_t^D\). At date \(1\),
\(h_1^I=h_1^D=\emptyset\). Suppose that a direct public history \(h_t^D\) has
been associated with the simulated indirect public history \(h_t^I\). After
report \(\widehat\tau_t\), define the simulated indirect message by
\[
m_t^I=\sigma_t^I(\widehat\tau_t,h_t^I).
\]
The direct mechanism then reproduces the distribution
\(\lambda_t(\cdot\mid h_t^I,m_t^I)\) over the mechanism response and
authorization set. For each realized \((y_t,D_t)\), it assigns the
signal-contingent plan
\[
\alpha_t(s)
=
\beta_t^I(
\widehat\tau_t,h_t^I,m_t^I,y_t,D_t,s),
\]
and sets
\[
R_t^D(
h_t^D,\widehat\tau_t,
y_t,(D_t,\alpha_t),o_t)
=
R_t(
h_t^I,m_t^I,y_t,D_t,o_t).
\]
The direct public history records the resulting report and realized public
mechanism choices, including the recommended plan, reward, and public outcome,
while the simulated indirect history is updated with the corresponding
indirect message and the same realized response, authorization set, reward,
and public outcome. The realized action is not inserted into either public
history unless revealed by the public outcome. This recursively defines the
history mapping at every continuation history.

Under truthful reporting and obedience, \(\psi\) therefore reproduces the same
conditional law of mechanism responses, authorization sets, actions, rewards,
public outcomes, states, and future private types as the original indirect
equilibrium. Although actions need not be publicly observed, the same realized
actions generate the same outcome distributions and enter the same physical
transition kernel. Action feasibility follows directly from the construction:
the plan \(\alpha_t\) reproduces the equilibrium action strategy of the reported
type and hence satisfies
\[
\alpha_t(s)\in\widehat A_t\cap D_t
\]
for every \(s\in\mathcal S_t\).

Now consider any admissible dynamic report--action deviation under \(\psi\).
Whenever a true type \(\tau_t\) reports
\(\widehat\tau_t\in C_t(\tau_t)\), the simulated message
\(\sigma_t^I(\widehat\tau_t,h_t^I)\) is feasible for
\(\widehat\tau_t\). Verification consistency gives
\[
\mathcal M_t(h_t^I,\widehat\tau_t)
\subseteq
\mathcal M_t(h_t^I,\tau_t),
\]
so the same message is feasible for the true type. After the induced mechanism
choice and realization of its true private signal, the true type can reproduce
the same private action deviation as under the direct mechanism. That action
induces the same distribution of public outcomes, rewards, future states,
private types, and continuation histories in the two mechanisms. Applying this
construction recursively therefore maps every admissible direct dynamic
deviation into an admissible deviation in the original indirect mechanism with
the same induced continuation payoff. Since
\((\sigma^I,\beta^I)\) is an equilibrium of the indirect mechanism, no such
deviation is profitable. Hence the constructed direct mechanism is dynamically
incentive compatible, and
\(\mathcal O\subseteq\mathcal O^{IC}\).

Conversely, any direct mechanism can be represented as an indirect mechanism
by taking type reports as messages, setting the feasible-message
correspondence of type \(\tau_t\) equal to \(C_t(\tau_t)\), and including the
recommended plan in the mechanism response. The indirect reward schedule then
reproduces the direct reward conditional on that augmented response,
authorization set, and public outcome. Dynamic incentive compatibility makes
truthful reporting and obedience sequentially optimal. Hence
\(\mathcal O^{IC}\subseteq\mathcal O\), completing the proof.
\end{proof}

Theorem~\ref{prop:dynamic_revelation} allows the analysis of dynamic
implementability to be restricted to direct mechanisms in which truthful
reporting is verification feasible and obedience to the recommended action plan
is optimal. The relevant incentive problem nevertheless remains dynamic:
current reports and privately chosen post-report actions affect the
distribution of public outcomes, continuation payoffs, future states, and
future deviation opportunities. We next give a recursive representation of
these reporting, obedience, and continuation incentives before introducing the
scalar structure used for envelope and optimal-design analysis.

\begin{remark}[Revelation and hidden action]
The revelation principle concerns private information, not observability of
the realized action. If the private type is common knowledge, the report space
can be taken to be a singleton, so truthful reporting and verification are
vacuous. The remaining implementation problem is then solely to induce
obedience of the privately chosen action, yielding the pure hidden-action
benchmark \(2B^{\mathrm{HA}}\).
\end{remark}

\subsection{Recursive characterization}

The preceding analysis defines dynamic incentive compatibility in terms of
complete report--action strategies after every history. The same
incentive problem can be represented recursively by identifying a sufficient
dynamic state and replacing arbitrary future deviations with current
report--action deviations linked through continuation values. This is
particularly important here because the agent first chooses a report and,
after observing the mechanism response and additional private information,
privately chooses an action whose observable consequences and physical effects
shape the continuation problem.

Let \(\chi_t:H_t\to Z_t\) be a measurable mapping and define
\(z_t=\chi_t(h_t)\). We call \(z_t\) a \textbf{sufficient public state} if
any two public histories mapped into the same \(z_t\) induce the same current
payoff-relevant environment, subjective continuation law, feasible
continuation choices, and continuation rule for every current private type.
Thus \(z_t\) may include public beliefs, monitoring statistics, physical state
variables, continuation promises, or other public variables needed to make
the continuation problem Markov. The full public history is always a valid
state by taking \(\chi_t\) to be the identity map.

For the designer's recursive problem below, we impose the corresponding
recursive sufficiency requirement: public histories mapped into the same
\(z_t\) must also induce the same feasible current mechanism choices and the
same admissible continuation-mechanism correspondence. Any continuation
promises or other state variables needed to preserve commitment are therefore
included in \(z_t\). This allows the agent's incentive problem and the
designer's continuation problem to be represented on the same public state.

We restrict attention in this subsection to direct mechanisms whose
continuation rules depend on \(h_t\) only through
\(z_t=\chi_t(h_t)\). After a feasible report
\(\widehat\tau_t\in C_t(\tau_t)\), the mechanism draws
\[
(\widetilde y_t,\widetilde D_t,\widetilde\alpha_t)
\sim\lambda_t^D(\cdot\mid z_t,\widehat\tau_t)
\]
and applies the reward schedule
\[
R_t^D(z_t,\widehat\tau_t,\cdot,\cdot,\cdot).
\]
The true type subsequently observes its private signal and chooses an action
from \(A_t\cap D_t\). Action feasibility requires the recommended plan to be
executable by the reported type, while a true type that can imitate that
report may have additional authorized actions available to it. The realized
action need not be public; its effects enter through the induced public
outcome and the physical transition law.

For a direct mechanism \(\psi\), define the agent's truthful-and-obedient
continuation value by
\[
V_t(\tau_t,z_t)
:=
V_t^\psi(
\sigma^{\mathrm{tr}},
d^{\mathrm{ob}}
\mid\tau_t,h_t),
\qquad
\chi_t(h_t)=z_t.
\]
Because \(z_t\) is a sufficient state and the continuation mechanism depends
on \(h_t\) only through \(z_t\), the truthful-and-obedient continuation value
is well defined as a function of \((\tau_t,z_t)\). For the one-period
recursion below, \(\mathbb P_t^{A,\psi}(\cdot\mid\cdot)\) denotes the
corresponding regular conditional marginal law generated by the current
mechanism rule and truthful-and-obedient continuation, with conditioning
variables specified at the relevant stage. Dynamic incentive compatibility
itself remains defined over complete public histories as in Section~3.2; we do
not require arbitrary off-path continuation strategies to depend on \(h_t\)
only through \(z_t\).

When \(T<\infty\), set \(V_{T+1}\equiv0\). We assume that, after every
feasible state, report, realized mechanism choice, and private signal, the
post-report action problem below attains its maximum and admits a measurable
maximizing action.

When \(T=\infty\), we additionally impose continuity at infinity: for every
date \(t\) and feasible state \((\tau_t,z_t)\),
\[
\lim_{n\to\infty}
\sup_{\substack{h_t:\,\chi_t(h_t)=z_t\\(\sigma,d)}}
\int
\sum_{\ell=t+n}^{\infty}
\delta^{\ell-t}
\left|
\widetilde u_\ell(
\widetilde a_\ell,\widetilde\theta_\ell)
+
\widetilde r_\ell
\right|
\,d\mathbb P_t(\cdot\mid\tau_t,h_t;\psi,\sigma,d)
=
0,
\]
where the supremum is over feasible public histories mapped into \(z_t\) and
admissible continuation report--action strategies. This condition makes
sufficiently distant continuation payoffs uniformly negligible and supports
the one-stage-deviation argument without restricting off-path strategies to
depend only on the sufficient public state. When \(T<\infty\), the
corresponding argument follows directly by backward induction.

For a true type \(\tau_t\), public state \(z_t\), current report
\(\widehat\tau_t\), realized mechanism choice
\((y_t,D_t,\alpha_t)\), and realized private signal \(s_t\), define the
\textbf{post-report continuation value}
\[
\begin{aligned}
&G_t(
\tau_t,z_t;
\widehat\tau_t,
y_t,D_t,\alpha_t,s_t)
\\
&\qquad
=
\max_{a\in A_t\cap D_t}
\int
\Big[
u_t(a,\widetilde\theta_t)
+
R_t^D(
z_t,\widehat\tau_t,
y_t,(D_t,\alpha_t),
\widetilde o_t)
+
\delta
V_{t+1}(
\widetilde\tau_{t+1},
\widetilde z_{t+1})
\Big]
\\
&\hspace{8em}
d\mathbb P_t
\bigl(
\cdot
\mid
\tau_t,z_t,
\widehat\tau_t,
y_t,D_t,\alpha_t,s_t,a
;\psi\bigr),
\end{aligned}
\]
where
\(\widetilde z_{t+1}=\chi_{t+1}(\widetilde h_{t+1})\).
The value \(G_t\) is evaluated at the agent's true type, including after a
false report; the report affects the mechanism response but not the type at
which the agent evaluates its continuation payoff. The action \(a\) enters
the conditional law because it changes the distribution of the public outcome
and may change the transition to future states and private types. It need not
enter either the reward schedule or the subsequent public history as a
separately observed variable.

The corresponding continuation value following report \(\widehat\tau_t\) is
\[
\widehat V_t(\tau_t,z_t;\widehat\tau_t)
=
\int
G_t(
\tau_t,z_t;
\widehat\tau_t,
\widetilde y_t,
\widetilde D_t,
\widetilde\alpha_t,
\widetilde s_t)
\,
d\mathbb P_t
\bigl(
\cdot
\mid
\tau_t,z_t,\widehat\tau_t
;\psi\bigr).
\]
This outer integral is over the mechanism's current randomization and the true
type's private signal; payoff-relevant uncertainty, the action-induced public
outcome, and the transition to the next continuation state are already
incorporated in \(G_t\). Thus
\(\widehat V_t(\tau_t,z_t;\widehat\tau_t)\) is the value of the true type
following report \(\widehat\tau_t\), after optimal feasible current action
reoptimization and truthful-and-obedient continuation thereafter.

The three objects \(G_t\), \(\widehat V_t\), and \(V_t\) separate the two
current incentive margins. The post-report problem determines the true type's
optimal private action following a given report and realized private signal,
while the reporting problem compares the resulting continuation values across
feasible reports. We call \((\tau_t,z_t)\) a \textbf{feasible state} if it can
arise after a public history, including a history generated by an
admissible deviation.

\begin{proposition}[Recursive characterization of dynamic incentive compatibility]
\label{prop:recursive_DIC}
An action-feasible mechanism \(\psi\) is dynamically incentive compatible if
and only if, at every state \((\tau_t,z_t)\), the following
conditions hold:
\begin{enumerate}
\item[(i)] \textbf{Recursive obedience.}
For every \((y_t,D_t,\alpha_t)\) in the support of
\(\lambda_t^D(\cdot\mid z_t,\tau_t)\) and every private signal \(s_t\),
\[
\begin{aligned}
\alpha_t(s_t)
\in
\arg\max_{a\in A_t\cap D_t}
\int
\Big[
&u_t(a,\widetilde\theta_t)
+
R_t^D(
z_t,\tau_t,
y_t,(D_t,\alpha_t),
\widetilde o_t)
\\
&+
\delta
V_{t+1}(
\widetilde\tau_{t+1},
\widetilde z_{t+1})
\Big]
\,
d\mathbb P_t
\bigl(
\cdot
\mid
\tau_t,z_t,\tau_t,
y_t,D_t,\alpha_t,s_t,a
;\psi\bigr).
\end{aligned}
\]

\item[(ii)] \textbf{Recursive honesty.}
The truthful report satisfies
\[
\tau_t
\in
\arg\max_{\widehat\tau_t\in C_t(\tau_t)}
\widehat V_t(\tau_t,z_t;\widehat\tau_t).
\]

\item[(iii)] \textbf{Promise consistency.}
The truthful continuation value satisfies
\[
V_t(\tau_t,z_t)
=
\widehat V_t(\tau_t,z_t;\tau_t).
\]
\end{enumerate}
\end{proposition}

\begin{proof}
Suppose first that the mechanism is dynamically incentive compatible. Fix a
feasible state \((\tau_t,z_t)\). Consider first deviations that keep the
current report truthful but allow the agent to change its current action after
observing the realized mechanism choice and private signal. Since such
action-only deviations are admissible, dynamic incentive compatibility
implies recursive obedience. In particular, a deviation need not be observed
directly to matter: it may change the distribution of the public outcome and
the continuation state.

Next, take any feasible report
\(\widehat\tau_t\in C_t(\tau_t)\). Consider the deviation that submits this
report and, conditional on each subsequently realized mechanism choice and
private signal, chooses the best feasible current action, before returning to
truthful and obedient play from date \(t+1\) onward. By construction, the
value of this deviation is
\(\widehat V_t(\tau_t,z_t;\widehat\tau_t)\). Dynamic incentive compatibility
therefore implies
\[
V_t(\tau_t,z_t)
\geq
\widehat V_t(\tau_t,z_t;\widehat\tau_t)
\qquad
\forall\widehat\tau_t\in C_t(\tau_t).
\]
For the truthful report, recursive obedience implies that the prescribed
action attains the post-report optimum. Hence
\[
V_t(\tau_t,z_t)=\widehat V_t(\tau_t,z_t;\tau_t),
\]
so recursive honesty and promise consistency follow.

Conversely, suppose conditions~(i)--(iii) hold after every state.
Consider any one-stage deviation consisting of a feasible current report and
an arbitrary feasible current action rule, followed by truthful and obedient
play thereafter. For every feasible report \(\widehat\tau_t\),
\(\widehat V_t(\tau_t,z_t;\widehat\tau_t)\) already allows the true type to
choose its optimal feasible current action after every realized mechanism
choice and private signal, taking into account the induced distribution of
observable outcomes and continuation states. Recursive honesty and promise
consistency therefore rule out every profitable deviation combining a false
report with current action reoptimization, while recursive obedience rules out
a profitable current action deviation following a truthful report.

Hence no profitable one-stage report--action deviation exists after any
feasible state. When \(T<\infty\), backward induction rules out every
admissible multistage deviation. When \(T=\infty\), continuity at infinity
yields the corresponding one-stage-deviation principle. The mechanism is
therefore dynamically incentive compatible.
\end{proof}

Recursive honesty and promise consistency yield the exact
\textbf{nested report--action condition}
\[
V_t(\tau_t,z_t)
\geq
\widehat V_t(\tau_t,z_t;\widehat\tau_t)
\qquad
\forall\widehat\tau_t\in C_t(\tau_t),
\]
with equality at the truthful report. Because \(\widehat V_t\) already
incorporates optimal current post-report action reoptimization, each
false-report comparison jointly evaluates the report and the best current
action available after that report. Recursive obedience remains separately
necessary because it imposes sequential optimality after each realized
mechanism choice and private signal on the truthful branch. Thus the nested
inequality summarizes the reporting comparisons, while
Proposition~\ref{prop:recursive_DIC} supplies the full sequential
characterization. Section~4 focuses on this structure for ordered scalar
types and derives its integral representation.

The recursion is stronger than a Bellman equation written only along the
truthful path. Every false report must be evaluated after allowing the true
type to solve its own post-report action problem. The formulation retains the
nested honesty-and-obedience structure emphasized by Bergemann, Koh, and
Morris \citeyearpar{BergemannKohMorris2026} and draws on the joint dynamic
reporting-and-action implementation approach of Tian
\citeyearpar{TianGeneralDynamic2026}, but here the two layers are
linked recursively because the current report determines the mechanism choice,
the true type then privately acts using its own capability and information,
and the resulting action changes the distribution of observable consequences
and may change the future state, private type, and continuation mechanism.

The representation also clarifies the relation to standard dynamic mechanism
design. In a direct-assignment environment the post-report action problem is
absent. The deviation value then reduces to the continuation payoff generated
by the allocation associated with the report, and recursive honesty becomes
the one-stage reporting comparison underlying the dynamic envelope and
integral monotonicity approach of Pavan, Segal, and Toikka
\citeyearpar{PavanSegalToikka2014}. The present environment adds endogenous
post-report action reoptimization to each dynamic reporting problem.

\begin{corollary}[Direct-assignment benchmark]
\label{cor:recursive_direct_assignment}
Under the conditions of Proposition~\ref{prop:recursive_DIC}, suppose the
mechanism directly assigns the utility-relevant action
\(a_t^D(z_t,\widehat\tau_t)\) after report \(\widehat\tau_t\), so the agent
has no post-report action choice. Then recursive obedience is vacuous, and
dynamic incentive compatibility is equivalent to
\[
\begin{aligned}
V_t(\tau_t,z_t)
=
\max_{\widehat\tau_t\in C_t(\tau_t)}
\int
\Big[
&u_t(
a_t^D(z_t,\widehat\tau_t),
\widetilde\theta_t)
\\
&+
R_t^D\!\left(
z_t,\widehat\tau_t,
\widetilde y_t,
(\widetilde D_t,\widetilde\alpha_t),
\widetilde o_t
\right)
\\
&+
\delta
V_{t+1}(
\widetilde\tau_{t+1},
\widetilde z_{t+1})
\Big]
\,
d\mathbb P_t
\bigl(
\cdot
\mid
\tau_t,z_t,\widehat\tau_t
;\psi\bigr),
\end{aligned}
\]
with the maximum attained at \(\widehat\tau_t=\tau_t\).
\end{corollary}

\begin{proof}
Under direct assignment there is no post-report action choice to be induced.
Consequently, the inner maximization defining \(G_t\) disappears, and
\(\widehat V_t\) is the current-plus-continuation payoff generated by the
action assigned to the report, integrated under the induced conditional law.
The result follows from recursive honesty and promise consistency in
Proposition~\ref{prop:recursive_DIC}.
\end{proof}

With post-report action discretion, by contrast, the report and action
optimizations cannot generally be collapsed into a conventional reporting
problem. Each feasible report induces a potentially different distribution
over mechanism responses, authorization sets, recommended plans, rewards,
public outcomes, and continuation states. Conditional on the realized
mechanism objects and its true private signal, the true type solves its own
action problem. Under Proposition~\ref{prop:recursive_DIC}, the nested
recursion can therefore be written as
\[
\begin{aligned}
V_t(\tau_t,z_t)
=
\max_{\widehat\tau_t\in C_t(\tau_t)}
\int
\Bigg[
\max_{a\in A_t\cap\widetilde D_t}
\int
\Big[
&u_t(a,\widetilde\theta_t)
+
R_t^D\!\left(
z_t,\widehat\tau_t,\widetilde y_t,
(\widetilde D_t,\widetilde\alpha_t),
\widetilde o_t
\right)
\\
&\quad+
\delta V_{t+1}(
\widetilde\tau_{t+1},
\widetilde z_{t+1})
\Big]
\,d\mathbb P_t
\bigl(
\cdot\mid
\tau_t,z_t,\widehat\tau_t,
\widetilde y_t,\widetilde D_t,
\widetilde\alpha_t,\widetilde s_t,a
;\psi\bigr)
\Bigg]
\\
&\qquad
d\mathbb P_t
\bigl(
\cdot\mid
\tau_t,z_t,\widehat\tau_t
;\psi\bigr),
\end{aligned}
\]
where dynamic incentive compatibility requires the outer maximizer to be
truthful and, on the truthful branch, the prescribed action to attain the
inner maximum after every realized mechanism choice and private signal.

Post-report action discretion therefore changes the value of every feasible
report rather than merely adding an independent constraint to a conventional
reporting problem. Authorization, outcome-contingent current rewards, and
continuation promises can change the value of action reoptimization, but the
same instruments may also affect current payoffs and future states. Recursive
obedience is consequently a distinct \textbf{incentive margin}: its economic
effect is the shadow value of restricting the designer to mechanisms that
sustain the prescribed post-report behavior, not a separate additive monetary
cost.

The same recursive structure can be used to represent the designer's
continuation problem. Let \(\Psi_t(z_t)\) denote the correspondence of
admissible current direct-mechanism rules at public state \(z_t\). A current
rule belongs to \(\Psi_t(z_t)\) if it can be extended by continuation rules
such that action feasibility and the recursive incentive conditions of
Proposition~\ref{prop:recursive_DIC} hold after every state and all
additional restrictions maintained by the particular design problem are
satisfied. Such restrictions may include participation constraints,
promise-keeping requirements, bounds on rewards, technological feasibility,
or other exogenous constraints. We assume that the resulting admissible
continuation problem is well posed: the relevant discounted-payoff integrals
and the designer's supremum over admissible continuation mechanisms are
finite. Let \(J_t(z_t)\) denote that supremum. By sufficiency of \(z_t\) for the
designer's continuation problem, \(\Psi_t(z_t)\) and the law of the next
public state jointly determine the feasible continuation problem.

When \(T=\infty\), we additionally impose continuity at infinity for the
designer's payoff:
\[
\lim_{n\to\infty}
\sup_{\psi\ \mathrm{admissible}}
\int
\sum_{\ell=t+n}^{\infty}
\delta^{\ell-t}
\left|
v_\ell(
\widetilde a_\ell,
\widetilde\theta_\ell,
\widetilde o_\ell)
-
c_\ell(\widetilde r_\ell)
\right|
\,
d\mathbb P_t(\cdot\mid z_t;\psi)
=
0.
\]
For \(T<\infty\), no additional tail condition is required.

\begin{theorem}[Recursive characterization of the designer's continuation problem]
\label{thm:recursive_principal}
Under the sufficiency requirement on \(z_t\) stated above and the maintained
well-posedness conditions, the designer's history-dependent continuation
problem is equivalent to the recursive program
\[
J_t(z_t)
=
\sup_{\psi_t\in\Psi_t(z_t)}
\int
\left[
v_t(
\widetilde a_t,
\widetilde\theta_t,
\widetilde o_t)
-
c_t(\widetilde r_t)
+
\delta
J_{t+1}(\widetilde z_{t+1})
\right]
\,
d\mathbb P_t(\cdot\mid z_t;\psi_t).
\]
When \(T<\infty\), \(J_{T+1}\equiv0\).
\end{theorem}

\begin{proof}
Consider any admissible continuation mechanism from public state \(z_t\).
By Proposition~\ref{prop:recursive_DIC}, its current rule and continuation
values satisfy the recursive incentive conditions after every state.
By sufficiency of \(z_t\) for the designer's continuation problem, conditional
on \(z_t\), the current payoff-relevant environment, admissible current
mechanism choices, induced next-state law, and admissible
continuation-mechanism correspondence are independent of the particular
public history leading to \(z_t\). The principle of optimality therefore
yields the stated recursion.

Conversely, state-contingent current rules in \(\Psi_t(z_t)\), together with
admissible continuation rules attaining or approximating the corresponding
continuation values, can be concatenated across public states. By construction,
the additional restrictions defining admissibility are preserved, and
Proposition~\ref{prop:recursive_DIC} implies dynamic incentive compatibility
of the resulting direct mechanism. Hence the recursive and
history-dependent formulations generate the same attainable designer
continuation values.

For \(T<\infty\), the equivalence follows by backward induction. For
\(T=\infty\), the maintained continuity-at-infinity conditions permit the
finite-horizon recursive argument to pass to the limit.
\end{proof}

Theorem~\ref{thm:recursive_principal} also provides the recursive foundation
for Section~5: changes in monitoring and information use alter the distribution
of \(\widetilde z_{t+1}\) and thereby both the designer's continuation value
\(J_{t+1}\) and the agent's continuation value \(V_{t+1}\).

\section{Optimal Dynamic Mechanism Design with Hidden Information and Hidden Action}

Section~3 characterized dynamic implementability when the agent may possess
private information and, after communicating with the mechanism, privately
choose an action that need not itself be publicly observed. We now turn to
optimal dynamic mechanism design and study how dynamic hidden information,
verification-restricted reporting, and hidden-action implementation jointly
shape the designer's problem.

As noted in Remark~\ref{rem:nested_type_environments}, the general private type
may contain heterogeneity in execution capability, behavioral objectives,
beliefs, and information capability. In this section, we focus on an ordered
scalar quasilinear environment for the hidden-information dimension:
execution capability, beliefs, and information capability are held fixed,
while the agent's gross payoff is parameterized by a scalar private type,
\(u_t(\eta_t,q_t,\theta_t)\), where
\(\eta_t\in E_t=[\underline\eta_t,\overline\eta_t]\subseteq\R\),
\(q_t\in\mathcal Q_t\) is the payoff-relevant allocation or control, and
\(\theta_t\in\Theta_t\) remains the common payoff-relevant state. The
designer's gross payoff is \(v_t(q_t,\theta_t)\).

Focusing on a scalar hidden-information dimension does not eliminate the
action problem of Section~3. Throughout this section, the execution capability
\(A_t\) is held fixed and suppressed from the scalar type, but the realized
action \(a_t\) remains the agent's post-report choice and need not be publicly
observed. Thus even with a single scalar private type \(\eta_t\), the
environment can contain both hidden information about \(\eta_t\) and hidden
action through \(a_t\). When the payoff-relevant control is implemented
through the recommended action, \(q_t\) denotes the control implemented on the
truthful-and-obedient path; in the direct-action case one may identify
\(q_t=\alpha_t(s_t)\). Under a deviation, however, the true type privately
chooses its actual action from its feasible set, and that action affects the
distribution of observable outcomes and continuation states as in Section~3.

Suppressing public-history arguments that are fixed throughout the local
analysis, let \(q_t(\widehat\eta^t)\in\mathcal Q_t\) denote the
truthful-and-obedient allocation or control induced by report history
\(\widehat\eta^t\), and let \(p_t(\widehat\eta^t)\) denote the corresponding
expected transfer from the agent to the designer. Period payoffs on the
implemented path are
\[
U_t^A
=
u_t(\eta_t,q_t,\theta_t)-p_t,
\qquad
U_t^P
=
v_t(q_t,\theta_t)+p_t.
\]
In the monetary-transfer specification of Section~3, \(p_t\) is the negative
of the induced expected reward. Because rewards condition on publicly
observable mechanism choices and outcomes rather than directly on a latent
action, this reduced-form expected transfer does not presume that the realized
action \(a_t\) is publicly observed. The linear screening specification
\(u_t(\eta_t,q_t,\theta_t)=\eta_t q_t\) is nested as a special case.

We impose a first-order Markov structure on the scalar type process. Let
\(F_1\) denote the distribution of \(\eta_1\), and let
\(F_{t+1}(\cdot\mid\eta_t)\) denote the conditional distribution of
\(\eta_{t+1}\) for \(t<T\). The transition law of the scalar type is exogenous
to the current mechanism choice.
Finally, specialize the verification preorder to the withholding structure
\[
C_t(\eta_t)=[\underline\eta_t,\eta_t].
\]
Thus a type \(\eta_t\) may report any
\(\widehat\eta_t\leq\eta_t\), but cannot report a higher type. This is the
scalar counterpart of the one-sided imitation structure in Section~2:
verifiable information may be concealed or understated but cannot be freely
counterfeited.\footnote{Bergemann, Koh, and Morris
\citeyearpar{BergemannKohMorris2026} use a one-sided verification structure to
capture the same concealment--counterfeiting asymmetry in a static AI-agent
environment. Related one-sided reporting restrictions appear much earlier in
Tian \citeyearpar{Tian1989}, where agents may understate but not overstate
their endowments to preserve feasibility, and in the explicit withholding
mechanism of Tian \citeyearpar{Tian1993}.}

The resulting environment therefore separates three conceptually distinct
margins. The first is \textbf{hidden information}: the evolving scalar type
\(\eta_t\) is privately observed. The second is \textbf{partial verification}:
conditional on hidden information, the verification technology restricts
which reports the true type can substantiate. The third is
\textbf{hidden-action implementation}: after reporting and receiving its
private signal, the agent chooses \(a_t\), which need not be publicly observed
and must therefore be induced through authorization, observable outcomes,
rewards, and continuation incentives. Partial verification modifies the
reporting problem associated with hidden information; hidden action creates a
separate moral-hazard problem.

This distinction generates a branched benchmark structure rather than a
single linear sequence. Let \(\mathrm{HI}\), \(\mathrm{HA}\), and
\(\mathrm{V}\) denote hidden information, hidden action, and one-sided
verification, respectively. The principal benchmark relations are

\begin{center}
\setlength{\unitlength}{1mm}
\begin{picture}(88,14)(0,-7)
\linethickness{0.4pt}

\put(4,0){\makebox(0,0){\(1B\)}}

\put(26,5.5){\makebox(0,0){\(2B^{\mathrm{HI}}\)}}
\put(62,5.5){\makebox(0,0){\(3B^{\mathrm{HI+V}}\)}}

\put(26,-5.5){\makebox(0,0){\(2B^{\mathrm{HA}}\)}}
\put(62,-5.5){\makebox(0,0){\(3B^{\mathrm{HI+HA}}\)}}

\put(84,0){\makebox(0,0){\(4B\)}}

\put(7,0.6){\vector(3,1){13}}
\put(33,5.5){\vector(1,0){20}}
\put(70,4.9){\vector(3,-1){11}}

\put(7,-0.6){\vector(3,-1){13}}
\put(33,-5.5){\vector(1,0){20}}
\put(70,-4.9){\vector(3,1){11}}

\end{picture}
\end{center}

Here \(1B\) is the complete-information benchmark with no
action-implementation friction;
\(2B^{\mathrm{HI}}\) is the conventional hidden-information problem with
directly assignable action;
\(2B^{\mathrm{HA}}\) is the complete-information hidden-action problem;
\(3B^{\mathrm{HI+V}}\) combines hidden information with one-sided
verification but removes the post-report action problem;
\(3B^{\mathrm{HI+HA}}\) combines hidden information and hidden action under
unrestricted reporting; and the full \(4B\) problem combines hidden
information, one-sided verification, and hidden action.
\footnote{The first- and second-best terminology follows the classical
distinction associated with Lipsey and Lancaster
\citeyearpar{LipseyLancaster1956}. We use the third- and fourth-best labels
descriptively to organize environments with multiple incentive and
implementation margins. The labels characterize benchmark constraint
configurations rather than imposing an unconditional welfare ranking across
different branches.}

This benchmark map also clarifies the relation to the existing literature.
Pavan, Segal, and Toikka \citeyearpar{PavanSegalToikka2014} provide the
central dynamic hidden-information benchmark: private information evolves
over time, but the payoff-relevant allocation is directly assigned, so there
is no separate hidden-action implementation problem. A large
principal--agent literature studies the joint presence of adverse selection
and moral hazard, including Laffont and Tirole
\citeyearpar{LaffontTirole1986}, Picard \citeyearpar{Picard1987},
Guesnerie, Picard, and Rey \citeyearpar{GuesneriePicardRey1989}, and
Meng and Tian \citeyearpar{MengTian2013}. Much of this literature is static
or specialized to particular contracting technologies. Dynamic exceptions
include Es\H{o} and Szentes \citeyearpar{EsoSzentes2017}, who study a
multiperiod contracting environment with a stochastically evolving hidden
type and hidden actions, as well as more specialized dynamic environments
combining adverse selection and moral hazard.

The present framework differs from both strands. It combines dynamically
evolving private information with a post-report action that remains privately
chosen after additional private information is realized, so that every
reporting deviation must be evaluated together with subsequent action
reoptimization and its effects on public outcomes and continuation states.
It further incorporates one-sided verification, authorization, and
history-dependent continuation incentives within the same recursive
mechanism. Thus, even before the verification restriction is imposed, the
\(3B^{\mathrm{HI+HA}}\) branch is a joint dynamic adverse-selection and
moral-hazard problem; one-sided verification adds the distinct reporting
margin that generates the full \(4B\) environment.

\subsection{Agent dynamic marginal utility and nested incentive compatibility}

For the differentiable analysis, fix a mechanism, contingent
continuation plan, and pre-type information state $j_t$.
Let $\{\epsilon_s\}_{s>t}$ be independent random variables,
each uniformly distributed on $(0,1)$ and independent of
$(\eta_t,j_t)$. Write
$\epsilon_{t+1}^{s}=(\epsilon_{t+1},\ldots,\epsilon_s)$.
The controlled scalar type process admits the representation
\[
\eta_s=Z_{t,s}(\eta_t,\epsilon_{t+1}^{s}),
\qquad s\ge t,\qquad Z_{t,t}(\eta_t)=\eta_t.
\]
The dependence of $Z_{t,s}$ on $j_t$, the mechanism,
and the continuation plan is suppressed. For the envelope
analysis, assume that $Z_{t,s}$ is differentiable in
$\eta_t$ almost everywhere, with integrably bounded derivatives.

We assume that the transition maps are differentiable in \(\eta_t\) almost
everywhere, that \(u_t\) is differentiable and equi-Lipschitz in its scalar
type, and that the resulting derivatives are integrably bounded and the
relevant continuation values are absolutely continuous. We also maintain
\textbf{scalar-envelope regularity}: holding the current report and prescribed
obedient action plan fixed, a local change in \(\eta_t\) affects continuation
utility only through its direct effect on gross utility and its induced effect
on future scalar types. For deviation values, the agent's post-report action
is an inner optimization variable, so its response to a local type perturbation
is handled by the envelope theorem under the regularity conditions stated
below.

Following Pavan, Segal, and Toikka
\citeyearpar{PavanSegalToikka2014}, define the \textbf{impulse response} of
\(\eta_s\) to \(\eta_t\) by
\[
I_{t,s}(\eta^s)
=
\int
\frac{\partial
Z_{t,s}(\eta_t,\varepsilon_{t+1}^{s})}
{\partial\eta_t}
\,\mathbb P
\bigl(d\varepsilon_{t+1}^{s}\mid \eta^s\bigr),
\qquad
I_{t,t}=1.
\]
It measures how a marginal change in the current private type propagates into
future private information. For example, if
\(\eta_{t+1}=\rho\eta_t+\varepsilon_{t+1}\), then
\(I_{t,s}=\rho^{\,s-t}\).

For a current report \(\widehat\eta_t\), let
\(V_t^{A,\mathrm{ob}}(\eta^t;\widehat\eta_t)\) denote the agent's
continuation utility when it reports \(\widehat\eta_t\), follows the
prescribed current action and continuation allocation, and thereafter reports
truthfully and obeys. The superscript ``\(\mathrm{ob}\)'' refers to
obedience, not observability of the realized action. Differentiating with
respect to the true current type while holding the current report fixed gives
the \textbf{agent's dynamic marginal continuation utility},
\[
\begin{aligned}
\frac{\partial
V_t^{A,\mathrm{ob}}(\eta^t;\widehat\eta_t)}
{\partial\eta_t}
=
\int
\sum_{s=t}^{T}
\delta^{s-t}
I_{t,s}(\widetilde\eta^s)
u_{s,\eta}
\left(
\widetilde\eta_s,
q_s(
\eta^{t-1},
\widehat\eta_t,
\widetilde\eta_{t+1}^{s}),
\widetilde\theta_s
\right)
\,
d\mathbb P_t
(\cdot\mid\eta^t,\widehat\eta_t;\psi,\mathrm{ob}).
\end{aligned}
\]
Here \(u_{s,\eta}\) is the agent's period-\(s\) marginal gross utility with
respect to its scalar private type. The expression measures the discounted
marginal effect of the current private type on continuation utility, including
its propagation into future private types through the impulse responses.

For the scalar environment considered here, write
\(V_t^A(\eta^t)\) for the truthful-and-obedient continuation value and
\(\widehat V_t^A(\eta^t;\widehat\eta_t)\) for the corresponding deviation
value after a feasible report
\(\widehat\eta_t\in C_t(\eta_t)\). As in Section~3,
\(\widehat V_t^A\) allows the true type to solve its own post-report action
problem after the report. Hence it already incorporates the effect of a
privately chosen action on observable outcomes and continuation states. The
derivative
\[
\frac{\partial
\widehat V_t^A(\eta^t;\widehat\eta_t)}
{\partial\eta_t}
\]
is the \textbf{agent's deviation-adjusted dynamic marginal continuation
utility}.

\begin{theorem}[Nested verification-restricted incentive compatibility]
\label{thm:verification_DIC}
Under the conditions of Proposition~\ref{prop:recursive_DIC},
an action-feasible mechanism \(\psi\) is dynamically incentive
compatible if and only if, after every public history
and for every current type, the following conditions hold:
\begin{enumerate}
\item[(i)] \textbf{Reporting incentive compatibility.}
\begin{equation}
\label{eq:nested_IC}
V_t^A(\eta^t)
\geq
\widehat V_t^A(\eta^t;\widehat\eta_t)
\qquad
\text{for every }
\widehat\eta_t\in C_t(\eta_t).
\end{equation}

\item[(ii)] \textbf{Post-report obedience.}
Following a truthful report, the prescribed action is optimal
at every feasible post-report information set, as characterized
by the recursive obedience condition of
Proposition~\ref{prop:recursive_DIC}.
\end{enumerate}
\end{theorem}

\begin{proof}
The result follows from Proposition~\ref{prop:recursive_DIC}.
The reporting condition rules out profitable false reports,
allowing the true type to reoptimize its post-report action.
The obedience condition ensures that the prescribed action is
sequentially optimal following truthful reporting. Together,
these conditions are equivalent to dynamic incentive compatibility
under the maintained one-stage-deviation assumptions.
\end{proof}

Theorem~\ref{thm:verification_DIC} characterizes dynamic incentive
compatibility through two complementary conditions. Reporting incentive
compatibility rules out profitable false reports, allowing the true type
to reoptimize its post-report action, while \(C_t(\eta_t)\) determines
which reports are feasible. Post-report obedience separately requires
the recommended action to be optimal following truthful reporting,
even when the realized action is not publicly observed. Together,
these conditions govern dynamic adverse selection and moral hazard
within the same recursive mechanism.

Bergemann, Koh, and Morris
\citeyearpar{BergemannKohMorris2026} provide the closest static benchmark for
the interaction between reporting and post-report action incentives. Unlike
their finite-space characterization through nested cycle conditions, the
present environment is dynamic, with evolving private information and
continuation incentives. Pavan, Segal, and Toikka
\citeyearpar{PavanSegalToikka2014} provide the corresponding dynamic benchmark
for hidden information under direct allocation, while Tian
\citeyearpar{TianGeneralDynamic2026} treats joint reporting and hidden-action
incentives in a broader dynamic environment. The present model adapts these
foundations to its AI-agent timing and partial-verification structure and permits the true type, after any report, to
privately reoptimize an action whose consequences affect public outcomes and
future states. It therefore combines dynamic hidden information and hidden
action before partial verification is imposed. The ordered differentiable
type structure permits these combined incentives to be characterized through
dynamic marginal continuation utilities and the integral conditions developed
below.

Fix a feasible report \(\widehat\eta_t\leq\eta_t\). Writing
\(\eta^t=(\eta^{t-1},\eta_t)\), hold the realized past type history
\(\eta^{t-1}\) fixed and vary only the current type. Because the withholding
correspondence is downward closed,
\(\widehat\eta_t\in C_t(x)\) for every
\(x\in[\widehat\eta_t,\eta_t]\). Absolute continuity therefore gives
\begin{align}
&
V_t^A(\eta^t)
-
\widehat V_t^A(\eta^t;\widehat\eta_t)
\nonumber
\\
&\quad=
V_t^A(\eta^{t-1},\widehat\eta_t)
-
\widehat V_t^A(
(\eta^{t-1},\widehat\eta_t);
\widehat\eta_t)
\nonumber
\\
&\qquad\quad+
\int_{\widehat\eta_t}^{\eta_t}
\left[
\frac{\partial V_t^A(\eta^{t-1},x)}
{\partial x}
-
\frac{\partial
\widehat V_t^A((\eta^{t-1},x);\widehat\eta_t)}
{\partial x}
\right]dx .
\label{eq:nested_gap_identity}
\end{align}
Hence the exact differentiable form of \eqref{eq:nested_IC} is that the
right-hand side of \eqref{eq:nested_gap_identity} be nonnegative for every
feasible report. The first two terms constitute the truthful-report
action-implementation gap---the obedience or moral-hazard margin when the
action is hidden---while the integral captures the false-report component of
the joint report--action incentive comparison.

A convenient sufficient condition takes a monotonicity form.

\begin{proposition}[Nested monotonicity]
\label{prop:nested_monotonicity}
Suppose the maintained regularity conditions hold and, after every
feasible history, the following conditions are satisfied:
\begin{enumerate}
\item[(i)] \textbf{Post-report obedience.}
The prescribed action is optimal following truthful reporting at
every feasible post-report information set.

\item[(ii)] \textbf{Nested monotonicity.}
For every \(r\in E_t\),
\[
\frac{\partial V_t^A(\eta^{t-1},x)}{\partial x}
\geq
\frac{\partial
\widehat V_t^A((\eta^{t-1},x);r)}
{\partial x}
\qquad
\text{for a.e. }x\geq r.
\]
\end{enumerate}
Then the mechanism is dynamically incentive compatible under
\(C_t(x)=[\underline\eta_t,x]\).
\end{proposition}

\begin{proof}
Post-report obedience implies
\[
V_t^A(\eta^{t-1},r)
=
\widehat V_t^A((\eta^{t-1},r);r)
\]
for every \(r\in E_t\), eliminating the truthful-report gap in
\eqref{eq:nested_gap_identity}. Nested monotonicity makes the
integrand nonnegative for every feasible false report. Hence
\eqref{eq:nested_IC} holds. Together, these conditions establish
dynamic incentive compatibility by
Theorem~\ref{thm:verification_DIC}.
\end{proof}

Proposition~\ref{prop:nested_monotonicity} is a tractable sufficient test, not
a replacement for the exact global condition: the integral in
\eqref{eq:nested_gap_identity} may remain nonnegative even if its integrand is
negative on part of the interval. It is nevertheless the natural analogue of
the familiar monotonicity test in classical screening and will be useful for
checking candidate solutions to the full optimization problem.

Verification also changes the usual envelope implication. Define the
\textbf{verification wedge}
\[
\kappa_t(\eta^t)
=
\frac{\partial V_t^A(\eta^t)}{\partial\eta_t}
-
\left.
\frac{\partial
V_t^{A,\mathrm{ob}}(\eta^t;\widehat\eta_t)}
{\partial\eta_t}
\right|_{\widehat\eta_t=\eta_t}.
\]
The first derivative moves along the truthful continuation-value schedule,
whereas the second holds the current report fixed before evaluating it at the
truthful report. Under unrestricted reporting, the ordinary dynamic-envelope
argument gives \(\kappa_t=0\) almost everywhere. Under the withholding
structure, however, truthful reporting is generally an upper-boundary choice
in the feasible report set, so dynamic incentive compatibility alone need not
imply the ordinary two-sided envelope equality. We call a mechanism
\textbf{envelope compatible} if
\(\kappa_t(\eta^t)=0\) for almost every \(\eta_t\) after every truthful
history.

\begin{proposition}[Sign of the verification wedge]
\label{prop:wedge_sign}
Suppose the mechanism is dynamically incentive compatible and satisfies the
maintained regularity and scalar-envelope conditions, and suppose
\(\widehat V_t^A(\eta^t;\widehat\eta_t)\) is left differentiable in
\(\widehat\eta_t\) at \(\widehat\eta_t=\eta_t\). Under
\(C_t(\eta_t)=[\underline\eta_t,\eta_t]\),
\(\kappa_t(\eta^t)\geq0\) for almost every \(\eta_t\) after every truthful
history.
\end{proposition}

\begin{proof}
The truthful-report condition in Theorem~\ref{thm:verification_DIC} gives
\(V_t^A(\eta^t)=\widehat V_t^A(\eta^t;\eta_t)\).
Differentiating this identity along the truthful-report locus yields
\[
\frac{\partial V_t^A(\eta^t)}{\partial\eta_t}
=
\left.
\frac{\partial
\widehat V_t^A(\eta^t;\widehat\eta_t)}
{\partial\eta_t}
\right|_{\widehat\eta_t=\eta_t}
+
\left.
\frac{\partial
\widehat V_t^A(\eta^t;\widehat\eta_t)}
{\partial\widehat\eta_t}
\right|_{\widehat\eta_t=\eta_t^-}.
\]
Sequential obedience together with scalar-envelope regularity implies, by the
envelope theorem for the inner post-report action problem,
\[
\left.
\frac{\partial
\widehat V_t^A(\eta^t;\widehat\eta_t)}
{\partial\eta_t}
\right|_{\widehat\eta_t=\eta_t}
=
\left.
\frac{\partial
V_t^{A,\mathrm{ob}}(\eta^t;\widehat\eta_t)}
{\partial\eta_t}
\right|_{\widehat\eta_t=\eta_t}.
\]
Hence
\[
\kappa_t(\eta^t)
=
\left.
\frac{\partial
\widehat V_t^A(\eta^t;\widehat\eta_t)}
{\partial\widehat\eta_t}
\right|_{\widehat\eta_t=\eta_t^-}.
\]
Because truthful reporting maximizes the deviation value over reports weakly
below \(\eta_t\), its left derivative at the upper boundary is nonnegative.
Therefore \(\kappa_t(\eta^t)\geq0\).
\end{proof}

Thus the verification wedge is not an additional action incentive constraint.
It is a rent consequence of the one-sided reporting boundary. A positive
wedge means that the truthful continuation-value schedule rises with the
private type faster than the agent's dynamic marginal continuation utility
holding the current report fixed. This reporting wedge is conceptually
distinct from the hidden-action implementation condition: the former concerns
which type reports can be imitated, while the latter concerns whether the
agent privately chooses the prescribed action after any given report.

\begin{proposition}[Local nonbinding verification]
\label{prop:local_envelope}
Suppose the mechanism is dynamically incentive compatible and the maintained
differentiability and scalar-envelope conditions hold. If, after a truthful
history and for almost every current type \(\eta_t\), the truthful report is
locally interior to the feasible report set, so that for some
\(\varepsilon>0\),
\[
(\eta_t-\varepsilon,\eta_t+\varepsilon)\cap E_t
\subseteq
C_t(\eta_t),
\]
then
\[
\frac{\partial V_t^A(\eta^t)}{\partial\eta_t}
=
\left.
\frac{\partial
V_t^{A,\mathrm{ob}}(\eta^t;\widehat\eta_t)}
{\partial\eta_t}
\right|_{\widehat\eta_t=\eta_t}
\]
at that history.
\end{proposition}

\begin{proof}
Local interiority restores the usual two-sided envelope argument for the
report choice, giving
\[
\frac{\partial V_t^A(\eta^t)}{\partial\eta_t}
=
\left.
\frac{\partial
\widehat V_t^A(\eta^t;\widehat\eta_t)}
{\partial\eta_t}
\right|_{\widehat\eta_t=\eta_t}.
\]
By the post-report obedience condition of
Theorem~\ref{thm:verification_DIC}, the prescribed action is
sequentially optimal following truthful reporting.
Under scalar-envelope regularity, the feasible
post-report action correspondence and the conditional law of other
payoff-relevant uncertainty do not move locally with the scalar perturbation.
The envelope theorem for the inner post-report action problem therefore gives
\[
\left.
\frac{\partial
\widehat V_t^A(\eta^t;\widehat\eta_t)}
{\partial\eta_t}
\right|_{\widehat\eta_t=\eta_t}
=
\left.
\frac{\partial
V_t^{A,\mathrm{ob}}(\eta^t;\widehat\eta_t)}
{\partial\eta_t}
\right|_{\widehat\eta_t=\eta_t}.
\]
Combining the two equalities gives the result.
\end{proof}

Thus local nonbinding verification restores the standard envelope implication.
The strict withholding structure generally places truth at the upper boundary,
so envelope compatibility remains an additional property of the baseline
one-sided model.

The direct-assignment case removes post-report action choice while retaining
verification-restricted reporting. This is stronger than merely observing the
action: direct assignment eliminates the inner action problem itself. More
generally, perfect observability together with sufficiently flexible
action-contingent contracting can also eliminate the moral-hazard friction,
but observability alone need not do so when contractual instruments are
restricted.

\begin{corollary}[Envelope-compatible direct-assignment reduction]
\label{cor:envelope_verification}
Under the assumptions of Theorem~\ref{thm:verification_DIC}, suppose the
utility-relevant allocation is directly assigned and the mechanism is envelope
compatible. Then the mechanism is dynamically incentive compatible if and
only if, after every history and for every
\(\widehat\eta_t\in C_t(\eta_t)\),
\[
\int_{\widehat\eta_t}^{\eta_t}
\left[
\left.
\frac{\partial
V_t^{A,\mathrm{ob}}((\eta^{t-1},x);r)}
{\partial x}
\right|_{r=x}
-
\frac{\partial
V_t^{A,\mathrm{ob}}((\eta^{t-1},x);\widehat\eta_t)}
{\partial x}
\right]dx
\geq0.
\]
\end{corollary}

\begin{proof}
Under direct assignment there is no post-report action reoptimization, so
\[
\widehat V_t^A(\eta^t;\widehat\eta_t)
=
V_t^{A,\mathrm{ob}}(\eta^t;\widehat\eta_t)
\]
for every feasible report. The truthful-report action gap in
\eqref{eq:nested_gap_identity} is therefore zero. Envelope compatibility gives
\[
\frac{\partial V_t^A(\eta^{t-1},x)}{\partial x}
=
\left.
\frac{\partial
V_t^{A,\mathrm{ob}}((\eta^{t-1},x);r)}
{\partial x}
\right|_{r=x}
\]
almost everywhere. Substitution into \eqref{eq:nested_gap_identity} gives the
stated condition.
\end{proof}

The hidden-information and hidden-action margins therefore enter a single
nested incentive problem but through different channels. Hidden information
creates report comparisons and continuation rents. Verification modifies the
set of reports against which truth must be compared. Hidden action changes the
value associated with every report because the true type can subsequently
reoptimize its privately chosen action. When \(a_t\) is not observed, the
latter is a genuine moral-hazard margin: incentives operate through observable
outcomes and continuation consequences rather than direct conditioning on the
realized action.

The conventional dynamic-screening characterization is obtained when
reporting is unrestricted, \(C_t(\eta_t)=E_t\), and the utility-relevant
allocation is directly assigned. These two reductions remove, respectively,
the one-sided verification structure and the post-report action problem. The
ordinary envelope condition then becomes
\[
\frac{\partial V_t^A(\eta^t)}{\partial\eta_t}
=
\left.
\frac{\partial
V_t^{A,\mathrm{ob}}(\eta^t;r)}
{\partial\eta_t}
\right|_{r=\eta_t}
\qquad
\text{for a.e. }\eta_t,
\]
and global incentive compatibility requires
\[
\int_{\widehat\eta_t}^{\eta_t}
\left[
\left.
\frac{\partial
V_t^{A,\mathrm{ob}}((\eta^{t-1},x);r)}
{\partial x}
\right|_{r=x}
-
\frac{\partial
V_t^{A,\mathrm{ob}}((\eta^{t-1},x);\widehat\eta_t)}
{\partial x}
\right]dx
\geq0
\]
for every \(\eta_t,\widehat\eta_t\in E_t\). This is the usual dynamic
integral monotonicity condition (\emph{Int-M}), the global counterpart of the
first-order incentive characterization (\emph{ICFOC}) in Pavan, Segal, and
Toikka \citeyearpar{PavanSegalToikka2014}. Their dynamic-screening benchmark
is therefore the \(2B^{\mathrm{HI}}\) branch of the present environment. The
general nested condition retains the same impulse-response and envelope
machinery while allowing both dynamic hidden information and hidden action,
with partial verification further modifying the feasible report domain.

The dynamic nature of the condition is important. A current report can alter
allocations along the continuation path, while a current hidden action can
alter the distribution of public outcomes and future states. The comparison
therefore concerns the entire continuation mechanism rather than a
single-period allocation. The impulse responses \(I_{t,s}\) summarize how the
effect of the current private type propagates to future private types and
thereby remains relevant for continuation allocations.

Finally, envelope compatibility recovers the usual dynamic
transfer-equivalence property. Integration gives
\[
V_t^A(\eta^t)
=
V_t^A(\eta^{t-1},\underline\eta_t)
+
\int_{\underline\eta_t}^{\eta_t}
\left.
\frac{\partial
V_t^{A,\mathrm{ob}}((\eta^{t-1},x);r)}
{\partial x}
\right|_{r=x}
\,dx.
\]
Conditional on envelope compatibility, the truthful-path allocation rule and
type process therefore determine truthful continuation utilities up to the
continuation utility of a reference type; quasilinearity then pins down the
expected net present value of transfers up to the corresponding constant.
Under one-sided verification without envelope compatibility, this conclusion
need not hold because \(\kappa_t\) remains an additional continuation-rent
object.

\subsection{Optimal design and verification-adjusted dynamic virtual surplus}

We now turn from the agent's incentive problem to the designer's optimization
problem. Participation is imposed at this stage: after observing \(\eta_1\),
the agent may reject the mechanism and obtain its normalized outside option of
zero. We use \textbf{technological feasibility} to refer to the exogenous
constraints on the allocation or control process imposed by physical,
resource, capacity, or other technological features of the environment. This
is distinct from action feasibility as defined in Section~3, which requires
the prescribed post-report action to be executable under the relevant
capability and authorization sets. Action feasibility alone does not implement
that action: when the agent privately chooses \(a_t\), obedience must be
sustained by post-report incentive compatibility.

For the optimal-contracting analysis, assume that the designer and agent
share a common prior over the payoff-relevant state, scalar type process,
and other uncertainty entering payoffs and transfers. Conditional
continuation beliefs are derived from this prior. This aligns the agent's
continuation values with the designer's transfer accounting; Sections~2--3
continue to allow heterogeneous beliefs. We maintain the scalar-envelope
regularity of Section~4.1.

The designer's full \(4B\) problem is therefore
\begin{equation}
\label{eq:full_4B_problem}
\begin{aligned}
\sup_{\psi}\quad
&
\Pi^P(\psi)
=
\int
\sum_{t=1}^{T}
\delta^{t-1}
\bigl(
v_t(q_t,\theta_t)+p_t
\bigr)
\,d\mathbb P(\cdot;\psi)
\\
\text{s.t.}\quad
&
V_1^A(\eta_1)\geq0
\qquad
\text{for every }\eta_1\in E_1,
\\
&
V_t^A(\eta^t)
\geq
\widehat V_t^A(\eta^t;\widehat\eta_t)
\qquad
\text{for every }
\widehat\eta_t\in C_t(\eta_t),
\\
&
\text{post-report obedience holds following truthful reports,}
\\
&
\psi\text{ is technologically and action feasible.}
\end{aligned}
\end{equation}
The constraints are participation, verification-feasible reporting
incentives, obedience after truthful reports, and technological and
action feasibility. The reporting comparisons allow optimal post-report
action reoptimization after any feasible false report; obedience separately
requires the recommended action to be sequentially optimal following a
truthful report. Under the baseline in which
\(a_t\) is not publicly observed, the latter is the dynamic moral-hazard
constraint. One-sided verification determines which hidden-information
deviations must be considered. These margins enter the optimization problem
jointly rather than as independently optimized stages.

Assume in addition that the initial distribution \(F_1\) admits a strictly
positive density \(f_1\) on
\(E_1=[\underline\eta_1,\overline\eta_1]\). Denote the
\textbf{hazard rate} and \textbf{inverse hazard rate} by
\begin{equation}
\label{eq:initial_hazard}
h_1(\eta_1)
:=
\frac{f_1(\eta_1)}{1-F_1(\eta_1)},
\qquad
\Lambda_1(\eta_1)
:=
\frac{1-F_1(\eta_1)}{f_1(\eta_1)}
=
\frac{1}{h_1(\eta_1)}.
\end{equation}

\begin{theorem}[Verification-adjusted dynamic virtual surplus]
\label{thm:verification_virtual_surplus}
Suppose the common-prior and scalar-envelope assumptions above hold, the
relevant expectations are integrable, and \(V_1^A\) is absolutely
continuous. Then the designer's payoff under any dynamically
incentive-compatible mechanism can be written as
\begin{equation}
\label{eq:verification_virtual_surplus}
\begin{aligned}
\Pi^P
=
\int
\Bigg[
\sum_{t=1}^{T}
\delta^{t-1}
\Big\{
&
u_t(\eta_t,q_t,\theta_t)
+
v_t(q_t,\theta_t)
\\
&
-
\Lambda_1(\eta_1)
I_{1,t}(\eta^t)
u_{t,\eta}(\eta_t,q_t,\theta_t)
\Big\}
-
\Lambda_1(\eta_1)\kappa_1(\eta_1)
\Bigg]
\,d\mathbb P(\cdot;\psi)
-
V_1^A(\underline\eta_1).
\end{aligned}
\end{equation}
\end{theorem}

\begin{proof}
Under the common prior, expected transfers paid to the designer equal the
agent's expected discounted gross utility minus its initial expected
continuation utility. Substitution into the designer's payoff gives
\[
\Pi^P
=
\int
\sum_{t=1}^{T}
\delta^{t-1}
\bigl(
u_t(\eta_t,q_t,\theta_t)
+
v_t(q_t,\theta_t)
\bigr)
\,d\mathbb P(\cdot;\psi)
-
\int_{E_1}
V_1^A(\eta_1)\,dF_1(\eta_1).
\]
By the scalar-envelope representation of Section~4.1 and the definition
of the verification wedge, for almost every initial type,
\[
\frac{dV_1^A(\eta_1)}{d\eta_1}
=
\int
\sum_{t=1}^{T}
\delta^{t-1}
I_{1,t}(\widetilde\eta^t)
u_{t,\eta}(\widetilde\eta_t,\widetilde q_t,\widetilde\theta_t)
\,d\mathbb P(\cdot\mid\eta_1;\psi)
+
\kappa_1(\eta_1).
\]
Absolute continuity and integration by parts therefore give
\[
\begin{aligned}
\int_{E_1}
V_1^A(\eta_1)\,dF_1(\eta_1)
=
V_1^A(\underline\eta_1)
+
\int
\Bigg[
\sum_{t=1}^{T}
\delta^{t-1}
\Lambda_1(\eta_1)
I_{1,t}(\eta^t)
u_{t,\eta}(\eta_t,q_t,\theta_t)
+
\Lambda_1(\eta_1)\kappa_1(\eta_1)
\Bigg]
\,d\mathbb P(\cdot;\psi),
\end{aligned}
\]
The common prior allows the conditional expectation in the envelope formula
to be integrated using the same joint law as the designer's objective.
Substitution gives \eqref{eq:verification_virtual_surplus}.
\end{proof}

The virtual-surplus identity separates rent accounting from action
implementation. The identity is written along the truthful-and-obedient path,
where \(q_t\) is the implemented control and \(p_t\) is the expected transfer
generated by observable outcomes. Hidden action therefore does not create a
new additive term in \eqref{eq:verification_virtual_surplus}. Instead, it
restricts which truthful-path controls and transfers can be implemented. This
distinction is central: hidden information enters the transformed objective
through continuation rents, while hidden action enters primarily through the
incentive-feasible mechanism set.

Motivated by Theorem~\ref{thm:verification_virtual_surplus}, define the
period-\(t\) \textbf{dynamic virtual surplus} by
\begin{equation}
\label{eq:dynamic_virtual_surplus}
\Phi_t(\eta^t,q_t,\theta_t)
=
u_t(\eta_t,q_t,\theta_t)
+
v_t(q_t,\theta_t)
-
\Lambda_1(\eta_1)
I_{1,t}(\eta^t)
u_{t,\eta}(\eta_t,q_t,\theta_t).
\end{equation}
The designer's payoff can then be written compactly as
\[
\Pi^P
=
\int
\left[
\sum_{t=1}^{T}
\delta^{t-1}
\Phi_t(\eta^t,q_t,\theta_t)
-
\Lambda_1(\eta_1)\kappa_1(\eta_1)
\right]
\,d\mathbb P(\cdot;\psi)
-
V_1^A(\underline\eta_1).
\]

This representation identifies three economically distinct components of the
transaction cost of dynamic control. We use \emph{transaction cost} in the
broad Coase--Williamson sense, rather than to mean a primitive monetary or
resource expenditure. The \textbf{informational transaction cost} is the
continuation rent required to elicit evolving hidden information. Its marginal
effect is captured inside dynamic virtual surplus by
\(\Lambda_1(\eta_1)I_{1,t}(\eta^t)u_{t,\eta}\): the inverse hazard rate
determines the initial rent weight, while the impulse response propagates that
rent effect to later dates.

The \textbf{verification transaction cost} is the additional continuation-rent
term associated with the one-sided reporting boundary. Because truthful
reporting may lie at the boundary of the feasible report set, the ordinary
two-sided envelope equality need not hold, and the resulting rent enters the
designer's objective through
\(\Lambda_1(\eta_1)\kappa_1(\eta_1)\).

The \textbf{implementation cost of hidden action}, or equivalently the
\textbf{obedience cost}, arises differently. The agent privately chooses its
post-report action and may reoptimize after observing information unavailable
when the report was made. When \(a_t\) is hidden, this is the standard
economic logic of moral hazard embedded in a dynamic mechanism-design problem:
the designer cannot contract directly on the latent action and must instead
induce it through authorization, observable outcomes, outcome-contingent
rewards, and continuation promises. Its general measure is therefore the
shadow cost of restricting the designer to mechanisms that satisfy the
post-report obedience requirement of Theorem~\ref{thm:verification_DIC}, rather than
a separate additive term in \eqref{eq:verification_virtual_surplus}.

The three components therefore enter the full problem through different
mathematical channels. Dynamic hidden information generates the standard
information-rent deduction from virtual surplus. One-sided verification can
generate an additional continuation-rent term through the verification wedge.
Hidden action restricts the set of implementable mechanisms through the
obedience condition. The first two appear in the rent representation; the
third is generally a shadow cost of implementation. They are jointly
determined by the same dynamic mechanism and should not be interpreted as
three additive monetary costs.

The decomposition also makes the two second-best branches transparent. In the
\(2B^{\mathrm{HI}}\) branch, the action is directly assigned, so the
hidden-action implementation condition disappears while the dynamic
information-rent term remains. In the \(2B^{\mathrm{HA}}\) branch, the scalar
type is common knowledge, so there is no screening rent or verification
problem, but the designer still faces the action-implementation constraint.
Thus eliminating hidden information does not restore the first best when the
action remains hidden; symmetrically, eliminating hidden action does not
restore the first best when private information remains.

Proposition~\ref{prop:wedge_sign} and initial participation imply the
virtual-surplus upper bound
\begin{equation}
\label{eq:virtual_surplus_upper_bound}
\Pi^P
\leq
\int
\sum_{t=1}^{T}
\delta^{t-1}
\Phi_t(\eta^t,q_t,\theta_t)
\,d\mathbb P(\cdot;\psi).
\end{equation}
This motivates the relaxed virtual-surplus candidate associated with the full
problem. Let \(\mathcal Q^F\) denote a relaxed class of technologically
feasible reduced-form allocation processes, including all implementable
truthful-and-obedient processes but not imposing incentive constraints.
Each candidate includes the physical law \(\mathbb P(\cdot;q)\) needed to
evaluate it, which coincides with \(\mathbb P(\cdot;\psi)\) whenever
that candidate is implemented by \(\psi\). Define
\begin{equation}
\label{eq:first_order_candidate}
q^\circ
\in
\arg\max_{q\in\mathcal Q^F}
\int
\sum_{t=1}^{T}
\delta^{t-1}
\Phi_t(\eta^t,q_t,\theta_t)
\,d\mathbb P(\cdot;q).
\end{equation}
When the virtual-surplus objective is concave and feasibility is historywise,
the corresponding first-order conditions characterize this candidate
pointwise; otherwise the relaxed problem must be solved subject to the
relevant technological constraints. Equation~\eqref{eq:first_order_candidate}
does not yet solve the full problem. It supplies the relaxed
hidden-information candidate that must be checked against the complete
reporting, action-implementation, and feasibility conditions.

The implementation check requires both the reporting comparisons and
truthful-report post-report obedience of Sections~3.4 and~4.1. Obedience
checks whether the recommended action is optimal at each relevant
post-report information set, given its observable-outcome and continuation
consequences; when the action is hidden, this is the moral-hazard condition.
For a feasible false report, the reporting comparison checks whether the true
type can profit by combining that report with optimal private action
reoptimization. This is the joint
hidden-information--hidden-action deviation. Verification determines which
false-report comparisons belong to \(C_t(\eta_t)\). The nested-monotonicity
condition of Proposition~\ref{prop:nested_monotonicity}, together with
truthful-report obedience, provides a convenient sufficient test; in the
unrestricted-report, direct-assignment reduction it
becomes the usual dynamic integral-monotonicity condition of Pavan, Segal, and
Toikka \citeyearpar{PavanSegalToikka2014}.

Envelope compatibility has a distinct role in this implementation check. It is
not another action incentive constraint: it determines whether the
verification-adjusted rent representation coincides with the ordinary
two-sided envelope representation. Thus a candidate may satisfy the full
nested incentive condition while generating \(\kappa_1>0\).

\begin{corollary}[Tightness of the virtual-surplus relaxation]
\label{cor:relaxed_optimality}
Let \(q^\circ\) solve \eqref{eq:first_order_candidate}. Suppose there exists a
technologically and action-feasible mechanism \(\psi^\circ\) implementing
\(q^\circ\) such that the nested incentive condition
\eqref{eq:nested_IC} holds after every history, post-report
obedience holds following truthful reporting at every relevant information
set, the mechanism is envelope compatible, and
\(V_1^A(\underline\eta_1)=0\). Then \(\psi^\circ\) solves the full \(4B\)
problem.
\end{corollary}

\begin{proof}
The bound \eqref{eq:virtual_surplus_upper_bound} applies to every feasible,
dynamically incentive-compatible, individually rational mechanism. The
allocation \(q^\circ\) maximizes this bound. Envelope compatibility sets
\(\kappa_1=0\), while binding participation of the lowest initial type sets
\(V_1^A(\underline\eta_1)=0\). Hence \(\psi^\circ\) attains the upper bound,
so no feasible mechanism can yield a larger designer payoff.
\end{proof}

Corollary~\ref{cor:relaxed_optimality} gives the solution procedure rather than
a derivation from a lower-best benchmark. One first formulates the full \(4B\)
problem, obtains its relaxed virtual-surplus candidate, and then verifies the
reporting incentives, post-report obedience, and the tightness of the rent
representation. If the hidden-action condition is nonbinding, the solution
collapses to the corresponding hidden-information branch. If the reporting
margin is absent while hidden action remains, the problem instead collapses
toward the hidden-action branch. The benchmark map is therefore inherently
branched rather than a single sequence.

\begin{corollary}[Hidden-information and hidden-action benchmark reductions]
\label{cor:best_reductions}
Under the maintained assumptions, the full problem admits the following
reductions:
\begin{enumerate}
\item[(i)]
If the scalar type is common knowledge and the action is directly assignable
or otherwise first-best implementable, the problem reduces to \(1B\).

\item[(ii)]
If \(\eta_t\) is private, reporting is unrestricted, and the payoff-relevant
action is directly assigned, the problem reduces to the conventional dynamic
hidden-information benchmark \(2B^{\mathrm{HI}}\).

\item[(iii)]
If \(\eta_t\) is common knowledge but \(a_t\) is privately chosen and hidden,
the problem reduces to the hidden-action benchmark \(2B^{\mathrm{HA}}\).

\item[(iv)]
If \(\eta_t\) is private and subject to one-sided verification but the
payoff-relevant action is directly assigned, the problem is
\(3B^{\mathrm{HI+V}}\).

\item[(v)]
If \(\eta_t\) is private, reporting is unrestricted, and \(a_t\) is privately
chosen and hidden, the problem is \(3B^{\mathrm{HI+HA}}\).

\item[(vi)]
If hidden information, one-sided verification, and hidden action are all
present, the problem is the full \(4B\).
\end{enumerate}
\end{corollary}

\begin{proof}
Each reduction follows by removing the corresponding informational or
action-implementation margin from the general nested problem. Direct
assignment eliminates post-report action reoptimization. Common knowledge of
the scalar type eliminates the reporting and information-rent problem and
makes verification irrelevant. Unrestricted reporting removes the one-sided
verification structure while retaining hidden information. Combining the
remaining margins yields the stated benchmark configurations.
\end{proof}

Corollary~\ref{cor:best_reductions} gives a compact taxonomy of the model. The
two second-best branches isolate the two classical incentive problems:
\(2B^{\mathrm{HI}}\) contains dynamic adverse selection without moral hazard,
whereas \(2B^{\mathrm{HA}}\) contains moral hazard without adverse selection.
The \(3B^{\mathrm{HI+HA}}\) branch is their joint dynamic problem even before
verification is introduced. The full \(4B\) environment then adds the
one-sided verification structure governing which hidden-information reports
are feasible.

The taxonomy should be distinguished from an unconditional welfare ranking
across branches. Removing a genuine action-implementation constraint weakly
expands the feasible mechanism set, so, holding the reporting technology fixed,
eliminating hidden action weakly raises the designer's attainable value.
Likewise, \(1B\) weakly dominates either second-best branch. One-sided
verification, however, changes the feasible deviation domain itself; it should
therefore be treated as a distinct reporting technology rather than
mechanically ordered as an additional resource constraint. The best labels
organize the constraint configurations; value comparisons across different
verification regimes require the corresponding feasible-set comparison.

\subsection{Distortion hierarchy and dynamics}

Section~4.2 shows that hidden information and hidden action enter the
designer's problem through different channels. Dynamic hidden information
changes the transformed objective through the propagated information-rent
term, whereas hidden action restricts the set of controls that can actually be
implemented. We now consider both margins jointly. This makes it possible to
characterize the allocation distortions associated with both branches of the
benchmark map rather than first solving a hidden-information-only problem.

Assume that technological feasibility is historywise separable and that the
scalar-type transition remains exogenous to the current mechanism choice.
Suppressing the payoff-relevant state in this environment, let \(q_t\geq0\)
denote the screening control and let \(\alpha_t\) denote the prescribed
post-report action. The agent's and designer's gross period payoffs on the
truthful-and-obedient path are
\[
u_t(\eta_t,q_t,\alpha_t)
=
\eta_tq_t-\frac{\gamma}{2}(\alpha_t-\bar a)^2,
\qquad
v_t(q_t,\alpha_t)
=
\beta q_t\alpha_t-\frac12q_t^2,
\]
where \(\gamma>0\), \(\beta>0\), and \(\bar a\in(0,1)\). The actual action
\(a_t\) remains privately chosen after the report. The mechanism cannot
condition directly on a latent action, but observable outcomes, current
rewards, and continuation promises may jointly generate a marginal incentive
for action. Let \(\iota_t\) denote this reduced-form marginal incentive and
suppose
\[
|\iota_t|\leq\bar\iota_t.
\]
Conditional on the report and the continuation mechanism, the post-report
action problem therefore takes the reduced form
\[
\max_{a_t}
\left\{
-\frac{\gamma}{2}(a_t-\bar a)^2-\iota_ta_t
\right\}.
\]
Implementing the recommendation \(\alpha_t\) requires
\begin{equation}
\label{eq:action_incentive_requirement}
\iota_t=-\gamma(\alpha_t-\bar a),
\end{equation}
so the available incentive technology imposes
\[
\gamma|\alpha_t-\bar a|\leq\bar\iota_t.
\]
This restriction is the hidden-action implementation margin of Section~4.2
made explicit.

Because \(u_{t,\eta}=q_t\), the dynamic virtual type remains the central
hidden-information object:
\begin{equation}
\label{eq:dynamic_virtual_type}
\eta_t^{\mathrm{vir}}(\eta^t)
:=
\eta_t
-
\Lambda_1(\eta_1)I_{1,t}(\eta^t).
\end{equation}
The truthful-path dynamic virtual surplus is therefore
\[
\Phi_t(\eta^t,q_t,\alpha_t)
=
\eta_t^{\mathrm{vir}}(\eta^t)q_t
+
\beta q_t\alpha_t
-
\frac12q_t^2
-
\frac{\gamma}{2}(\alpha_t-\bar a)^2.
\]
Thus hidden information replaces the current type by its dynamic virtual
counterpart, while hidden action restricts the set of implementable
\((q_t,\alpha_t)\) pairs.

To obtain closed-form controls, suppose
\[
\gamma>\beta^2
\]
and restrict attention to parameters for which the relevant solutions are
interior apart from the bound on marginal action incentives. For any scalar
coefficient \(x\), define
\begin{equation}
\label{eq:Q_H_definitions}
Q_t(x)
:=
\frac{\gamma(x+\beta\bar a)}
{\gamma-\beta^2},
\qquad
H_t(x)
:=
Q_t(x)
-
\frac{\beta}{\gamma}
\bigl[\beta Q_t(x)-\bar\iota_t\bigr]_+,
\end{equation}
where \([z]_+:=\max\{z,0\}\).

The interpretation is immediate. If action implementation is unrestricted,
the pointwise first-order conditions are
\begin{equation}
\label{eq:virtual_allocation_FOC}
q_t=x+\beta\alpha_t,
\qquad
\gamma(\alpha_t-\bar a)=\beta q_t,
\end{equation}
and hence \(q_t=Q_t(x)\). At this unconstrained solution, the required
marginal action incentive has magnitude \(\beta Q_t(x)\). If
\(\beta Q_t(x)>\bar\iota_t\), the desired action cannot be implemented.
The largest implementable recommendation is then
\[
\alpha_t
=
\bar a+\frac{\bar\iota_t}{\gamma},
\]
and reoptimization over \(q_t\) yields \(q_t=H_t(x)\). Hence \(Q_t\)
describes the allocation absent a hidden-action friction, whereas \(H_t\)
incorporates the obedience constraint.

For the reporting margin, maintain the regularity condition that
\(\eta_t^{\mathrm{vir}}(\eta^t)\) is nondecreasing in the current type,
holding the past history fixed. We also maintain the local benchmark assumption
that the current report affects the current control and incentive terms but
does not alter the continuation rule once the next scalar type is realized.
Because the post-report action problem is independent of the true scalar type
conditional on the report, choosing the payment intercept to satisfy the
ordinary envelope formula reduces the downward false-report comparison to
\[
\int_r^\eta
\bigl[q_t(x)-q_t(r)\bigr]\,dx
\geq0,
\qquad
r\leq\eta.
\]
Thus any nondecreasing control schedule satisfies the one-sided reporting
condition. Since both \(Q_t\) and \(H_t\) are increasing in their argument,
the benchmark controls below are monotone whenever
\(\eta_t^{\mathrm{vir}}\) is monotone. The same payment schedule therefore
implements unrestricted reporting when the conventional global
integral-monotonicity condition holds and remains feasible under one-sided
verification. In this benchmark the verification wedge is consequently zero.

\begin{proposition}[Dynamic distortion hierarchy with hidden action]
\label{prop:distortion_dynamics}
Suppose the preceding assumptions hold, the relevant reporting conditions are
satisfied, and \(I_{1,t}(\eta^t)\geq0\). Then:

\begin{enumerate}
\item[(i)]
The four primitive hidden-information and hidden-action benchmark controls are
\[
q_t^{1B}
=
Q_t(\eta_t),
\qquad
q_t^{2B^{\mathrm{HI}}}
=
Q_t(\eta_t^{\mathrm{vir}}),
\]
and
\[
q_t^{2B^{\mathrm{HA}}}
=
H_t(\eta_t),
\qquad
q_t^{3B^{\mathrm{HI+HA}}}
=
H_t(\eta_t^{\mathrm{vir}}).
\]
When one-sided verification is nonbinding as above,
\[
q_t^{3B^{\mathrm{HI+V}}}
=
q_t^{2B^{\mathrm{HI}}},
\qquad
q_t^{4B}
=
q_t^{3B^{\mathrm{HI+HA}}}.
\]

\item[(ii)]
The resulting pointwise allocation hierarchy is
\[
q_t^{1B}
\geq
q_t^{2B^{\mathrm{HI}}}
=
q_t^{3B^{\mathrm{HI+V}}}
\geq
q_t^{4B},
\]
and
\[
q_t^{1B}
\geq
q_t^{2B^{\mathrm{HA}}}
\geq
q_t^{3B^{\mathrm{HI+HA}}}
=
q_t^{4B}.
\]
Moreover,
\[
q_t^{2B^{\mathrm{HI}}}
\geq
q_t^{3B^{\mathrm{HI+HA}}}.
\]
No pointwise ranking between
\(q_t^{2B^{\mathrm{HI}}}\) and \(q_t^{2B^{\mathrm{HA}}}\) is implied.

\item[(iii)]
The hidden-information distortion is
\[
q_t^{1B}
-
q_t^{2B^{\mathrm{HI}}}
=
\frac{\gamma}{\gamma-\beta^2}
\Lambda_1(\eta_1)I_{1,t}(\eta^t),
\]
while the complete-information hidden-action distortion is
\[
q_t^{1B}
-
q_t^{2B^{\mathrm{HA}}}
=
\frac{\beta}{\gamma}
\bigl[
\beta q_t^{1B}-\bar\iota_t
\bigr]_+.
\]
Conditional on hidden information, the additional obedience distortion is
\[
q_t^{2B^{\mathrm{HI}}}
-
q_t^{3B^{\mathrm{HI+HA}}}
=
\frac{\beta}{\gamma}
\bigl[
\beta q_t^{2B^{\mathrm{HI}}}-\bar\iota_t
\bigr]_+.
\]
Hence, when verification is nonbinding,
\begin{equation}
\label{eq:full_distortion_decomposition}
q_t^{1B}
-
q_t^{4B}
=
\frac{\gamma}{\gamma-\beta^2}
\Lambda_1(\eta_1)I_{1,t}(\eta^t)
+
\frac{\beta}{\gamma}
\bigl[
\beta q_t^{2B^{\mathrm{HI}}}-\bar\iota_t
\bigr]_+.
\end{equation}

\item[(iv)]
Holding \(\eta_1\), \(\eta_t\), and the remaining primitives fixed, away from
the kink at
\(\beta Q_t(\eta_t^{\mathrm{vir}})=\bar\iota_t\),
\[
\frac{\partial q_t^{3B^{\mathrm{HI+HA}}}}
{\partial I_{1,t}}
=
-\Lambda_1(\eta_1)
\begin{cases}
\dfrac{\gamma}{\gamma-\beta^2},
&
\beta Q_t(\eta_t^{\mathrm{vir}})<\bar\iota_t,
\\[1.2ex]
1,
&
\beta Q_t(\eta_t^{\mathrm{vir}})>\bar\iota_t.
\end{cases}
\]
When verification is nonbinding, the same expression applies to
\(q_t^{4B}\).
\end{enumerate}
\end{proposition}

\begin{proof}
For any coefficient \(x\), the unconstrained pointwise objective is
\[
xq+\beta q\alpha-\frac12q^2
-\frac{\gamma}{2}(\alpha-\bar a)^2.
\]
Strict concavity follows from \(\gamma>\beta^2\). Its first-order conditions
are \eqref{eq:virtual_allocation_FOC}, whose unique solution is
\[
q=Q_t(x),
\qquad
\alpha
=
\bar a+\frac{\beta}{\gamma}Q_t(x).
\]
By \eqref{eq:action_incentive_requirement}, implementing this recommendation
requires marginal incentive
\[
|\iota_t|
=
\beta Q_t(x).
\]
If \(\beta Q_t(x)\leq\bar\iota_t\), the hidden-action constraint is
nonbinding and the control remains \(Q_t(x)\). If
\(\beta Q_t(x)>\bar\iota_t\), the largest implementable recommendation is
\[
\alpha
=
\bar a+\frac{\bar\iota_t}{\gamma}.
\]
Conditional on this recommendation, the first-order condition in \(q\) gives
\[
q
=
x+\beta\bar a+\frac{\beta\bar\iota_t}{\gamma}
=
Q_t(x)
-
\frac{\beta}{\gamma}
\bigl[\beta Q_t(x)-\bar\iota_t\bigr],
\]
which is \(H_t(x)\). This establishes the two control mappings in
\eqref{eq:Q_H_definitions}.

Setting \(x=\eta_t\) gives the complete-information benchmarks, while setting
\(x=\eta_t^{\mathrm{vir}}\) gives the hidden-information benchmarks. Under the
maintained reporting regularity, the resulting controls are nondecreasing in
the current type, so the reporting conditions are satisfied. One-sided
verification is then nonbinding and the verification-constrained branches
coincide with their unrestricted-report counterparts.

Since
\[
\eta_t-\eta_t^{\mathrm{vir}}
=
\Lambda_1(\eta_1)I_{1,t}(\eta^t)
\geq0,
\]
and both \(Q_t\) and \(H_t\) are increasing,
the allocation orderings in part~(ii) follow. Moreover,
\(H_t(x)\leq Q_t(x)\) for every \(x\), which gives
\(q_t^{2B^{\mathrm{HI}}}\geq
q_t^{3B^{\mathrm{HI+HA}}}\).

Because \(Q_t\) is affine with slope
\(\gamma/(\gamma-\beta^2)\),
\[
Q_t(\eta_t)-Q_t(\eta_t^{\mathrm{vir}})
=
\frac{\gamma}{\gamma-\beta^2}
\Lambda_1(\eta_1)I_{1,t}(\eta^t).
\]
Also,
\[
Q_t(x)-H_t(x)
=
\frac{\beta}{\gamma}
\bigl[\beta Q_t(x)-\bar\iota_t\bigr]_+.
\]
Evaluating this expression at \(x=\eta_t\) and
\(x=\eta_t^{\mathrm{vir}}\) gives the two hidden-action distortion formulas.
Adding
\[
q_t^{1B}-q_t^{4B}
=
\bigl(q_t^{1B}-q_t^{2B^{\mathrm{HI}}}\bigr)
+
\bigl(q_t^{2B^{\mathrm{HI}}}-q_t^{4B}\bigr)
\]
gives \eqref{eq:full_distortion_decomposition}.

Finally,
\[
\frac{\partial\eta_t^{\mathrm{vir}}}{\partial I_{1,t}}
=
-\Lambda_1(\eta_1).
\]
When the action constraint is slack,
\(H_t'=Q_t'=\gamma/(\gamma-\beta^2)\). When it binds,
\[
H_t'(x)
=
Q_t'(x)
\left(1-\frac{\beta^2}{\gamma}\right)
=
1.
\]
The derivative in part~(iv) follows.
\end{proof}

Proposition~\ref{prop:distortion_dynamics} gives an allocation counterpart to
the branched benchmark structure introduced at the beginning of Section~4.
Dynamic hidden information and hidden action generate distinct distortions.
The first replaces the actual type by the dynamic virtual type; the second
projects the corresponding relaxed control onto the set supportable by the
available action incentives. Their interaction is therefore sequential in the
closed-form representation but joint in the underlying mechanism: the
hidden-action distortion is evaluated at the control already adjusted for
dynamic information rents.

A useful special case makes the time profile explicit.

\begin{corollary}[AR(1) distortion decomposition]
\label{cor:type_linear_distortion}
Suppose
\[
\eta_{t+1}
=
\rho\eta_t+\varepsilon_{t+1},
\qquad
\rho\in(0,1),
\]
and \(\eta_1\) is uniformly distributed on \([0,1]\). Then
\[
I_{1,t}
=
\rho^{\,t-1},
\qquad
\Lambda_1(\eta_1)
=
1-\eta_1,
\]
and
\[
\eta_t^{\mathrm{vir}}
=
\eta_t-(1-\eta_1)\rho^{\,t-1}.
\]
The six benchmark controls are therefore
\[
\begin{aligned}
q_t^{1B}
&=
Q_t(\eta_t),
&
q_t^{2B^{\mathrm{HI}}}
&=
Q_t\!\left(
\eta_t-(1-\eta_1)\rho^{\,t-1}
\right),
&
q_t^{3B^{\mathrm{HI+V}}}
&=
q_t^{2B^{\mathrm{HI}}},
\\
q_t^{2B^{\mathrm{HA}}}
&=
H_t(\eta_t),
&
q_t^{3B^{\mathrm{HI+HA}}}
&=
H_t\!\left(
\eta_t-(1-\eta_1)\rho^{\,t-1}
\right),
&
q_t^{4B}
&=
q_t^{3B^{\mathrm{HI+HA}}}.
\end{aligned}
\]
In particular,
\[
q_t^{1B}-q_t^{2B^{\mathrm{HI}}}
=
\frac{\gamma}{\gamma-\beta^2}
(1-\eta_1)\rho^{\,t-1},
\]
\[
q_t^{1B}-q_t^{2B^{\mathrm{HA}}}
=
\frac{\beta}{\gamma}
\bigl[
\beta q_t^{1B}-\bar\iota_t
\bigr]_+,
\]
and
\[
q_t^{1B}-q_t^{4B}
=
\frac{\gamma}{\gamma-\beta^2}
(1-\eta_1)\rho^{\,t-1}
+
\frac{\beta}{\gamma}
\bigl[
\beta q_t^{2B^{\mathrm{HI}}}-\bar\iota_t
\bigr]_+.
\]
Thus the propagated hidden-information component decays geometrically with the
horizon, holding the current type and the remaining primitives fixed, whereas
the hidden-action component depends on whether the required action incentive
exceeds the available bound. The total \(4B\) distortion therefore need not
decline monotonically over calendar time.
\end{corollary}

\begin{proof}
Under the AR(1) transition,
\(I_{1,t}=\rho^{\,t-1}\), while uniformity of the initial type gives
\(\Lambda_1(\eta_1)=1-\eta_1\). Substitution into
\eqref{eq:dynamic_virtual_type} and
Proposition~\ref{prop:distortion_dynamics} yields the stated controls and
distortion formulas.
\end{proof}

The pure linear-screening benchmark is nested as the special case
\(\beta=0\). Then \(Q_t(x)=H_t(x)=x\), the hidden-action margin disappears,
and Corollary~\ref{cor:type_linear_distortion} reduces to
\[
q_t^{1B}-q_t^{2B^{\mathrm{HI}}}
=
\Lambda_1(\eta_1)I_{1,t}(\eta^t),
\]
which is the informational-distortion formula obtained from the conventional
dynamic-screening problem.

The pointwise ordering above is specific to the maintained
linear-quadratic implementation technology. Its purpose is to make the
branched benchmark map operational at the allocation level. In the general
environment of Section~4.2, the dynamic virtual type continues to summarize
the propagated hidden-information distortion, while hidden action operates
through the nested implementation constraint and one-sided verification may
generate a nonzero verification wedge. There is therefore no general
pointwise ordering of controls across different verification technologies.

At the level of designer values, comparisons remain governed by feasible-set
inclusion. Holding the reporting technology fixed, adding the hidden-action
implementation restriction weakly lowers attainable value, and introducing
private information into the complete-information hidden-action problem does
the same. In the present benchmark, verification is nonbinding, so
\[
\Pi^{3B^{\mathrm{HI+V}}}
=
\Pi^{2B^{\mathrm{HI}}},
\qquad
\Pi^{4B}
=
\Pi^{3B^{\mathrm{HI+HA}}}.
\]
No general ranking is implied between
\(\Pi^{2B^{\mathrm{HI}}}\) and
\(\Pi^{2B^{\mathrm{HA}}}\), since the two branches remove different primitive
incentive problems. The branch-specific value comparisons can be summarized
as follows:

\begin{center}
\setlength{\unitlength}{1mm}
\begin{picture}(88,18)(0,-9)
\linethickness{0.4pt}

\put(4,0){\makebox(0,0){\(\Pi^{1B}\)}}

\put(26,5.5){\makebox(0,0){\(\Pi^{2B^{\mathrm{HI}}}\)}}
\put(62,5.5){\makebox(0,0){\(\Pi^{3B^{\mathrm{HI+V}}}\)}}

\put(26,-5.5){\makebox(0,0){\(\Pi^{2B^{\mathrm{HA}}}\)}}
\put(62,-5.5){\makebox(0,0){\(\Pi^{3B^{\mathrm{HI+HA}}}\)}}

\put(84,0){\makebox(0,0){\(\Pi^{4B}\)}}

\put(8,0.6){\line(3,1){12}}
\put(15.5,4.5){\makebox(0,0){\(\geq\)}}

\put(34,5.5){\line(1,0){18}}
\put(43,8.1){\makebox(0,0){\(\geq\)}}

\put(71,4.9){\line(3,-1){10}}
\put(78.5,4.45){\makebox(0,0){\(\geq\)}}

\put(8,-0.6){\line(3,-1){12}}
\put(15.5,-4.5){\makebox(0,0){\(\geq\)}}

\put(34,-5.5){\line(1,0){18}}
\put(43,-8.1){\makebox(0,0){\(\geq\)}}

\put(71,-4.9){\line(3,1){10}}
\put(78.5,-4.45){\makebox(0,0){\(\geq\)}}

\end{picture}
\end{center}

The equalities that arise in the present linear-quadratic benchmark should not
be read as replacing the weak inequalities represented in the figure. The
figure records the general feasible-set ordering across benchmark
configurations. In this benchmark, the verification constraint is nonbinding,
so the corresponding upper- and lower-right inequalities collapse to
equalities. With a binding verification margin, those links can instead be
strict.

The subsection therefore yields a single distortion decomposition for the
full benchmark structure. Dynamic hidden information enters through
\(\eta_t^{\mathrm{vir}}\), hidden action enters through the implementability
mapping \(H_t\), and one-sided verification remains a distinct reporting
margin that is nonbinding in the present closed-form benchmark but need not be
so in the general model.

\section{Dynamic Monitoring and Information Design}

Monitoring generates evidence that can be used to condition future mechanism
choices. In AI applications, monitoring may include evaluations, audits, tests,
benchmark measurements, and other procedures that generate evidence about
otherwise imperfectly observed capabilities, behavior, or payoff-relevant
states. We use \textbf{monitoring} as the formal term for this information
technology. Monitoring is distinct from verification. Verification restricts
which claims or reports the agent can substantiate, whereas monitoring generates
evidence that may enter the public continuation state and influence future
authorization, rewards, deployment, and other mechanism choices.

Information has two conceptually distinct margins: its \textbf{production} and
its \textbf{use}. The baseline environment takes the technologies that generate
information as primitive. The agent's information capability \(\pi_t\)
determines the private signal it receives before choosing its action, while the
primitive monitoring technology determines the raw evidence available to the
mechanism after current behavior. The realized action \(a_t\) need not be
publicly observed, but it may affect the distribution of monitoring evidence.
In this section, these information-production technologies remain fixed. What
becomes endogenous is information use: the mechanism chooses how generated
monitoring evidence is processed, disclosed, and incorporated into future
beliefs, permissions, rewards, and other continuation decisions. Endogenizing
information production itself would introduce learning and experimentation and
is reserved for the broader system-level analysis, where information
production, information use, posterior evolution, incentives, and future
feasibility are jointly determined \citep{Tian2026}.

The analysis combines two established strands of theory. On the dynamic
mechanism-design side, continuation values summarize future incentives, while
current choices affect those values through the evolution of the state and
private information, as in \citet{PavanSegalToikka2014} and
\citet{Pavan2026}. On the information-design side, Bayesian persuasion studies
the choice of information subject to Bayes plausibility and receiver
optimality \citep{KamenicaGentzkow2011}, while \citet{BergemannMorris2019}
formulate information design through recommendation and obedience.
\citet{DovalEly2020} provide a sequential information-design benchmark in
which information and extensive-form structure interact over time.
\citet{BergemannHeumannMorris2026} integrate information design and mechanism
design within a common static screening framework, while
\citet{BergemannKohMorris2026} emphasize that mechanisms for AI agents must
support obedience jointly with truthful revelation of privately known
characteristics. Their evaluation structure concerns verification of type
claims; the monitoring technology introduced here is a distinct object that
generates post-action evidence for use in the continuation mechanism.

The present environment links dynamic mechanism design and information design
through the continuation mechanism. Information use changes not only what
becomes publicly known, but also the public state on which future
authorization, rewards, deployment, and other mechanism choices may depend. A
public monitoring message can therefore change the continuation payoff from a
current report or post-report action. When \(a_t\) is hidden, the primitive
monitoring technology may depend on that action even though the subsequent
information-use policy cannot condition on it directly. Information design and
incentive design are consequently connected through continuation values. We
first characterize which public information can be generated from a fixed
monitoring technology, then incorporate information use into the recursive
mechanism problem, and finally study monitoring capacity, selective disclosure,
and the conditions under which information use can be separated from current
incentive design.

\subsection{Monitoring Capacity and Feasible Information Use}

Let \(x_t\in\mathcal X_t\) denote the realization of the
\textbf{raw monitoring signal}, generated after the current action and public
outcome are realized. Write $
\xi_t:=(z_t,\widehat\tau_t,y_t,(D_t,\alpha_t),o_t)$
for the current publicly available mechanism context. Conditional on this
public context and the hidden current variables,
\[
\widetilde x_t
\sim
Q_t^{\mathrm{mon}}
\bigl(
\cdot
\mid
\xi_t,\theta_t,\tau_t,a_t
\bigr),
\]
where \(Q_t^{\mathrm{mon}}\) is a fixed and commonly known stochastic kernel
\[
Q_t^{\mathrm{mon}}:
Z_t\times\Theta_t\times\mathcal T_t\times\mathcal T_t\times
Y_t\times\mathcal L_t\times\mathcal A_t\times\mathcal O_t
\to\Delta(\mathcal X_t).
\]

Conditional on a fixed observable context
\(\xi_t\), the kernel \(Q_t^{\mathrm{mon}}\) is a statistical experiment
about the hidden current variables \((\theta_t,\tau_t,a_t)\). We use
\textbf{monitoring capacity} to refer to the informativeness of this primitive
experiment. Thus monitoring may reveal information about the payoff-relevant
state, the agent's private characteristics, its privately chosen action, or
some combination of these objects. In particular applications, the monitored
object may be a specific capability, behavioral property, or performance
measure. Fixing the monitoring technology does not require the monitoring
signal to be independent of behavior. Reports, mechanism choices, actions,
and outcomes may affect the distribution of the evidence generated. What is
fixed is the kernel that maps the current environment into that evidence.

Let \(\mathcal E_t^{\mathrm{pub}}\) denote the measurable space of
\textbf{public monitoring messages}. An \textbf{information-use policy} is a
stochastic kernel \(\zeta_t(\cdot\mid\xi_t,x_t)\) from the current observable
context and raw monitoring signal into
\(\Delta(\mathcal E_t^{\mathrm{pub}})\). After observing \(x_t\), the
mechanism generates $
\widetilde e_t
\sim
\zeta_t
\bigl(
\cdot
\mid
\xi_t,x_t
\bigr).$

Thus \(Q_t^{\mathrm{mon}}\) determines what monitoring evidence is available,
whereas \(\zeta_t\) determines how that evidence is used publicly. The
primitive monitoring kernel may depend on the hidden variables
\((\theta_t,\tau_t,a_t)\), since they affect the distribution of the evidence.
The information-use policy cannot condition directly on those hidden
variables; it conditions only on information available to the mechanism,
including the realized raw monitoring signal. In the observable-action special
case, \(a_t\) can instead be included in the public context.

Because \(e_t\) is public, it enters the public history from date \(t+1\)
onward. We append it to the publicly observed date-\(t\) record and continue
to denote the resulting history by \(h_{t+1}\in H_{t+1}\). Under the public
history convention of Section~3, \(h_{t+1}\) therefore records the date-\(t\)
communication, realized mechanism choices, reward, public outcome, and public
monitoring message, together with all earlier public information. The realized
action is not a separate component of the public history unless it is publicly
observed; information about a hidden action may instead be conveyed through
the public outcome or monitoring message. The sufficient-state representation
of Section~3.4 applies to this enriched history:
\[
\chi_{t+1}:H_{t+1}\to Z_{t+1},
\qquad
z_{t+1}=\chi_{t+1}(h_{t+1}).
\]
The sufficient public state may therefore include monitoring-induced beliefs
or statistics, continuation promises, future authorization, or any other
public variables needed to make the continuation problem Markov.

The primitive monitoring kernel and the information-use policy together induce
a \textbf{public experiment}. For every measurable
\(B\subseteq\mathcal E_t^{\mathrm{pub}}\), define
\[
Q_t^\zeta
\bigl(
B\mid\xi_t,\theta_t,\tau_t,a_t
\bigr)
:=
\int_{\mathcal X_t}
\zeta_t
\bigl(
B\mid\xi_t,x
\bigr)
Q_t^{\mathrm{mon}}
\bigl(
dx\mid\xi_t,\theta_t,\tau_t,a_t
\bigr).
\]
The current environment
\((\xi_t,\theta_t,\tau_t,a_t)\) determines the distribution of raw monitoring
evidence through the fixed primitive kernel \(Q_t^{\mathrm{mon}}\). Holding
this monitoring technology fixed, the information-use policy \(\zeta_t\)
determines how the realized raw signal is transformed into a public message.
Conditional on a fixed current context \(\xi_t\), \(Q_t^{\mathrm{mon}}\) and
\(Q_t^\zeta\) are therefore experiments about the same hidden variables
\((\theta_t,\tau_t,a_t)\), with signal spaces \(\mathcal X_t\) and
\(\mathcal E_t^{\mathrm{pub}}\), respectively. In kernel-composition notation,
\(Q_t^\zeta=Q_t^{\mathrm{mon}}\zeta_t\). We use the Blackwell order
\(\succeq_B\) in this garbling sense.

\begin{definition}[Information feasibility]
\label{def:information_feasibility}
For a given current observable context \(\xi_t\), a public experiment \(Q_t'\)
is \textbf{information feasible} from the primitive monitoring technology
\(Q_t^{\mathrm{mon}}\) if there exists an admissible information-use policy
\(\zeta_t\) such that
\[
Q_t'
=
Q_t^{\mathrm{mon}}\zeta_t,
\qquad\text{or equivalently}\qquad
Q_t^{\mathrm{mon}}\succeq_B Q_t'.
\]
\end{definition}

\begin{proposition}[Feasible monitoring information]
\label{prop:feasible_information}
For a given current observable context \(\xi_t\):
\begin{enumerate}
\item[(i)] \textbf{Blackwell feasibility.}
The information-feasible public experiments are exactly the Blackwell
garblings of \(Q_t^{\mathrm{mon}}\).

\item[(ii)] \textbf{Bayes plausibility.}
For any induced prior over the hidden current variables
\((\theta_t,\tau_t,a_t)\) conditional on \(\xi_t\), every information-feasible
public experiment induces a Bayes-plausible distribution of posterior beliefs:
the average posterior equals the prior. The converse need not hold.
\end{enumerate}
\end{proposition}

\begin{proof}
Part~(i) follows directly from the definition. Every information-use policy
\(\zeta_t\) stochastically post-processes the raw monitoring signal and
therefore generates a Blackwell garbling of \(Q_t^{\mathrm{mon}}\).
Conversely, if \(Q_t^{\mathrm{mon}}\succeq_B Q_t'\), the definition of the
Blackwell order provides an information-use policy \(\zeta_t\) such that
\(Q_t'=Q_t^{\mathrm{mon}}\zeta_t\).

For part~(ii), Bayes plausibility follows from iterated expectations: averaging
the posterior induced by the public monitoring message reproduces the
conditional prior. The converse need not hold because Bayes plausibility
restricts only the mean posterior, whereas information feasibility requires
the entire posterior distribution to be generated by a garbling of
\(Q_t^{\mathrm{mon}}\). For example, if \(Q_t^{\mathrm{mon}}\) is
uninformative about \((\theta_t,\tau_t,a_t)\) conditional on \(\xi_t\),
every feasible garbling is also uninformative, even though a nondegenerate
conditional prior admits informative Bayes-plausible posterior splittings.
\end{proof}

Proposition~\ref{prop:feasible_information} separates \textbf{monitoring
capacity} from \textbf{information use}. The primitive experiment
\(Q_t^{\mathrm{mon}}\) determines the informational frontier available to the
mechanism, while \(\zeta_t\) selects how much of that information enters the
public continuation state. Bayes plausibility is therefore necessary but need
not characterize the feasible information set when monitoring capacity is
limited. A highly informative monitoring technology may optimally be used only
through a coarser public experiment.

\subsection{Recursive Joint Design of Information and Incentives}

We now incorporate information use into the recursive program of Section~3.4.
The current direct outcome rule \(\psi_t\) determines the report-contingent
mechanism response, authorization, recommended plan, and reward schedule, while
the information-use policy \(\zeta_t\) determines how raw monitoring evidence
is transformed into public information and hence affects the next public state.
We maintain the full-commitment assumption of the baseline model: at date \(0\),
the designer commits to both the mechanism and the information-use policy.
Action feasibility and dynamic incentive compatibility retain their definitions
from Section~3. A privately chosen action may affect the distribution of raw
monitoring evidence even though \(\zeta_t\) cannot condition on that action
directly.

To parallel the admissible-mechanism correspondence of Section~3.4, let
\(\Psi_t^{I}(z_t)\) denote the correspondence of admissible current
mechanism--information pairs \((\psi_t,\zeta_t)\) at public state \(z_t\).
A pair belongs to \(\Psi_t^{I}(z_t)\) if \(\zeta_t\) is information feasible
from \(Q_t^{\mathrm{mon}}\) and the pair can be extended by continuation
mechanism and information-use rules such that action feasibility, the recursive
incentive conditions of Proposition~\ref{prop:recursive_DIC}, and all other
restrictions maintained by the design problem hold after every state.
We maintain the well-posedness and, when \(T=\infty\), continuity-at-infinity
conditions of Section~3.4.

\begin{theorem}[Recursive joint design of information and incentives]
\label{thm:recursive_information_design}
The designer's history-dependent joint design problem is equivalent to the
recursive program
\[
J_t(z_t)
=
\sup_{(\psi_t,\zeta_t)\in\Psi_t^{I}(z_t)}
\int
\left[
v_t(
\widetilde a_t,
\widetilde\theta_t,
\widetilde o_t)
-
c_t(\widetilde r_t)
+
\delta
J_{t+1}(\widetilde z_{t+1})
\right]
\,d\mathbb P_t(\cdot\mid z_t;\psi_t,\zeta_t).
\]
When \(T<\infty\), \(J_{T+1}\equiv0\).
\end{theorem}

\begin{proof}
Proposition~\ref{prop:feasible_information} characterizes the public
experiments attainable from the primitive monitoring technology. Conditional
on an information-feasible \(\zeta_t\), the current mechanism rule remains
subject to the action-feasibility and recursive incentive conditions
characterized in Section~3. The information-use policy additionally determines
how current monitoring evidence enters the public history and therefore
changes the induced law of the next sufficient public state. When monitoring
depends on the hidden action, different action choices may induce different
laws of that evidence and hence different continuation states.

The principle-of-optimality argument of
Theorem~\ref{thm:recursive_principal} therefore applies with
\((\psi_t,\zeta_t)\) as the current design choice and
\(\Psi_t^{I}(z_t)\) as the admissible current-choice correspondence. This gives
the stated recursion.

Conversely, state-contingent pairs
\((\psi_t,\zeta_t)\in\Psi_t^{I}(z_t)\), together with admissible continuation
rules attaining or approximating the corresponding continuation values, can be
concatenated across public states. By construction,
Proposition~\ref{prop:feasible_information} guarantees information feasibility,
while Proposition~\ref{prop:recursive_DIC} guarantees dynamic incentive
compatibility of the resulting direct mechanism. The remaining restrictions
defining admissibility are preserved by the continuation construction. Hence
the recursive and history-dependent formulations generate the same attainable
designer continuation values.

For \(T<\infty\), equivalence follows by backward induction. For
\(T=\infty\), it follows under the continuity-at-infinity conditions maintained
in Section~3.4.
\end{proof}

Theorem~\ref{thm:recursive_information_design} extends the recursive
mechanism-design representation of
Theorem~\ref{thm:recursive_principal} by making information use an additional
design variable. Changing \(\zeta_t\) changes the distribution of
\(\widetilde z_{t+1}\), and therefore changes both the designer's continuation
value \(J_{t+1}\) and the agent's continuation value \(V_{t+1}\). The same
information policy that improves future decision quality may therefore alter
the return to a current report or post-report action. Information is
consequently both an epistemic instrument and an incentive instrument.

This interaction is conceptually related to the integrated information- and
mechanism-design perspective of \citet{BergemannHeumannMorris2026}. Their
static screening environment admits a majorization formulation that is
specific to that setting and is not imported here. The sequential-information
perspective of \citet{DovalEly2020} likewise provides a useful dynamic
benchmark, but the present problem fixes the primitive monitoring technology
and asks how its use interacts with a dynamic mechanism containing private
type reports and post-report actions. A related nonseparability result arises
in \citet{Tian2026}, where information production, information use, incentives,
posterior evolution, and feasibility are jointly endogenous. Here the narrower
source of nonseparability is the effect of information use on the continuation
values supporting dynamic incentive compatibility.

A simple Gaussian specification makes the state recursion transparent.
Suppress private-type uncertainty and current-action dependence of the
monitoring technology, and let the payoff-relevant state be scalar. Suppose
\(\widetilde\theta_t\mid z_t\sim N(\mu_t,\Sigma_t)\), and let the raw
monitoring signal satisfy
\[
\widetilde x_t
=
\widetilde\theta_t+\varepsilon_t^E,
\qquad
\varepsilon_t^E\sim N(0,\sigma_E^2).
\]
Let the information-use policy add independent Gaussian noise
\(\varepsilon_t^D\sim N(0,\sigma_D^2)\), so
\(e_t=x_t+\varepsilon_t^D\). Writing
\((\mu_t^+,\Sigma_t^+)\) for the posterior moments,
\[
\mu_t^+
=
\mu_t+
\frac{\Sigma_t}{\Sigma_t+\sigma_E^2+\sigma_D^2}(e_t-\mu_t),
\qquad
\Sigma_t^+
=
\frac{\Sigma_t(\sigma_E^2+\sigma_D^2)}
{\Sigma_t+\sigma_E^2+\sigma_D^2}.
\]
If the payoff-relevant state follows
\[
\widetilde\theta_{t+1}
=
\rho_\theta\widetilde\theta_t
+
\varepsilon_{t+1}^{\theta},
\qquad
\varepsilon_{t+1}^{\theta}
\sim N(0,\sigma_\theta^2),
\]
then
\[
\mu_{t+1}
=
\rho_\theta\mu_t^+,
\qquad
\Sigma_{t+1}
=
\rho_\theta^2\Sigma_t^+
+
\sigma_\theta^2.
\]
Hence \((\mu_t,\Sigma_t)\) is a low-dimensional component of the sufficient
public state \(z_t\). Within this Gaussian class, \(\sigma_D^2=0\) uses the
raw monitoring signal without additional garbling, while larger
\(\sigma_D^2\) produces a less informative public signal. These formulas
isolate the decision-value channel of information. Once future authorization,
rewards, or deployment depend on the updated public state, the same garbling
choice may also change current incentive comparisons, which is the channel
captured by Theorem~\ref{thm:recursive_information_design}.

\subsection{Selective Disclosure and Incentive Separation}

The distinction between the information that can be generated and the
information that is actually used yields two different comparative statics.
More informative monitoring capacity weakly expands the designer's opportunity
set when information can be freely discarded. A more informative public
experiment, by contrast, need not improve the designer's objective, because
information use changes the continuation decision problem and may also change
incentives. A separate question is whether information use can be optimized
independently of those incentive effects.

For a primitive monitoring experiment \(Q_t^{\mathrm{mon}}\), let
\(J_t^{Q^{\mathrm{mon}}}(z_t)\) denote the designer's optimal continuation
value when information use is chosen optimally from the information-feasible
garblings of \(Q_t^{\mathrm{mon}}\), holding the other primitives and future
monitoring technologies fixed. Blackwell comparisons below are understood
conditional on each observable context \(\xi_t\), with
\(Q_t^{\mathrm{mon}}\) viewed as an experiment about
\((\theta_t,\tau_t,a_t)\).

\begin{proposition}[Monitoring capacity and selective disclosure]
\label{prop:evaluation_disclosure}
Under the monitoring and information-design setting developed in this section,
the following statements hold.
\begin{enumerate}
\item[(i)] \textbf{Informational monotonicity under free disposal.}
If, at every observable context \(\xi_t\),
\(Q_t^{\mathrm{mon}\prime}\succeq_B Q_t^{\mathrm{mon}}\), and information
can be garbled or ignored at no additional cost, then
\[
J_t^{Q^{\mathrm{mon}\prime}}(z_t)
\geq
J_t^{Q^{\mathrm{mon}}}(z_t).
\]

\item[(ii)] \textbf{Full disclosure need not be optimal.}
There exist environments in which the primitive monitoring experiment is fully
revealing but an optimal information-use policy is strictly less informative
than full disclosure.
\end{enumerate}
\end{proposition}

\begin{proof}
For part~(i), Blackwell dominance implies that, at every feasible observable
context \(\xi_t\), \(Q_t^{\mathrm{mon}}\) can be reproduced by a garbling of
\(Q_t^{\mathrm{mon}\prime}\). The same garbling applies to every feasible
\((\theta_t,\tau_t,a_t)\). Take any mechanism and information-use policy
feasible under \(Q_t^{\mathrm{mon}}\). Under the more informative primitive
technology, first apply a garbling that reproduces
\(Q_t^{\mathrm{mon}}\), and then apply the original information-use policy.
The resulting conditional distribution of public monitoring messages is
therefore identical to that under the original joint design.

This replication holds following truthful-and-obedient behavior and following
every admissible report--action deviation. Hence it reproduces the same
distribution of continuation states, actions, rewards, and payoffs under every
incentive comparison and therefore preserves dynamic incentive compatibility.
Every joint design attainable under \(Q_t^{\mathrm{mon}}\) is consequently
attainable under \(Q_t^{\mathrm{mon}\prime}\), which implies the stated
inequality.

For part~(ii), consider a two-state continuation problem nested in the present
environment. Let \(\theta\in\{B,G\}\), with each state having probability
\(1/2\). At date \(t+1\), after observing the date-\(t\) public monitoring
message, the agent chooses between investing and not investing. Investment
gives the agent payoff \(-1\) in state \(B\) and \(x\in(0,1)\) in state \(G\),
while not investing gives zero. The designer receives payoff one whenever the
agent invests and zero otherwise. Suppose the primitive monitoring experiment
fully reveals the state.

Under full disclosure, the agent invests only in state \(G\), so the
designer's expected payoff is \(1/2\). Now consider any information-use policy
and let \(p_G\) and \(p_B\) denote the probabilities that the induced
continuation decision is investment in states \(G\) and \(B\), respectively.
For every public message after which investment is optimal, obedience requires
the agent's conditional expected payoff from investment to be nonnegative.
Aggregating these obedience inequalities over all such messages gives
\(p_B\leq xp_G\). The designer's expected payoff therefore satisfies
\[
\frac{p_B+p_G}{2}
\leq
\frac{(1+x)p_G}{2}
\leq
\frac{1+x}{2}.
\]

The bound is attained by a binary public recommendation that recommends
investment with probability \(p_G=1\) in state \(G\) and with probability
\(p_B=x\) in state \(B\). Following an investment recommendation, the agent
is indifferent between investing and not investing and obeys under weak
obedience; a no-investment recommendation occurs only in state \(B\) and is
strictly obeyed. The designer's expected payoff is therefore
\((1+x)/2>1/2\). Since the investment recommendation occurs with positive
probability in both states, the induced public experiment is a strict garbling
of the fully revealing primitive monitoring experiment. Thus full disclosure
is strictly suboptimal.
\end{proof}

Part~(i) is an option-value result. Better primitive monitoring capacity cannot
hurt when the designer can reproduce every policy available under a less
informative technology. Part~(ii) concerns a different comparison: conditional
on a given monitoring technology, using more of the available information need
not be optimal. A fully revealing raw signal expands the information-feasible
set, while the optimal policy may deliberately coarsen that signal before it
enters the continuation decision problem. This selective-disclosure logic is
familiar from information design \citep{BergemannMorris2019}; here it operates
inside a dynamic mechanism in which the resulting public information also
determines continuation control.

The distinction is especially relevant for AI-agent mechanisms. A more
informative monitoring technology, such as a more accurate capability or
behavioral evaluation, expands the menu of state-contingent continuation
policies. Automatically revealing or using every monitoring result, however,
may change obedience, the return to strategic underperformance, or future
authorization incentives. Because the public monitoring message enters the
continuation mechanism, information use can itself function as an
intertemporal incentive instrument. More monitoring capacity and more
disclosure are therefore different design choices.

Proposition~\ref{prop:evaluation_disclosure} distinguishes monitoring capacity
from information use. A separate question is when information use can be
chosen independently of the current incentive problem.
Theorem~\ref{thm:recursive_information_design} identifies the relevant object:
the effect of \(\zeta_t\) on the continuation-value difference between
truthful-and-obedient behavior and a feasible deviation.

Fix the current direct outcome rule \(\psi_t\) and the continuation mechanism
from date \(t+1\) onward. For a current true type \(\tau_t\), public state
\(z_t\), and an admissible one-stage report--action deviation
\((\widehat\tau_t,d_t)\), followed by truthful and obedient continuation play,
define the \textbf{continuation incentive difference} under the agent's
subjective beliefs by
\[
\begin{aligned}
\Delta V_{t+1}^{\zeta}
(\tau_t,z_t;\widehat\tau_t,d_t)
:={}&
\int
V_{t+1}(
\widetilde\tau_{t+1},
\widetilde z_{t+1})
\,d\mathbb P_t
(\cdot\mid
\tau_t,z_t,\tau_t,d_t^{\mathrm{ob}};\psi,\zeta_t)
\\
&-
\int
V_{t+1}(
\widetilde\tau_{t+1},
\widetilde z_{t+1})
\,d\mathbb P_t
(\cdot\mid
\tau_t,z_t,\widehat\tau_t,d_t;\psi,\zeta_t).
\end{aligned}
\]
The two conditional laws incorporate the continuation histories generated by
the corresponding current behavior and by the public monitoring message under
\(\zeta_t\). When monitoring depends on the hidden action, a report--action
deviation may change the distribution of raw evidence even though the action
itself never enters the public history.

Conditional on the current mechanism rule and continuation mechanism, call
information use \textbf{incentive separable} at date \(t\) if changing
\(\zeta_t\) within the information-feasible set leaves every current dynamic
incentive comparison unchanged.

\begin{corollary}[Information--incentive separation]
\label{cor:information_separation}
Fix \(\psi_t\) and the continuation mechanism from date \(t+1\) onward, and
suppose the current information-use policy affects date-\(t\) incentives only
through the continuation state.

\begin{enumerate}
\item[(i)] \textbf{Separation.}
Information use is incentive separable at date \(t\) if and only if
\[
\Delta V_{t+1}^{\zeta}
(\tau_t,z_t;\widehat\tau_t,d_t)
\]
is invariant across information-feasible policies \(\zeta_t\) for every
feasible \((\tau_t,z_t)\) and every admissible one-stage report--action
deviation \((\widehat\tau_t,d_t)\). Conditional on the fixed mechanism rules,
information use can then be optimized without re-solving the current incentive
constraints.

\item[(ii)] \textbf{Nonseparability.}
If this continuation incentive difference varies with \(\zeta_t\) for some
feasible \((\tau_t,z_t)\) and admissible one-stage report--action deviation,
then changing information use changes the corresponding dynamic incentive
constraint, and separate optimization of information use and current
incentives is not valid in general.
\end{enumerate}
\end{corollary}

\begin{proof}
Conditional on the fixed current outcome rule, the current payoff difference
between truthful-and-obedient behavior and a given one-stage report--action
deviation is independent of \(\zeta_t\), because the monitoring signal and
public monitoring message are generated only after the current action and
outcome. The corresponding dynamic incentive inequality therefore differs
across information-use policies only through
\(\delta\Delta V_{t+1}^{\zeta}\). Hence every date-\(t\) dynamic incentive
comparison is invariant to \(\zeta_t\) if and only if
\(\Delta V_{t+1}^{\zeta}\) is invariant across information-feasible policies
for every admissible deviation. If it varies for some admissible deviation,
changing \(\zeta_t\) changes the corresponding dynamic incentive inequality.
\end{proof}

The separation condition has substantive content. In the Gaussian benchmark
above, the raw monitoring law is independent of the current report and action,
and Bayes plausibility gives \(\E[\mu_t^+\mid z_t]=\mu_t\) for every Gaussian
garbling. If a truthful-versus-deviation continuation-value difference is
affine in \(\mu_t^+\), its expectation is therefore independent of
\(\sigma_D^2\). Information use may then change posterior precision and the
designer's continuation problem without changing the relevant current
incentive comparison.

By contrast, nonlinear continuation incentive differences generally depend on
the amount of garbling. Under the same Gaussian benchmark,
\[
\operatorname{Var}(\mu_t^+\mid z_t)
=
\frac{\Sigma_t^2}
{\Sigma_t+\sigma_E^2+\sigma_D^2},
\]
and hence
\[
\E[(\mu_t^+)^2\mid z_t]
=
\mu_t^2
+
\frac{\Sigma_t^2}
{\Sigma_t+\sigma_E^2+\sigma_D^2}.
\]
Thus a nonzero quadratic component in the continuation incentive difference
generically makes its expectation depend on \(\sigma_D^2\). Information use
then changes the relative continuation payoff from truthful-and-obedient
behavior and the deviation, so information and incentive design are
nonseparable.

When the monitoring technology depends on the hidden action, there is an
additional source of nonseparability: changing \(\zeta_t\) changes how
differences in privately chosen behavior are transmitted into public
continuation states. Information use may therefore strengthen or weaken a
hidden-action incentive without conditioning directly on the action itself.

This distinction also connects to the implementation margins of Section~4.
Information can improve future state contingency and control while
simultaneously tightening or relaxing dynamic incentive constraints and, in
the scalar verification environment, changing the continuation rents
associated with the verification wedge. The decision value and the incentive
value of information are therefore distinct.

The same logic clarifies strategic behavior under AI evaluation. Current
reports and actions may change the distribution of raw monitoring evidence,
while the use of that evidence determines the continuation mechanism. An agent
may therefore care about an evaluation not merely because it reveals
information, but because the resulting public state determines future
permissions, rewards, or deployment opportunities. The mechanism must account
for how monitoring affects both truthful revelation and subsequent action
choice. This is the dynamic counterpart of the sandbagging and evaluation
concerns emphasized by \citet{BergemannKohMorris2026} and provides the
theoretical basis for the strategic-evaluation application developed below.

The present section keeps primitive information production fixed. Within that
boundary, it characterizes feasible information use, embeds information use in
the recursive mechanism-design problem, distinguishes monitoring capacity from
selective disclosure, and identifies the boundary between separable and
nonseparable information and incentive design. Endogenizing information
production itself leads to the broader interaction among learning, incentives,
posterior evolution, and feasibility studied in \citet{Tian2026}.

\section{Multiple AI Agents}

We now extend the baseline mechanism and incentive framework of
Sections~2--4 to multiple interacting AI agents. The extension preserves
private types, partial verification, post-report private information,
authorization, and post-report action discretion, but adds strategic
interaction across agents. An agent's beliefs may concern other agents'
private types and future behavior, and the mechanism may correlate
recommendations, authorization, rewards, and continuation opportunities
across agents. The dynamic environment also allows private types to evolve
over time. Latent current actions may affect public outcomes, monitoring
evidence, physical transitions, and hence the resulting public continuation
state, even when the actions themselves are not publicly observed.

Strategic interaction requires an explicit equilibrium concept and a
distinction between information privately available to individual agents,
information available to the mechanism, and information publicly observed
by all agents. We formulate the multi-agent environment as a dynamic
Bayesian game and use perfect Bayesian equilibrium, allowing heterogeneous
subjective beliefs, as the solution concept. The mechanism can condition
on confidential reports and privately communicate correlated
recommendations without necessarily revealing them to other agents.
When the monitoring structure of Section~5 is present, public monitoring
messages are incorporated into the common public history and sufficient
public state. Throughout this section, the primitive monitoring technology
and information-use policy are held fixed.

\subsection{Multi-Agent Environment and Mechanisms}

There are \(n\) agents indexed by \(i\in N:=\{1,\ldots,n\}\). For each date
\(t\), let \(\mathcal A_{it}\) be agent \(i\)'s ambient action space and let
\(\mathcal A_t:=\prod_{i=1}^n\mathcal A_{it}\). Agent \(i\)'s private type is
the agent-specific analogue of Section~2:
\[
\tau_{it}
=
(A_{it},u_{it},b_{it},\pi_{it})
\in\mathcal T_{it},
\qquad
\tau_t
=
(\tau_{1t},\ldots,\tau_{nt})
\in
\mathcal T_t
:=
\prod_{i=1}^n\mathcal T_{it}.
\]
Here \(A_{it}\subseteq\mathcal A_{it}\) is its execution capability,
\(u_{it}\) its period payoff representation, \(b_{it}\) its subjective belief,
and \(\pi_{it}\) its information capability. The realized action profile is
\(a_t=(a_{1t},\ldots,a_{nt})\in\prod_{i=1}^n A_{it}\subseteq\mathcal A_t\).

The belief component now includes uncertainty about the common state and
the other agents. For type \(\tau_{it}\), let
\(b_{it}(\cdot\mid h_t)\in
\Delta(\Theta_t\times\mathcal T_{-i,t})\), where
\(\mathcal T_{-i,t}:=\prod_{j\neq i}\mathcal T_{jt}\).
When higher-order beliefs matter, the type spaces
\((\mathcal T_{it})_{i\in N}\) are understood as coherent type spaces whose
belief components encode beliefs about other agents' types, including
their belief components and hence their higher-order beliefs. A universal
type space in the sense of Mertens and Zamir
\citeyearpar{MertensZamir1985} provides the fully general benchmark, while
applications may instead use a finite-dimensional sufficient private
state.

The subjective beliefs are dynamically coherent: conditional beliefs at
successive information sets are obtained from the agents' subjective
models and the commonly known signal, outcome, and transition
technologies wherever Bayesian conditioning applies. The current type
is understood to summarize the agent's payoff- and belief-relevant
private information before date-\(t\) communication, as in Section~2.
If earlier private observations remain relevant, the type space is
enlarged to include a sufficient private history. Subsequent private
mechanism communications and interim signals further update the
agent's beliefs within the period. The conditional beliefs used for
sequential optimization are specified formally below.

As in Section~2, we assume that the designer's conditional law of
\((\theta_t,\tau_{-i,t})\), given \((h_t,\tau_{it})\), satisfies
\[
\mathbb P_t(\cdot\mid h_t,\tau_{it})
\ll
b_{it}(\cdot\mid h_t),
\]
whenever this conditional law is defined.

Let \(\mathcal S_t:=\prod_{i=1}^n\mathcal S_{it}\). Conditional on
the current state and type profile, the interim signal profile is
generated by the joint kernel
\[
\widetilde s_t
\sim
\pi_t^N(\cdot\mid\theta_t,\tau_t)
\in\Delta(\mathcal S_t).
\]
The joint signal technology is common knowledge, while each agent's
realized information capability \(\pi_{it}\) remains private. The
\(i\)-th marginal of \(\pi_t^N\) is
\(\pi_{it}(\cdot\mid\theta_t)\), and private signals may be
correlated across agents. Public outcomes and dynamic evolution
extend the kernels of Section~2 to the full action and type profiles:
\[
\widetilde o_t
\sim Q_t(\cdot\mid\theta_t,a_t),
\qquad
(\widetilde\theta_{t+1},\widetilde\tau_{t+1})
\sim K_{t+1}(\cdot\mid h_t,\theta_t,\tau_t,a_t,o_t).
\]
Thus an agent's current action may affect other agents' future types,
capabilities, and information, even when the action itself is not
publicly observed.

The period payoffs for the designer and agent \(i\), respectively, are
\[
v_t(a_t,\theta_t,o_t)
-
\sum_{j=1}^n c_{jt}(r_{jt}),
\qquad
u_{it}(a_t,\theta_t)+r_{it}.
\]
Allowing \(u_{it}\) to depend on the full action profile accommodates
strategic complements, strategic substitutes, congestion, competition,
coordination, and other primitive cross-agent payoff effects.

Partial verification is agent specific. Let \(\succeq_{it}\) be a
verification preorder on \(\mathcal T_{it}\), with claim set
\[
C_{it}(\tau_{it})
:=
\left\{
\widehat\tau_{it}\in\mathcal T_{it}:
\tau_{it}\succeq_{it}\widehat\tau_{it}
\right\}.
\]
The feasible report profile is
\(C_t(\tau_t):=\prod_{i=1}^nC_{it}(\tau_{it})\).
Capability consistency requires
\(\widehat A_{it}\subseteq A_{it}\) whenever
\(\tau_{it}\succeq_{it}\widehat\tau_{it}\).

The mechanism notation extends Section~3 agent by agent, with an
important distinction between common public information and
confidential communication. Let \(H_t\) denote the space of common
public histories immediately before date-\(t\) communication, with
\(H_1=\{\emptyset\}\), and let \(h_t\in H_t\) denote a realized public
history. For \(t\geq2\), \(h_t\) records the publicly observed events
at earlier dates, including public mechanism responses, disclosed
communications and control choices, authorization profiles, realized
rewards, public outcomes, and, when the monitoring structure of
Section~5 is present, public monitoring messages. A realized action is
not a separate component of the public history unless it is publicly
observed; information about a latent action may instead be conveyed
through public outcomes or other publicly observed signals.

Unlike the single-agent environment, agents' reports and recommendations
need not be publicly observed. Reports are submitted confidentially to
the mechanism, and each agent privately receives its own recommended
plan. Neither other agents' reports nor their recommended plans enter
\(h_t\) unless the mechanism explicitly discloses them. This convention
preserves the distinction between the information available to the
mechanism and the information available to each agent.

To allow continuation rules to depend on confidential past
communications, let \(\overline H_t\) denote the space of histories
observed by the mechanism, with \(\overline H_1=\{\emptyset\}\), and
write \(\bar h_t\in\overline H_t\) for a realized mechanism history.
The record \(\bar h_t\) contains the common public history and the
mechanism's past confidential communications and realized choices,
but not privately observed signals or latent actions unless they
have become observable to the mechanism. The public history \(h_t\)
is the publicly observed projection of \(\bar h_t\).
Thus a committed mechanism may condition future decisions on
confidential reports without disclosing those reports to other
agents. The original public-history formulation is recovered when
the mechanism's relevant history is entirely public.

For each date \(t\), let the mechanism response space have the form
\[
Y_t
=
Y_t^{\mathrm{pub}}
\times
\prod_{i=1}^n Y_{it}^{\mathrm{pr}},
\qquad
y_t
=
(y_t^{\mathrm{pub}},
y_{1t}^{\mathrm{pr}},\ldots,y_{nt}^{\mathrm{pr}}).
\]
The component \(y_t^{\mathrm{pub}}\) is publicly observed, whereas
\(y_{it}^{\mathrm{pr}}\) is privately communicated to agent \(i\).
The mechanism may correlate these components. A fully public
mechanism response is obtained when the private response spaces
are singletons.

For each agent \(i\), let \(M_{it}\) denote its ambient message space
and define \(M_t:=\prod_{i=1}^nM_{it}\). Let
\[
\mathcal M_t(h_t,\tau_t)
:=
\prod_{i=1}^n
\mathcal M_{it}(h_t,\tau_{it}),
\]
with verification consistency
\(\mathcal M_{it}(h_t,\widehat\tau_{it})
\subseteq\mathcal M_{it}(h_t,\tau_{it})\)
whenever
\(\tau_{it}\succeq_{it}\widehat\tau_{it}\).
The report or message of each agent is observed by the mechanism
but not automatically by the other agents.

Let
\(\mathcal D_{it}:=
2^{\mathcal A_{it}}\setminus\{\emptyset\}\)
and
\(\mathcal D_t^N:=\prod_{i=1}^n\mathcal D_{it}\).
We retain the convention that the realized authorization profile
\(D_t=(D_{1t},\ldots,D_{nt})\) is publicly announced. Authorization
may therefore convey information about the mechanism's confidential
observations, but does not itself disclose the underlying reports.
For the indirect mechanism, let
\(\mathcal R_{it}
:=
\{R:Y_t\times\mathcal D_t^N\times\mathcal O_t\to\R\}\)
denote agent \(i\)'s reward-schedule space.

A multi-agent indirect mechanism is
\(\Gamma=(M,\mathcal M,\varphi)\), with
\[
\varphi_t:
\overline H_t\times M_t
\longrightarrow
\Delta(Y_t\times\mathcal D_t^N)
\times
\prod_{i=1}^n\mathcal R_{it}.
\]
For each \((\bar h_t,m_t)\), the outcome rule specifies a joint
distribution
\(\lambda_t(\cdot\mid\bar h_t,m_t)\)
over the mechanism response and authorization profile, together
with reward schedules
\(R_{it}(\bar h_t,m_t,\cdot,\cdot,\cdot)\in\mathcal R_{it}\).
Thus, after \((y_t,D_t)\) is realized, agent \(i\)'s reward may
depend on the realized mechanism response, the full authorization
profile, and the public outcome,
\[
r_{it}
=
R_{it}(\bar h_t,m_t,y_t,D_t,o_t).
\]
The mechanism can condition rewards on its confidential information
and its own realized choices, even when the agents do not observe
all those variables. Reward determination does not condition directly
on a latent action profile. When the public outcome reveals some or
all actions, the formulation nests the corresponding action-contingent
reward through \(o_t\). The mechanism may therefore correlate
authorization across agents and condition rewards on joint publicly
observed consequences and realized mechanism choices.

Under an indirect mechanism, the agents simultaneously submit feasible
messages. Conditional on the message profile and the mechanism's
history, the mechanism draws \((y_t,D_t)\). Each agent observes
\(y_t^{\mathrm{pub}}\), the public authorization profile \(D_t\), and
its own private response \(y_{it}^{\mathrm{pr}}\). The joint interim
signal profile is then generated according to \(\pi_t^N\), and each
agent observes its own signal before simultaneously choosing a feasible
action \(a_{it}\in A_{it}\cap D_{it}\). The realized action profile
generates the public outcome and rewards and may affect the subsequent
physical state and private-type profile. The public and mechanism
histories are then updated using their respective observable records.

For direct mechanisms, let
\[
\mathcal L_{it}
:=
\left\{
(D_i,\alpha_i):
D_i\in\mathcal D_{it},\
\alpha_i:\mathcal S_{it}\to D_i
\text{ measurable}
\right\},
\qquad
\mathcal L_t
:=
\prod_{i=1}^n\mathcal L_{it}.
\]
Let
\(\mathcal R_{it}^D
:=
\{R:Y_t\times\mathcal L_t\times\mathcal O_t\to\R\}\)
denote agent \(i\)'s direct-mechanism reward-schedule space.

A direct multi-agent mechanism is
\(\Gamma^D=(\mathcal T,\psi)\), where
\[
\psi_t:
\overline H_t\times\mathcal T_t
\longrightarrow
\Delta(Y_t\times\mathcal L_t)
\times
\prod_{i=1}^n\mathcal R_{it}^D.
\]

For a report profile \(\widehat\tau_t\),
\(\psi_t(\bar h_t,\widehat\tau_t)\) specifies the joint kernel
\(\lambda_t^D(\cdot\mid\bar h_t,\widehat\tau_t)\)
over the mechanism response and authorization--plan profile,
together with reward schedules
\(R_{it}^D(\bar h_t,\widehat\tau_t,\cdot,\cdot,\cdot)\).
The kernel \(\lambda_t^D\) may correlate mechanism responses,
authorization sets, and signal-contingent plans across agents. If
\[
\ell_t
=
((D_{1t},\alpha_{1t}),\ldots,(D_{nt},\alpha_{nt}))
\in\mathcal L_t
\]
is realized, the authorization profile \(D_t\) is publicly announced,
while agent \(i\) privately receives its own plan
\(\alpha_{it}\), together with its private response
\(y_{it}^{\mathrm{pr}}\). The other agents' recommended plans
are not disclosed unless the mechanism explicitly makes them public.

Agent \(i\)'s reward may depend on the report profile, realized
mechanism response, the full authorization--plan profile, and the
public outcome,
\[
r_{it}
=
R_{it}^D
(\bar h_t,\widehat\tau_t,y_t,\ell_t,o_t).
\]
The mechanism observes the report and recommendation profiles and
can condition rewards on them, although individual agents need not
observe the full profiles. The realized reward vector is publicly
observed after the current actions and outcomes, as in the baseline
public-history convention. Any information thereby revealed is
incorporated into subsequent beliefs and continuation incentives.

The direct mechanism is \textbf{action feasible} if, for every
mechanism history \(\bar h_t\), every report profile
\(\widehat\tau_t\), and every realized
\((y_t,((D_{it},\alpha_{it}))_{i\in N})\) in the support of
\(\lambda_t^D(\cdot\mid\bar h_t,\widehat\tau_t)\),
\[
\alpha_{it}(s_{it})
\in
\widehat A_{it}\cap D_{it}
\qquad
\text{for every }i\in N
\text{ and every }s_{it}\in\mathcal S_{it}.
\]
As in the single-agent model, communication, capability,
authorization, and recommendation are distinct objects. The claim
set \(C_{it}(\tau_{it})\) restricts what agent \(i\) may report;
\(A_{it}\) describes what it can execute; \(D_{it}\) describes
what the mechanism permits; and \(\alpha_{it}\) specifies what
the mechanism recommends after each private signal. The agent's
realized action remains its own choice from \(A_{it}\cap D_{it}\)
and need not be publicly observed.

The information structure of the direct mechanism is otherwise
parallel to that of the indirect mechanism. After observing its
current private type, each agent submits a confidential feasible
report \(\widehat\tau_{it}\in C_{it}(\tau_{it})\). The mechanism
draws the joint response and authorization--plan profile. Each
agent observes the public response and authorization profile and
privately receives its own response and recommended plan. The
joint signal profile is then generated using the true type profile,
not the reported one, and agent \(i\) observes only \(s_{it}\).
The agents simultaneously choose their realized actions, which
generate public outcomes, rewards, and continuation consequences.
In particular, a false report does not alter the true execution
capability or information experiment of the deviating agent.

We now define equilibrium for the multi-agent environment.
For each agent $i$ and date $t$, let $\mathcal H_{it}$ denote
the set of feasible private decision histories at the reporting
and post-report action stages. These histories include the common
public history and the agent's privately observed types, messages,
mechanism responses, recommendations, and signals, subject to the
timing of observation specified above. An agent's strategy consists
of a reporting strategy $\sigma_i$ and a post-report action strategy
$d_i$, with actions in $A_{it}\cap D_{it}$. Let $\mu_i$ denote
the agent's system of conditional subjective beliefs.

For each $h_{it}\in\mathcal H_{it}$, let $\mathcal G_i(h_{it})$
denote the set of admissible measurable continuation
report--action strategies, with decisions already taken at
$h_{it}$ held fixed. Let
$V_{it}(\sigma_i,d_i\mid h_{it};\Gamma)$ denote agent $i$'s
expected continuation payoff under $(\sigma_i,d_i)$,
holding the other agents' equilibrium strategies fixed
and evaluating expectations under his conditional beliefs.
We maintain the measurability, integrability, and regular
conditional probability conditions needed for these objects
to be well defined.

\begin{definition}[Multi-agent perfect Bayesian equilibrium]
\label{def:multiagent-pbe}
For a committed mechanism $\Gamma$, a PBE is an assessment
$(\sigma^*,d^*,\mu^*)$ satisfying, for every agent $i$,
date $t$, and feasible private decision history
$h_{it}\in\mathcal H_{it}$:
\begin{enumerate}
\item[(i)] \textbf{Sequential rationality.}
\[
V_{it}(\sigma_i^*,d_i^*\mid h_{it};\Gamma)
\ge
V_{it}(\sigma_i,d_i\mid h_{it};\Gamma),
\quad
\forall(\sigma_i,d_i)\in\mathcal G_i(h_{it}).
\]

\item[(ii)] \textbf{Belief consistency.}
\[
\mu_i^*(\cdot\mid h_{it})=
\begin{cases}
\mathbb P_{it}(\cdot\mid h_{it}\,;\,\Gamma,\sigma^*,d^*),
& \text{if }\mathbb P_{it}(h_{it}\,;\,\Gamma,\sigma^*,d^*)>0,\\
\text{any admissible conditional belief},
& \text{otherwise}.
\end{cases}
\]
\end{enumerate}
\end{definition}

For non-atomic histories, belief consistency is understood
through regular conditional distributions, with Bayesian
agreement almost surely under the equilibrium-induced
subjective law. Agents need not share a common prior with one
another or with the designer. Each agent evaluates continuation
payoffs under his own conditional beliefs, while the designer
evaluates outcome laws under its own probability model.

The definition applies to both indirect and direct mechanisms.
Under a direct mechanism, an equilibrium is
\textbf{truthful and obedient} if agents report their true
current types and follow their prescribed signal-contingent
plans. An outcome law is \textbf{equilibrium attainable}
if it is induced by a PBE of an admissible mechanism.
We consider \textbf{partial implementation}: an outcome need
only be supported by at least one equilibrium.

The distinction between public and private information is
maintained throughout. Confidential reports, private
recommendations, and raw monitoring evidence enter the common
public history only through the committed disclosure policy.

\subsection{Revelation and Recursive Incentives}

Fix an action-feasible direct mechanism \(\Gamma^D\). Holding the other
agents' equilibrium strategies fixed, let
\(V_{it}(\sigma_i,d_i\mid h_{it};\Gamma^D)\) denote agent \(i\)'s
subjective continuation payoff under an admissible report--action
strategy \((\sigma_i,d_i)\). Write \(\sigma_i^{\mathrm{tr}}\) for truthful
reporting and \(d_i^{\mathrm{ob}}\) for obedience. Let
\(\mathcal G_i^A(h_{it})\) denote the admissible post-report action
continuation strategies at \(h_{it}\in\mathcal H_{it}^{A}\), holding
subsequent reporting at its equilibrium strategy.

\begin{definition}[Multi-agent dynamic incentive compatibility]
An action-feasible direct mechanism \(\Gamma^D\) is
\textbf{multi-agent dynamically incentive compatible} if it admits
a belief-consistent assessment with sequentially optimal off-path
continuations and, for every agent \(i\) and date \(t\), the following
conditions hold:

\begin{enumerate}
\item[(i)] \textbf{Reporting incentive compatibility.}
At every feasible reporting information history
\(h_{it}\in\mathcal H_{it}^{R}\) on the truthful-and-obedient
continuation,
\[
V_{it}
(\sigma_i^{\mathrm{tr}},d_i^{\mathrm{ob}}\mid h_{it};\Gamma^D)
\geq
V_{it}
(\sigma_i,d_i\mid h_{it};\Gamma^D),
\qquad
\forall(\sigma_i,d_i)\in\mathcal G_i(h_{it}).
\]

\item[(ii)] \textbf{Post-report obedience.}
At every feasible post-report information history
\(h_{it}\in\mathcal H_{it}^{A}\) following truthful reporting,
\[
d_i^{\mathrm{ob}}
\in
\arg\max_{d_i\in\mathcal G_i^A(h_{it})}
V_{it}(\sigma_i^*,d_i\mid h_{it};\Gamma^D).
\]
\end{enumerate}
\end{definition}

Following a false report, the supporting equilibrium assessment
requires sequential optimality but does not impose obedience to the
original recommendation. Thus reporting deviations allow subsequent
action reoptimization, while obedience is required separately on the
truthful continuation.

We next extend the revelation principle of Section~3. As there,
let \(\mathcal O\) and \(\mathcal O^{IC}\) denote, respectively, the
sets of equilibrium-attainable and dynamically implementable outcome
laws, now for the multi-agent environment restricted to the class
satisfying the conditions below. The revelation argument
maintains the following conditions: reported types and simulated
mechanism histories contain the information required to reproduce
indirect equilibrium strategies; private response spaces can transmit
the corresponding recommended plans; and the simulation preserves
verification feasibility, each agent's conditional information, and
unilateral deviation opportunities, including at off-path histories.
When current types are insufficient, their spaces are understood to
include the necessary private-history information, as specified in
Section~6.1.

\begin{theorem}[Multi-agent dynamic revelation principle]
\label{thm:multiagent_revelation}
Under the preceding conditions,
\[
\mathcal O=\mathcal O^{IC}.
\]
\end{theorem}

\begin{proof}
Fix an indirect mechanism and a perfect Bayesian equilibrium satisfying
the maintained conditions. After each confidential type-report profile,
construct a direct mechanism that simulates the indirect equilibrium
messages and mechanism history. The mechanism reproduces the joint
distribution of responses, authorization sets, rewards, and
signal-contingent action plans. Each agent's plan is generated from
its indirect equilibrium strategy and communicated privately, preserving
the agents' conditional information and correlations across private
responses. Truthful reporting and obedience therefore reproduce the
original outcome law, and action feasibility follows from the
indirect equilibrium strategies.

Consider a unilateral report--action deviation by agent \(i\). For
every \(\widehat\tau_{it}\in C_{it}(\tau_{it})\), verification
consistency makes the simulated equilibrium message of
\(\widehat\tau_{it}\) feasible for the true type. Information
preservation allows the agent's subsequent direct action choices
and continuation observations to be reproduced under the indirect
mechanism, with the same subjective conditional laws and continuation
payoffs. A profitable direct deviation would therefore imply a
profitable indirect deviation, contradicting sequential optimality.
The simulated equilibrium assessment also supplies consistent
off-path beliefs and sequentially optimal continuation strategies.
Hence \(\mathcal O\subseteq\mathcal O^{IC}\).

Conversely, any action-feasible direct mechanism can be represented
as an indirect mechanism with type reports as messages and privately
communicated recommended plans. Its verification restrictions,
reward rules, and equilibrium assessment are preserved. Thus
\(\mathcal O^{IC}\subseteq\mathcal O\).
\end{proof}

We now characterize recursive incentives. Let
\(z_t=\chi_t(h_t)\) be a sufficient public state. Suppose that
\((\tau_{it},z_t)\) is sufficient for agent \(i\)'s continuation
problem at every private history, including histories
generated by unilateral deviations. Sufficiency covers conditional
beliefs, the mechanism's confidential record, other agents'
equilibrium strategies, and feasible continuation choices. Otherwise,
the relevant private histories must be retained in the recursion.

Let \(V_{it}(\tau_{it},z_t)\) denote agent \(i\)'s equilibrium
continuation value, which is truthful and obedient on the intended
continuation. For a feasible current report \(\widehat\tau_{it}\),
let
\(\widehat V_{it}(\tau_{it},z_t;\widehat\tau_{it})\)
denote the payoff from that report, optimal feasible current
post-report action reoptimization, and equilibrium continuation
at the resulting histories. These values are evaluated under
agent \(i\)'s subjective beliefs, holding the other agents'
equilibrium strategies fixed.

Assume that the post-report action problems attain measurable maxima.
For \(T=\infty\), maintain the continuity-at-infinity condition of
Section~3.4 agentwise, uniformly over feasible private histories
and unilateral continuation strategies. For \(T<\infty\), set
\(V_{i,T+1}\equiv0\).

\begin{theorem}[Recursive characterization with multiple agents]
\label{thm:multiagent_recursive}
Under the preceding conditions, \(\Gamma^D\) is multi-agent
dynamically incentive compatible if and only if, for every
\(i\in N\) and every feasible truthful-continuation state:

\begin{enumerate}
\item[(i)] \textbf{Recursive obedience.}
At each post-report information history \(h_{it}\) following
a truthful report, with realized authorization \(D_{it}\),
recommended plan \(\alpha_{it}\), and private signal \(s_{it}\),
\[
\begin{aligned}
\alpha_{it}(s_{it})
\in\arg\max_{a_i\in A_{it}\cap D_{it}}
\int\Big[
&u_{it}((a_i,\widetilde a_{-i,t}),\widetilde\theta_t)
+\widetilde r_{it}
\\
&+\delta V_{i,t+1}
(\widetilde\tau_{i,t+1},\widetilde z_{t+1})
\Big]
\,d\mathbb P_{it}
(\cdot\mid h_{it};\Gamma^D,a_i).
\end{aligned}
\]

\item[(ii)] \textbf{Recursive honesty.}
\[
\tau_{it}
\in
\arg\max_{\widehat\tau_{it}\in C_{it}(\tau_{it})}
\widehat V_{it}
(\tau_{it},z_t;\widehat\tau_{it}).
\]

\item[(iii)] \textbf{Promise consistency.}
\[
V_{it}(\tau_{it},z_t)
=
\widehat V_{it}
(\tau_{it},z_t;\tau_{it}).
\]
\end{enumerate}
\end{theorem}

In condition~(i), the conditional law incorporates the other agents'
equilibrium actions, the current outcome, and the action-dependent
transition. Conditions~(ii)--(iii) equivalently give the
\textbf{nested report--action condition}
\[
V_{it}(\tau_{it},z_t)
\geq
\widehat V_{it}
(\tau_{it},z_t;\widehat\tau_{it}),
\qquad
\forall\widehat\tau_{it}\in C_{it}(\tau_{it}),
\]
with equality at the truthful report. These conditions are agentwise,
but continuation incentives may remain strategically coupled.

\begin{proof}
Necessity follows from one-stage unilateral deviations. Current
action deviations following truthful reporting yield recursive
obedience. For any feasible report \(\widehat\tau_{it}\), the
deviation that submits that report, optimally rechooses the current
action, and then follows the equilibrium continuation strategies
has value
\(\widehat V_{it}(\tau_{it},z_t;\widehat\tau_{it})\).
Dynamic incentive compatibility therefore implies recursive
honesty, while obedience at the truthful report gives promise
consistency.

Conversely, \(\widehat V_{it}\) incorporates optimal current action
reoptimization and the continuation effects of each feasible report.
Conditions~(i)--(iii) rule out profitable one-stage report--action
deviations on the truthful continuation, while the supporting
equilibrium assessment ensures sequential optimality at other
feasible histories. Backward induction for \(T<\infty\), or
agentwise uniform continuity at infinity for \(T=\infty\), rules
out profitable multistage deviations. Together with the belief
consistency of Section~6.1, these conditions establish
truthful-and-obedient perfect Bayesian equilibrium.
\end{proof}

\subsection{Optimal Multi-Agent Dynamic Mechanism Design}

To obtain a benchmark parallel to Section~4, we focus on a multi-agent
environment with independent scalar private types under a common prior shared
by the designer and all agents. For each \(i\in N\), let
\[
\eta_{it}\in
E_{it}
=
[\underline\eta_{it},\overline\eta_{it}],
\qquad
\eta_i^t
=
(\eta_{i1},\ldots,\eta_{it}),
\]
and let \(\eta_t=(\eta_{1t},\ldots,\eta_{nt})\). The mechanism induces an
allocation vector \(q_t=(q_{1t},\ldots,q_{nt})\in\mathcal Q_t\).

As in Section~4, the scalar private-type formulation concerns the
hidden-information dimension and retains post-report action discretion. The
allocation \(q_{it}\) is agent \(i\)'s truthful-and-obedient control or allocation; when
implemented through a recommended action, the realized action \(a_{it}\)
remains privately chosen and need not be publicly observed. A deviation may
therefore involve action reoptimization and affect public outcomes, monitoring
evidence, physical transitions, and continuation states.

Agent \(i\)'s gross utility is \(u_{it}(\eta_{it},q_{it},\theta_t)\).
Let \(p_{it}\) denote the payment from agent \(i\) to the designer.
The period payoffs for agent \(i\) and the designer, respectively, are
\[
u_{it}(\eta_{it},q_{it},\theta_t)-p_{it},
\qquad
v_t(q_t,\theta_t)+\sum_{i=1}^n p_{it}.
\]
As in Section~4, transfers may be induced through observable mechanism choices
and outcomes; the formulation does not require the realized action to be
publicly observed.

Reports satisfy
\(C_{it}(\eta_{it})=[\underline\eta_{it},\eta_{it}]\).
The initial scalar types are independent across agents under the common
prior. Each scalar-type process admits the independent-shock representation
and regular impulse responses of Section~4.1, with transitions exogenous to
mechanism choices and independent across agents conditional on the maintained
public information. We impose the corresponding agentwise scalar-envelope
regularity: holding an agent's report fixed, a marginal change in its true
scalar type affects its continuation utility only through its current and
future gross payoffs and the induced evolution of its own scalar types, with
no additional derivatives of the conditional laws of the common state or
other agents' private information. These restrictions, together with the
common prior, allow the dynamic-envelope and transfer-accounting arguments
of Section~4 to apply agent by agent. Let \(I_{i,1,t}(\eta_i^t)\) denote the
impulse response of \(\eta_{it}\) to \(\eta_{i1}\), with
\(I_{i,1,1}=1\). Sections~6.1--6.2 continue to allow heterogeneous beliefs
and correlated private types.

Let \(V_{it}(\eta_i^t,z_t)\) denote agent \(i\)'s truthful-and-obedient
continuation value and, for
\(\widehat\eta_{it}\in C_{it}(\eta_{it})\), let
\(\widehat V_{it}(\eta_i^t,z_t;\widehat\eta_{it})\) denote its deviation
value after optimal current post-report action reoptimization, holding the
other agents' equilibrium strategies fixed. These values retain the
sufficiency and sequential optimality requirements of Section~6.2.

Participation is imposed after each agent observes its initial private type.
The designer's full multi-agent \(4B\) problem is
\begin{equation}
\label{eq:multiagent_full_problem}
\begin{aligned}
\sup_{\psi}\quad
&
\Pi^P(\psi)
=
\int
\sum_{t=1}^{T}
\delta^{t-1}
\left(
v_t(q_t,\theta_t)
+
\sum_{i=1}^n p_{it}
\right)
\,d\mathbb P(\cdot;\psi)
\\
\text{s.t.}\quad
&
V_{i1}(\eta_{i1})\geq0,
\qquad
\forall i,\ \forall\eta_{i1}\in E_{i1},
\\
&
V_{it}(\eta_i^t,z_t)
\geq
\widehat V_{it}
(\eta_i^t,z_t;\widehat\eta_{it}),
\\[-0.2em]
&
\hspace{5em}
\forall i,\ \forall t,\
\forall\widehat\eta_{it}\in C_{it}(\eta_{it}),
\\
&
\psi\text{ satisfies post-report obedience and is}
\\
&
\hspace{2em}\text{technologically and action feasible.}
\end{aligned}
\end{equation}
The incentive constraints apply after every history. In particular,
the reporting comparisons allow the true type to reoptimize its current
action, while post-report obedience imposes sequential optimality after the
truthful report, as in Theorem~\ref{thm:multiagent_recursive}. The three
margins of Section~4 remain: hidden information, one-sided verification, and
hidden-action implementation.

Accordingly, the multi-agent benchmark structure is branched rather than
linear. Holding joint feasibility fixed, \(2B^{\mathrm{HI}}\) denotes the
hidden-information problem with directly assignable actions;
\(2B^{\mathrm{HA}}\) the complete-information hidden-action problem;
\(3B^{\mathrm{HI+V}}\) combines hidden information and one-sided verification
without hidden action; \(3B^{\mathrm{HI+HA}}\) combines hidden information
and hidden action under unrestricted reporting; and \(4B\) combines all three
margins. The \(1B\) benchmark removes both private-information and
action-implementation frictions. These labels apply to all agents; mixed
agent-specific configurations are also possible.

Let \(\kappa_{i1}(\eta_{i1})\) denote agent \(i\)'s initial verification
wedge, as in Section~4.1. Suppose \(F_{i1}\) admits a strictly positive
density \(f_{i1}\) on \(E_{i1}\), and define the inverse hazard rate by
\begin{equation}
\label{eq:multiagent_initial_hazard}
\Lambda_{i1}(\eta_{i1})
:=
\frac{1-F_{i1}(\eta_{i1})}
{f_{i1}(\eta_{i1})}.
\end{equation}

\begin{theorem}[Multi-agent verification-adjusted dynamic virtual surplus]
\label{thm:multiagent_virtual_surplus}
Under the maintained common-prior, independence, and scalar-envelope
assumptions, suppose \(V_{i1}\) is absolutely continuous for every
\(i\in N\). Then every multi-agent dynamically incentive-compatible
mechanism \(\psi\) satisfies
\begin{equation}
\label{eq:multiagent_virtual_surplus_representation}
\begin{aligned}
\Pi^P(\psi)
={}&
\int
\Bigg[
\sum_{t=1}^{T}
\delta^{t-1}
\Bigg\{
v_t(q_t,\theta_t)
+
\sum_{i=1}^n
\Big[
u_{it}(\eta_{it},q_{it},\theta_t)
\\
&\qquad\qquad
-
\Lambda_{i1}(\eta_{i1})
I_{i,1,t}(\eta_i^t)
u_{it,\eta_i}(\eta_{it},q_{it},\theta_t)
\Big]
\Bigg\}
-
\sum_{i=1}^n
\Lambda_{i1}(\eta_{i1})
\kappa_{i1}(\eta_{i1})
\Bigg]
\,d\mathbb P(\cdot;\psi)
\\
&-
\sum_{i=1}^n
V_{i1}(\underline\eta_{i1}).
\end{aligned}
\end{equation}
\end{theorem}
\begin{proof}
Under the common prior and quasilinear transfers, the designer's expected
payoff equals expected joint gross surplus minus the sum of the agents'
initial expected continuation utilities:
\[
\begin{aligned}
\Pi^P
={}&
\int
\sum_{t=1}^{T}
\delta^{t-1}
\left(
v_t(q_t,\theta_t)
+
\sum_{i=1}^n u_{it}(\eta_{it},q_{it},\theta_t)
\right)
\,d\mathbb P(\cdot;\psi)
\\
&-
\sum_{i=1}^n
\int_{E_{i1}}V_{i1}(\eta_{i1})\,dF_{i1}(\eta_{i1}).
\end{aligned}
\]
For each agent, the scalar-envelope representation in Section~4.1 and the
verification wedge imply, almost everywhere,
\[
\frac{dV_{i1}(\eta_{i1})}{d\eta_{i1}}
=
\int
\sum_{t=1}^{T}
\delta^{t-1}
I_{i,1,t}(\widetilde\eta_i^t)
u_{it,\eta_i}
(\widetilde\eta_{it},\widetilde q_{it},\widetilde\theta_t)
\,d\mathbb P(\cdot\mid\eta_{i1};\psi)
+
\kappa_{i1}(\eta_{i1}).
\]
Absolute continuity and integration by parts therefore give
\[
\begin{aligned}
&\int_{E_{i1}}
V_{i1}(\eta_{i1})\,dF_{i1}(\eta_{i1})
=
V_{i1}(\underline\eta_{i1})
\\
&\quad+
\int
\Bigg[
\sum_{t=1}^{T}
\delta^{t-1}
\Lambda_{i1}(\eta_{i1})
I_{i,1,t}(\eta_i^t)
u_{it,\eta_i}(\eta_{it},q_{it},\theta_t)
\\
&\hspace{9em}+
\Lambda_{i1}(\eta_{i1})\kappa_{i1}(\eta_{i1})
\Bigg]
\,d\mathbb P(\cdot;\psi).
\end{aligned}
\]
The common prior makes the agentwise continuation expectations and the
designer's transfer accounting compatible with the same joint law.
Summing over agents proves
\eqref{eq:multiagent_virtual_surplus_representation}.
\end{proof}

Define the period-\(t\) \textbf{multi-agent dynamic virtual surplus} by
\begin{equation}
\label{eq:multiagent_dynamic_virtual_surplus}
\begin{aligned}
\Phi_t^N(\eta^t,q_t,\theta_t)
:=
&
\sum_{i=1}^n
\Big[
u_{it}(\eta_{it},q_{it},\theta_t)
-
\Lambda_{i1}(\eta_{i1})
I_{i,1,t}(\eta_i^t)
u_{it,\eta_i}(\eta_{it},q_{it},\theta_t)
\Big]
\\
&
+
v_t(q_t,\theta_t),
\end{aligned}
\end{equation}
where \(\eta^t=(\eta_1^t,\ldots,\eta_n^t)\). Then
\[
\Pi^P
=
\int
\left[
\sum_{t=1}^{T}
\delta^{t-1}
\Phi_t^N(\eta^t,q_t,\theta_t)
-
\sum_{i=1}^n
\Lambda_{i1}(\eta_{i1})\kappa_{i1}(\eta_{i1})
\right]
\,d\mathbb P(\cdot;\psi)
-
\sum_{i=1}^n V_{i1}(\underline\eta_{i1}).
\]

The three margins of Section~4 carry over agent by agent: informational
rents enter through dynamic virtual surplus, verification rents through the
verification wedges, and hidden-action implementation through the nested
obedience conditions and the implementable mechanism set. Hidden action does
not add a separate term to virtual surplus; it restricts which allocations
can be implemented.

By the agentwise counterpart of Proposition~\ref{prop:wedge_sign},
\(\kappa_{i1}(\eta_{i1})\geq0\) almost everywhere. Together with initial
participation, this gives
\begin{equation}
\label{eq:multiagent_virtual_upper_bound}
\Pi^P
\leq
\int
\sum_{t=1}^{T}
\delta^{t-1}
\Phi_t^N(\eta^t,q_t,\theta_t)
\,d\mathbb P(\cdot;\psi).
\end{equation}

Let \(\mathcal Q^F\) denote the technologically feasible reduced-form
allocation processes, without imposing incentive, informational, or
post-report action-implementation restrictions. The relaxed virtual-surplus
candidate is
\begin{equation}
\label{eq:multiagent_relaxed_problem}
q^\circ
\in
\arg\max_{q\in\mathcal Q^F}
\int
\sum_{t=1}^{T}
\delta^{t-1}
\Phi_t^N(\eta^t,q_t,\theta_t)
\,d\mathbb P(\cdot;q).
\end{equation}
As in Section~4.2, this candidate must satisfy the agentwise nested
report--action incentive conditions and the remaining feasibility and
information restrictions before it can solve the full problem. Envelope
compatibility determines whether the virtual-surplus upper bound is attained.

\begin{corollary}[Tightness of the multi-agent virtual-surplus relaxation]
\label{cor:multiagent_relaxation_tightness}
Let \(q^\circ\) solve \eqref{eq:multiagent_relaxed_problem}. Suppose there
exists a technologically and action-feasible direct mechanism \(\psi^\circ\)
implementing \(q^\circ\) that is multi-agent dynamically incentive
compatible under Section~6.2, is envelope compatible for every agent, and
satisfies
\[
V_{i1}(\underline\eta_{i1})=0,
\qquad i=1,\ldots,n.
\]
Then \(\psi^\circ\) solves the full multi-agent \(4B\) problem.
\end{corollary}

\begin{proof}
The upper bound \eqref{eq:multiagent_virtual_upper_bound} applies to every
feasible dynamically incentive-compatible mechanism satisfying initial
participation. The allocation \(q^\circ\) maximizes this bound. Agentwise
envelope compatibility eliminates the verification wedges, while binding
participation of each lowest initial type eliminates the remaining initial
rents. Thus \(\psi^\circ\) attains the upper bound and is optimal.
\end{proof}

Corollary~\ref{cor:multiagent_relaxation_tightness} gives the solution
procedure of Section~4.2: formulate the full \(4B\) problem, obtain a relaxed
virtual-surplus candidate, and verify report--action incentives, feasibility,
and rent tightness for every agent. If hidden action is removed by direct
assignment while verification remains, the problem reduces to
\(3B^{\mathrm{HI+V}}\). If reporting is unrestricted while hidden action
remains, it reduces to \(3B^{\mathrm{HI+HA}}\). Removing both additional
restrictions yields \(2B^{\mathrm{HI}}\). Removing private information while
retaining hidden action yields \(2B^{\mathrm{HA}}\), and removing both
private information and the action-implementation friction yields \(1B\).

\begin{corollary}[Branched reductions of the multi-agent problem]
\label{cor:multiagent_nested_reductions}
Under the maintained assumptions, direct assignment for every agent reduces
the full problem to \(3B^{\mathrm{HI+V}}\); unrestricted reporting for every
agent reduces it to \(3B^{\mathrm{HI+HA}}\); imposing both reductions yields
\(2B^{\mathrm{HI}}\); removing private information while retaining hidden
action yields \(2B^{\mathrm{HA}}\); and removing both private information and
the action-implementation friction yields \(1B\).
\end{corollary}

\begin{proof}
Each reduction removes the corresponding informational, verification, or
action-implementation restriction from the joint mechanism problem.
\end{proof}

To parallel Section~4.3, suppose in addition that gross utility is type
linear,
\[
u_{it}(\eta_{it},q_{it},\theta_t)
=
\eta_{it}g_{it}(q_{it},\theta_t),
\qquad i=1,\ldots,n,
\]
where \(g_{it}\) is increasing and differentiable in \(q_{it}\). Since
\(u_{it,\eta_i}=g_{it}\), define the
\textbf{agent-specific dynamic virtual type} by
\begin{equation}
\label{eq:multiagent_dynamic_virtual_type}
\eta_{it}^{\mathrm{vir}}(\eta_i^t)
:=
\eta_{it}
-
\Lambda_{i1}(\eta_{i1})
I_{i,1,t}(\eta_i^t).
\end{equation}
Then
\[
\Phi_t^N(\eta^t,q_t,\theta_t)
=
\sum_{i=1}^n
\eta_{it}^{\mathrm{vir}}(\eta_i^t)
g_{it}(q_{it},\theta_t)
+
v_t(q_t,\theta_t).
\]

For the following pointwise characterization, assume a historywise separable
technology with no intertemporal allocation constraints and a law of
\((\eta^t,\theta_t)\) that is exogenous to the allocation process. The
relaxed process may be conditioned on \((\eta^t,\theta_t)\); when
\(\theta_t\) is unobserved, this is a
full-information relaxation whose implementability must be checked. Under
these restrictions, the relaxed candidate is characterized pointwise by
\[
q_t^\circ(\eta^t,\theta_t)
\in
\arg\max_{q_t\in\mathcal Q_t}
\left\{
\sum_{i=1}^n
\eta_{it}^{\mathrm{vir}}(\eta_i^t)
g_{it}(q_{it},\theta_t)
+
v_t(q_t,\theta_t)
\right\}.
\]
When no joint feasibility constraint binds and the solution is interior,
\begin{equation}
\label{eq:multiagent_relaxed_foc}
\eta_{it}^{\mathrm{vir}}(\eta_i^t)
g_{it,q}(q_{it}^\circ,\theta_t)
+
v_{t,q_i}(q_t^\circ,\theta_t)
=
0,
\qquad i=1,\ldots,n.
\end{equation}

When technologies, including hidden-action implementation technologies,
separate across agents, the distortion results of
Proposition~\ref{prop:distortion_dynamics} apply agent by agent, with
Corollary~\ref{cor:type_linear_distortion} giving the corresponding
linear-quadratic benchmark. Joint feasibility adds a distinct allocation
margin. For example, suppose the only within-period allocation restrictions
are nonnegativity and shared capacity,
\[
q_{it}\geq0,\qquad
\sum_{i=1}^n q_{it}
\leq
\overline q_t,\qquad\overline q_t>0.
\]

\begin{proposition}[Dynamic virtual allocation under shared capacity]
\label{prop:multiagent_distortions}
Suppose the preceding pointwise assumptions hold, \(g_{it,q}>0\),
and the relaxed period objective \(\Phi_t^N\) is strictly concave
in \(q_t\). At every positive component of the relaxed solution,
there exists \(\xi_t\geq0\) such that
\begin{equation}
\label{eq:multiagent_capacity_foc}
\eta_{it}^{\mathrm{vir}}(\eta_i^t)
g_{it,q}(q_{it}^\circ,\theta_t)
+
v_{t,q_i}(q_t^\circ,\theta_t)
=
\xi_t,
\qquad i:q_{it}^\circ>0,
\end{equation}
with
\begin{equation}
\label{eq:multiagent_capacity_complementarity}
\xi_t
\left(
\overline q_t
-
\sum_{i=1}^n q_{it}^\circ
\right)
=
0.
\end{equation}
Thus common capacity is allocated according to marginal dynamic virtual value.
\end{proposition}

\begin{proof}
The result follows from the Kuhn--Tucker conditions for the strictly concave
relaxed period problem.
\end{proof}

Proposition~\ref{prop:multiagent_distortions} separates the informational
and scarcity margins. Dynamic screening replaces each current private type
by its dynamic virtual counterpart, while common scarcity compares marginal
dynamic virtual values across agents. Hidden-action implementation remains a
separate restriction on whether the relaxed allocation can be induced. Two
agents with the same current type may receive different allocations when
their initial rent weights or type persistence differ. If the capacity
constraint is slack and the designer's gross payoff separates across agents,
the relaxed allocation reduces to the corresponding agentwise benchmark of
Section~4.

The independence benchmark therefore inherits the single-agent hidden
information, verification, and hidden-action structure agent by agent,
while joint feasibility links the resulting allocations. Correlated
information and coupled incentives introduce further interactions, which
are considered in Section~6.4.

\subsection{Peer Information, Coupled Incentives, and Continuation Externalities}

Section~6.3 imposed independence to obtain the agentwise dynamic-rent
representation. We now allow correlated private information and coupled
incentives. The agentwise nested condition of Section~6.2 remains the
implementability criterion, but its components become strategically linked:
peer information changes the value of false-report deviations, while coupled
rewards and continuation policies change truthful-report action incentives.

Fix a feasible public state \(z_t\) and agent \(i\). For each current type
\(\tau_{it}\), let
\(b_{it}^{-i}(\cdot\mid z_t;\tau_{it})\)
denote agent \(i\)'s subjective conditional distribution of the
contemporaneous reports of the other agents under truthful play. Reports are
submitted simultaneously, before any agent observes the other agents' current
reports. Let \(\operatorname{sc}_{it}(p,x)\) be a strictly proper scoring rule
on the relevant family of predictive distributions. For a feasible report
\(\widehat\tau_{it}\in C_{it}(\tau_{it})\), define the
\textbf{peer-information divergence}
\begin{equation}
\label{eq:peer_information_divergence}
\begin{aligned}
\Delta_{it}^{\mathrm{peer}}
(\tau_{it},\widehat\tau_{it}\mid z_t)
:=
\int
\Big[
&\operatorname{sc}_{it}
\bigl(b_{it}^{-i}(\cdot\mid z_t;\tau_{it}),x\bigr)
\\
&-\operatorname{sc}_{it}
\bigl(b_{it}^{-i}(\cdot\mid z_t;\widehat\tau_{it}),x\bigr)
\Big]
\,db_{it}^{-i}(x\mid z_t;\tau_{it}).
\end{aligned}
\end{equation}
Strict propriety implies \(\Delta_{it}^{\mathrm{peer}}\geq0\), with strict
inequality whenever the two predictive distributions are distinct.

Consider an additive peer-score reward with scale \(w_{it}\geq0\),
based on simultaneous reports and independent of agent \(i\)'s
realized action. We impose \textbf{continuation neutrality}:
introducing the score leaves information, continuation strategies,
and conditional continuation payoffs unchanged relative to the same
mechanism without peer scoring, apart from the current scoring transfer.
In particular, the score must reveal no additional information. This
holds, for example, if the report profile is disclosed after current
actions regardless of whether peer scoring is used; reports need not
be public when submitted. Under quasilinear payoffs, the score then
leaves post-report action incentives unchanged.

Let \(V_{it}^{0}(\tau_{it},z_t)\) and
\(\widehat V_{it}^{0}(\tau_{it},z_t;\widehat\tau_{it})\)
denote the truthful and deviation continuation values, respectively,
excluding agent \(i\)'s peer-score transfer while holding all other
mechanism rules fixed. The deviation value incorporates optimal
current post-report action reoptimization. Define the
\textbf{critical peer-score scale}
\begin{equation}
\label{eq:peer_score_threshold}
w_{it}^*(z_t)
:=
\sup_{\substack{
\tau_{it},\,
\widehat\tau_{it}\in C_{it}(\tau_{it})\\
\widehat\tau_{it}\neq\tau_{it}}}
\frac{
\left[
\widehat V_{it}^{0}
(\tau_{it},z_t;\widehat\tau_{it})
-
V_{it}^{0}(\tau_{it},z_t)
\right]_+
}{
\Delta_{it}^{\mathrm{peer}}
(\tau_{it},\widehat\tau_{it}\mid z_t)
}.
\end{equation}
The ratio is interpreted as \(+\infty\) when its numerator is positive
and its denominator is zero, and as zero when both are zero.

\begin{theorem}[Dynamic peer discipline]
\label{thm:dynamic_peer_discipline}
Suppose the finite-horizon or continuity-at-infinity conditions of
Section~3.4 hold, agents report simultaneously, and peer scoring is
action independent and continuation neutral. Then:
\begin{enumerate}
\item[(i)] If \(w_{it}^*(z_t)<\infty\), any feasible scale
\(w_{it}\geq w_{it}^*(z_t)\) makes every feasible one-stage
report--action deviation involving a false report by agent \(i\)
unprofitable at \(z_t\).

\item[(ii)] If part~(i) holds after every state for every
agent and truthful-report post-report obedience holds at every
relevant information set, the mechanism is multi-agent dynamically
incentive compatible.
\end{enumerate}
\end{theorem}

\begin{proof}
Fix a true type \(\tau_{it}\) and a feasible false report
\(\widehat\tau_{it}\). Simultaneous reporting ensures that agent
\(i\)'s report does not change the distribution of the other agents'
current reports. Action independence and continuation neutrality
leave the post-report action problem and subsequent continuation
payoffs unchanged. Hence the gain from the false report, including
optimal current action reoptimization, is
\[
\widehat V_{it}^{0}
(\tau_{it},z_t;\widehat\tau_{it})
-
V_{it}^{0}(\tau_{it},z_t)
-
w_{it}\Delta_{it}^{\mathrm{peer}}
(\tau_{it},\widehat\tau_{it}\mid z_t).
\]
By \eqref{eq:peer_score_threshold}, this gain is nonpositive whenever
\(w_{it}\geq w_{it}^*(z_t)\), proving part~(i).

If these comparisons hold after every state and obedience
holds on the truthful branch, every agent satisfies the nested
reporting and action conditions of Section~6.2. Theorem~
\ref{thm:multiagent_recursive} then establishes part~(ii).
\end{proof}

Theorem~\ref{thm:dynamic_peer_discipline} shows how peer information
can discipline false reporting without altering post-report action
incentives. This separation relies on continuation neutrality:
additional information revealed by scoring or changes in subsequent
behavior may also affect continuation incentives.

For the normalized quadratic score on a finite peer-report space,
\[
\Delta_{it}^{\mathrm{peer}}
(\tau_{it},\widehat\tau_{it}\mid z_t)
=
\left\|
b_{it}^{-i}(\cdot\mid z_t;\tau_{it})
-
b_{it}^{-i}(\cdot\mid z_t;\widehat\tau_{it})
\right\|_2^2.
\]
We say that \textbf{predictive separation} holds if every feasible
false report induces a subjective prediction of other agents' reports
distinct from that induced by the true type. On finite type and
report spaces, this condition ensures a strictly positive minimum
peer-information divergence. With bounded remaining report--action
deviation gains, it follows that \(w_{it}^*(z_t)<\infty\). Partial
verification restricts the required comparisons to feasible reports,
while truthful-report obedience remains a separate condition.

We next provide a primitive sufficient condition for predictive
separation based on affiliation. Consider ordered scalar current
types \(\eta_{it}\in E_{it}^F\), where each \(E_{it}^F\) is finite,
and a strictly positive joint probability mass function
\(f_t(\eta_t\mid z_t)\) on
\(E_t^F=\prod_{i=1}^n E_{it}^F\). For type profiles \(\eta_t\)
and \(\eta_t'\), let \(\eta_t\vee\eta_t'\) and
\(\eta_t\wedge\eta_t'\) denote their coordinatewise maximum and
minimum. Current types are \textbf{strictly affiliated} conditional
on \(z_t\) if
\begin{equation}
\label{eq:affiliation}
f_t(\eta_t\vee\eta_t'\mid z_t)
f_t(\eta_t\wedge\eta_t'\mid z_t)
>
f_t(\eta_t\mid z_t)
f_t(\eta_t'\mid z_t)
\end{equation}
whenever \(\eta_t\) and \(\eta_t'\) are incomparable. Under strict positivity, this is strict supermodularity of
\(\log f_t(\cdot\mid z_t)\) on incomparable profiles. In a smooth
continuous environment, a corresponding sufficient condition is $
\frac{\partial^2\log f_t(\eta_t\mid z_t)}
{\partial\eta_{it}\partial\eta_{jt}}>0$ for $i\neq j.$ This is the increasing-differences condition underlying the familiar
single-crossing property.

To connect affiliation to agents' reporting incentives, assume that
their subjective peer predictions agree with the conditional
distributions induced by \(f_t\) under truthful reporting:
\[
b_{it}^{-i}(\eta_{-i,t}\mid z_t;\eta_{it})
=
\frac{f_t(\eta_{it},\eta_{-i,t}\mid z_t)}
{\displaystyle\sum_{x_{-i}\in E_{-i,t}^F}
f_t(\eta_{it},x_{-i}\mid z_t)}.
\]
This assumption concerns current peer predictions and does not require
a common prior over the entire dynamic environment.

\begin{corollary}[Peer discipline under strictly affiliated types]
\label{cor:affiliated_peer_discipline}
Suppose the finite-support quadratic-score benchmark has bounded remaining
report--action gains, the types are strictly affiliated, and each agent has
nondegenerate peer-type support. If the subjective peer predictions satisfy
the preceding conditional-distribution restriction, then
\(w_{it}^*(z_t)<\infty\) for every agent \(i\), and a sufficiently large
feasible peer-score scale disciplines every feasible false report.
\end{corollary}

\begin{proof}
Fix \(i\) and \(\eta_{it}'>\eta_{it}\). Strict affiliation makes the
likelihood ratio
\(f_t(\eta_{it}',\eta_{-i,t}\mid z_t)/
f_t(\eta_{it},\eta_{-i,t}\mid z_t)\)
strictly increasing across comparable, distinct peer-type profiles.
Nondegenerate peer-type support therefore rules out a constant ratio, so the
conditional peer-type distributions at \(\eta_{it}'\) and \(\eta_{it}\)
are distinct. The maintained agreement transfers this separation to the
agent's subjective predictions of truthful peer reports.

Strict propriety gives positive divergence for every feasible false report.
The finite report set gives a positive minimum divergence, and bounded
remaining deviation gains imply \(w_{it}^*(z_t)<\infty\). Theorem~\ref{thm:dynamic_peer_discipline}
completes the proof.
\end{proof}

Corollary~\ref{cor:affiliated_peer_discipline} makes strict affiliation a
sufficient, but not necessary, condition for finite-support peer discipline
when agents' subjective peer predictions agree with the affiliated law.
This conclusion need not extend to continuous types. Nearby types may induce
arbitrarily similar predictions, and a smooth quadratic score imposes a
second-order loss against a local false report.

In the smooth scalar environments of Sections~4 and~6.3, arbitrarily nearby
downward reports are feasible under the withholding structure. To study a
finite-dimensional quadratic score, let \(x\) now denote a finite-valued
statistic of the co-agents' truthful reports, and suppose
\(\eta\mapsto b_{it}^{-i}(\cdot\mid z_t;\eta)\)
is differentiable as a probability vector for this statistic. For a nearby
downward report \(\widehat\eta=\eta-h\),
\begin{equation}
\label{eq:local_peer_divergence}
\Delta_{it}^{\mathrm{peer}}
(\eta,\eta-h\mid z_t)
=
h^2
\left\|
\frac{\partial
b_{it}^{-i}(\cdot\mid z_t;\eta)}
{\partial\eta}
\right\|_2^2
+
o(h^2)
\qquad
\text{as }h\downarrow0.
\end{equation}
Thus the peer-score loss from a local false report is second order.

\begin{proposition}[Continuous-type limit of peer discipline]
\label{prop:continuous_peer_limit}
Suppose that, for a feasible downward local deviation
\(\widehat\eta=\eta-h\),
\[
\widehat V_{it}^{0}(\eta,z_t;\eta-h)
-
V_{it}^{0}(\eta,z_t)
=
mh+o(h)
\]
for some \(m>0\). Then \(w_{it}^*(z_t)=+\infty\). Consequently, no finite
quadratic peer-score scale can eliminate this first-order local reporting
gain by itself.
\end{proposition}

\begin{proof}
The non-peer deviation gain is of order \(h\), whereas
\eqref{eq:local_peer_divergence} is of order at most \(h^2\). If the
probability-vector derivative is nonzero, the ratio in
\eqref{eq:peer_score_threshold} is of order \(1/h\) and diverges. If the
derivative is zero, the divergence is \(o(h^2)\), and the ratio diverges
whenever its denominator is positive; if the denominator vanishes, the
convention in \eqref{eq:peer_score_threshold} applies.
\end{proof}

Proposition~\ref{prop:continuous_peer_limit} identifies the boundary of
finite-support peer discipline. Affiliation can make peer information
informative and can discipline discrete or nonlocal deviations, but it does
not alter the local order comparison. The verification and dynamic
implementability conditions of Sections~4 and~6.3 therefore remain relevant
even when agents' types are affiliated.

Peer information changes false-report incentives. We next consider the
truthful-report action margin. At agent \(i\)'s post-report information
history \(h_{it}\), let
\(a_{it}^{\mathrm{ob}}:=\alpha_{it}(s_{it})\)
denote its prescribed action and
\(a_t^{\mathrm{ob}}=(a_{1t}^{\mathrm{ob}},\ldots,a_{nt}^{\mathrm{ob}})\)
the prescribed profile. As in Section~6.2, agent \(i\) does not generally
observe \(a_{-i,t}^{\mathrm{ob}}\). Holding the remaining payoff components
fixed, suppose the part of its current gross payoff depending on its own
action is
\[
-\frac{\gamma_{it}}{2}(a_{it}-a_{it}^0)^2,
\qquad \gamma_{it}>0,
\]
where \(a_{it}^0\) is its baseline action.

Let \(\overline R_{it}(a_i\mid h_{it})\) denote the expected current reward
when agent \(i\) chooses \(a_i\), conditional on its post-report
information and the other agents' equilibrium strategies. This expectation
integrates over the unobserved actions of other agents and the public outcome
and other current observables on which its reward may depend. Let
\(\overline V_{i,t+1}(a_i\mid h_{it})\) denote the corresponding expected
truthful-and-obedient continuation value from date \(t+1\), accounting for
the action-dependent public outcomes, monitoring evidence, and physical
transitions. Both expectations are evaluated under the conditional law
\(\mathbb P_{it}(\cdot\mid h_{it};\Gamma^D,a_i)\); neither requires the
mechanism to observe the actual actions.

\begin{proposition}[Coupled current and continuation incentives]
\label{prop:coupled_continuation_obedience}
Suppose the own-action set is an interval, the conditional expected current
reward and continuation value are differentiable in \(a_i\), and each
post-report action problem is strictly concave in its own action. Then:
\begin{enumerate}
\item[(i)] \textbf{Sequential Bayesian obedience.}
At every post-report information history \(h_{it}\) following a truthful
report, an interior prescribed action \(a_{it}^{\mathrm{ob}}\) is optimal if
and only if
\begin{equation}
\label{eq:coupled_obedience_foc}
\gamma_{it}(a_{it}^{\mathrm{ob}}-a_{it}^0)
=
\left.
\frac{\partial\overline R_{it}(a_i\mid h_{it})}{\partial a_i}
\right|_{a_i=a_{it}^{\mathrm{ob}}}
+
\delta
\left.
\frac{\partial\overline V_{i,t+1}(a_i\mid h_{it})}{\partial a_i}
\right|_{a_i=a_{it}^{\mathrm{ob}}}.
\end{equation}

\item[(ii)] \textbf{Bounded current incentives.}
Holding the continuation rule fixed, define the marginal current reward
required at \(h_{it}\) by
\begin{equation}
\label{eq:required_current_incentive}
\iota_{it}^{\mathrm{req}}(h_{it})
:=
\gamma_{it}(a_{it}^{\mathrm{ob}}-a_{it}^0)
-
\delta
\left.
\frac{\partial\overline V_{i,t+1}(a_i\mid h_{it})}{\partial a_i}
\right|_{a_i=a_{it}^{\mathrm{ob}}}.
\end{equation}
If the admissible reward system can jointly deliver any collection of
conditional own-action marginal expected-reward incentives in
\([ -\overline\iota_{it},\overline\iota_{it}]\) across agents and feasible
post-report information sets, then the prescribed interior action profile
is supportable through current rewards if and only if
\[
|\iota_{it}^{\mathrm{req}}(h_{it})|
\leq\overline\iota_{it},
\qquad
\text{for every }i\in N\text{ and feasible }h_{it}.
\]
\end{enumerate}
\end{proposition}

\begin{proof}
Conditional on \(h_{it}\), agent \(i\)'s expected payoff from \(a_i\)
contains the own-action terms
\[
-\frac{\gamma_{it}}{2}(a_i-a_{it}^0)^2
+
\overline R_{it}(a_i\mid h_{it})
+
\delta\overline V_{i,t+1}(a_i\mid h_{it}).
\]
At an interior optimum its derivative vanishes, yielding
\eqref{eq:coupled_obedience_foc}. Strict concavity makes the condition
sufficient.

Rearranging \eqref{eq:coupled_obedience_foc} gives
\(\iota_{it}^{\mathrm{req}}(h_{it})\). Under the maintained joint
flexibility of reward instruments, the required conditional marginal
incentives are attainable if and only if they lie within the stated bounds
at every information set. Strict concavity then supplies optimality of each
prescribed action.
\end{proof}

Equation~\eqref{eq:coupled_obedience_foc} includes both current and
continuation incentive channels. In the observable-action benchmark, when
opponents' actions are known at the relevant decision stage, the conditional
expected-reward and continuation terms can be written as functions
\(\overline R_{it}(a_t)\) and \(\overline V_{i,t+1}(a_t)\) evaluated at
\(a_{-i,t}^{\mathrm{ob}}\); the resulting first-order condition is the
usual action-profile Nash condition. In the general environment, the agent
instead optimizes against a conditional distribution of its opponents'
unobserved actions. Coupled rewards can make current expected marginal
rewards depend on peers' reports and observable behavior, as in Bergemann,
Koh, and Morris \citeyearpar{BergemannKohMorris2026}. The dynamic continuation
term captures changes in future authorization, rewards, tasks, and deployment
induced by current actions through public outcomes, monitoring, and
transitions. Thus agentwise obedience can be strategically coupled even when
primitive current payoffs are separable.

The componentwise bound in part~(ii) is a benchmark for sufficiently flexible
reward instruments. If the required marginal incentives are not jointly
attainable through a single reward profile based on publicly contractible
variables, componentwise bounds alone are not sufficient. The designer must
then alter rewards, authorization, continuation promises, or the target
actions, as in the hidden-action problem of Section~4.

Sections~6.3 and~6.4 therefore describe two multi-agent regimes within the
same incentive framework. Under the independent-type benchmark of
Section~6.3, dynamic information rents and verification wedges aggregate
agent by agent, while hidden-action implementation remains embedded in each
agent's nested obedience condition and joint feasibility links allocations.
With correlated information and coupled continuation mechanisms, the same
agentwise nested conditions remain exact but become strategically linked:
peer information changes false-report incentives, while coupled expected
rewards and continuation opportunities change truthful-report action
incentives. The interaction is a coupling of agentwise equilibrium
conditions, not an additional aggregate incentive constraint.

\section{AI-Relevant Features Beyond Standard Dynamic Contracting}

The preceding sections combine established dynamic mechanism-design and
contracting foundations with a set of information, incentive, and control
margins that are particularly natural for AI agents. The comparison is useful
in three layers. Some analytical structure is \textbf{inherited} from standard
dynamic screening and dynamic contracting; some familiar conclusions are
\textbf{modified} by partial verification and post-report action discretion;
and additional design margins arise from technologically enforceable
authorization, dynamic monitoring and information use, and interaction among
multiple agents.

These distinctions do not define a separate theory for artificial agents.
The individual ingredients may also arise with human agents, firms, experts,
or organizations. Their joint presence is especially natural for AI systems:
private information may concern objectives, execution capability, beliefs, or
information capability; communication need not determine subsequent behavior;
actions may remain privately chosen after communication; and access to tools,
communication channels, external systems, computational resources, or
deployment can itself be controlled by the mechanism.
Table~\ref{tab:ai_specific_margins} summarizes the resulting comparison.

\begin{table}[!htbp]
\centering
\small
\begin{tabularx}{\textwidth}
{>{\raggedright\arraybackslash}p{0.22\textwidth}
 >{\raggedright\arraybackslash}X
 >{\raggedright\arraybackslash}X}
\toprule
Feature & Economic role & Relation to the standard benchmark \\
\midrule

Evolving private information
&
The agent's payoff-relevant private information may evolve over time and
affect future incentives.
&
Inherited from dynamic mechanism design. Continuation values, impulse
responses, envelope methods, information rents, and dynamic virtual surplus
provide the basic analytical tools.
\\[0.35em]

Private execution and information capability
&
The agent may privately know which actions it can execute and which
information it can obtain, in addition to its behavioral objective and
subjective beliefs.
&
Accommodated by the richer private type. These components may themselves
evolve through the type-transition process.
\\[0.35em]

Partial and one-sided verification
&
The set of claims that can be substantiated may depend on the true type, and
concealment need not be symmetric with counterfeiting.
&
Modifies the reporting domain. In the scalar withholding environment,
truth may be a boundary report, so the standard envelope implication can fail
and a verification wedge can arise.
\\[0.35em]

Post-report action discretion
&
After reporting, the agent may receive additional private information and
choose its own action rather than having the payoff-relevant action directly
assigned.
&
Introduces a hidden-action implementation problem when the realized action is
not publicly observed. Each report must be evaluated together with the true
type's optimal post-report action, so honesty and obedience form a nested
adverse-selection and moral-hazard problem.
\\[0.35em]

Technological authorization
&
What the agent can execute need not coincide with what the mechanism permits.
Access to tools, communication, external systems, or deployment can therefore
be controlled directly.
&
Adds a control instrument distinct from execution capability and from the
recommended action.
\\[0.35em]

Dynamic monitoring and information use
&
A primitive monitoring technology generates evidence, potentially as a
function of latent behavior, while the mechanism chooses how that evidence
enters the public continuation state.
&
Adds an information-design margin. Monitoring capacity determines the
informational frontier; information use determines which feasible public
experiment is employed and may itself affect hidden-action and reporting
incentives.
\\[0.35em]

Multi-agent information and incentives
&
Agents may have correlated private information, and their reports, privately
chosen actions, authorization, observable consequences, rewards, and
continuation opportunities may interact.
&
Adds peer-information, coupled-incentive, joint-feasibility, and cross-agent
continuation channels while retaining agentwise unilateral incentive
constraints.
\\
\bottomrule
\end{tabularx}
\caption{AI-relevant features relative to standard dynamic contracting.}
\label{tab:ai_specific_margins}
\end{table}

\subsection{Inherited Structure from Dynamic Mechanism Design}

The dynamic structure builds on standard mechanism design and contract
theory, including the integrated reporting--action incentive framework and
branched contracting comparison of Tian
\citeyearpar{TianGeneralDynamic2026}. Private information may evolve over time, privately chosen
actions may affect future states, and continuation values summarize the
incentives generated by future mechanism choices. A sufficient public state
can therefore represent the continuation problem recursively, and the
revelation principle reduces the analysis to direct mechanisms subject to the
relevant reporting and action-incentive constraints.

The scalar environment of Section~4 makes the connection to conventional
dynamic screening especially transparent. When reporting is unrestricted and
the payoff-relevant action is directly assigned, the post-report action
problem and the verification boundary disappear. Dynamic envelope methods then
relate continuation utility to the propagation of private information through
impulse responses, yielding the familiar information-rent and virtual-surplus
representation of Pavan, Segal, and Toikka
\citeyearpar{PavanSegalToikka2014}. This is the
\(2B^{\mathrm{HI}}\) branch of the benchmark map.

The hidden-action branch is equally familiar in economic logic. When the
relevant private type is common knowledge but the agent privately chooses an
action that cannot be contracted on directly, the remaining problem is one of
moral hazard: observable outcomes, rewards, authorization, and continuation
promises must provide the incentives needed to implement desired behavior.
This is the \(2B^{\mathrm{HA}}\) branch. The broader contracting literature has
long studied environments combining adverse selection and moral hazard,
including Laffont and Tirole \citeyearpar{LaffontTirole1986}, Picard
\citeyearpar{Picard1987}, Guesnerie, Picard, and Rey
\citeyearpar{GuesneriePicardRey1989}, and Meng and Tian
\citeyearpar{MengTian2013}; dynamic environments combining the two frictions
include Es\H{o} and Szentes \citeyearpar{EsoSzentes2017}.

The present framework therefore builds on, rather than replaces, the
dynamic-screening and hidden-action strands and their joint treatment in
\citet{TianGeneralDynamic2026}. It adapts that integrated contracting
approach to the AI-agent timing and adds one-sided verification and
AI-relevant control instruments. The verification dimension creates
additional branches beyond the first-best, two-second-best, and third-best
comparison of the general framework. The two second-best branches are conceptually distinct:
eliminating hidden action leaves the \(2B^{\mathrm{HI}}\) dynamic-screening
problem, whereas eliminating private information leaves the
\(2B^{\mathrm{HA}}\) moral-hazard problem. First best requires both frictions
to disappear. Adding one-sided verification to the former gives
\(3B^{\mathrm{HI+V}}\), while adding hidden action to dynamic screening gives
\(3B^{\mathrm{HI+HA}}\); the full \(4B\) environment contains hidden
information, one-sided verification, and hidden-action implementation.

The recursive foundation remains common across these cases. Verification
restricts the reports available to a true type; authorization restricts the
actions the mechanism permits; and post-report action discretion requires the
mechanism to account for what the true type will actually do after any
feasible report. These features change the relevant implementation and rent
conditions without changing the underlying dynamic-programming logic.

\subsection{Modified and Additional AI-Control Margins}

The first distinction is between \textbf{execution capability} and
\textbf{authorization}. Execution capability describes what the agent can
actually do, whereas authorization describes what the mechanism permits it to
do. The two need not coincide. This distinction matters dynamically because
authorization can depend on past reports, realized mechanism choices, public
outcomes, monitoring information, and other publicly observed variables, and
therefore becomes part of the continuation mechanism that shapes current
incentives.

Partial verification modifies the reporting margin. Under one-sided
verification, a type may be able to conceal or understate a verifiable
attribute without being able to counterfeit a higher one. Truthful reporting
can therefore lie at the boundary of the feasible report set rather than at an
interior point of an unrestricted reporting problem. Section~4 shows that
dynamic incentive compatibility then need not imply the ordinary two-sided
envelope equality. The verification wedge measures the resulting
continuation-rent margin, while local nonbinding verification restores the
standard envelope implication.

Post-report action discretion modifies a different part of the incentive
problem. Communication does not mechanically determine behavior. After a
report, the true type may observe additional private information and choose
among the actions that are both executable and authorized. When that action
is not publicly observed, this is a dynamic moral-hazard problem embedded
inside the reporting mechanism. A false report must therefore be evaluated
after allowing the true type to reoptimize its current action. Dynamic
incentive compatibility consequently requires honesty and obedience jointly:
feasible false reports are report--action deviations, while the truthful-report
comparison imposes sequential obedience.

These modifications clarify the three economic margins identified in
Section~4. The \textbf{informational transaction cost} is inherited from
dynamic screening and is governed by the initial rent weight and the impulse
responses through which private information remains relevant over time. The
\textbf{verification transaction cost} is the additional continuation rent
generated when one-sided reportability prevents the ordinary envelope
representation from being attained. The \textbf{action-implementation cost},
or obedience cost, is different: a latent post-report action restricts the set
of mechanisms capable of sustaining the desired behavior. It therefore
appears as a shadow cost of implementability rather than as a separate
additive rent term. These margins operate through different mathematical
channels even though they are jointly determined by the same mechanism.

This decomposition also explains why the dynamic virtual type summarizes only
the hidden-information distortion. The inverse hazard rate and impulse
response determine the propagated information-rent adjustment. Hidden action
does not generate another virtual-type term; it restricts the implementable
set through the nested obedience condition. One-sided verification operates
through yet another channel, the verification wedge. The dynamic virtual type,
verification wedge, and hidden-action implementation condition therefore
provide complementary rather than competing descriptions of the full
\(4B\) problem.

Dynamic monitoring introduces another design margin. Section~5 distinguishes
the primitive \textbf{monitoring technology} from the mechanism's
\textbf{information-use policy}. The informativeness of the primitive
experiment determines monitoring capacity, while the information-use policy
selects a public experiment from the Blackwell-feasible set. A latent action
may affect the distribution of raw monitoring evidence even though the
information-use policy cannot condition directly on that action. More
informative monitoring capacity weakly expands the designer's opportunity set
when information can be freely discarded, but more informative public use of
that capacity need not be optimal.

The reason is dynamic. Public monitoring information changes the continuation
state and may therefore affect future authorization, rewards, deployment, and
other mechanism choices. It can consequently change the continuation payoff
associated with current truthful behavior and current report--action
deviations. Information use can be optimized separately from current
incentive design only when the relevant truthful-versus-deviation
continuation-value differences are invariant to the information-use policy.
Otherwise, information use is itself an incentive instrument. With hidden
action, this channel is especially important because information use may
strengthen or weaken the continuation incentive generated by evidence about
privately chosen behavior without making that behavior directly observable.

The baseline analysis keeps information production fixed. Thus the private
information capability and primitive monitoring technology are taken as
given, while the use of generated monitoring evidence is endogenous.
Endogenizing information production would introduce the additional interaction
among learning, incentives, posterior evolution, and future feasibility
studied in Tian \citeyearpar{Tian2026}.

\subsection{Multi-Agent Interactions and Synthesis}

The multi-agent extension preserves the preceding incentive structure but
adds cross-agent information, feasibility, and continuation effects. Under
the independence benchmark of Section~6.3, the single-agent rent accounting
extends agent by agent: dynamic information rents and verification wedges
enter the multi-agent virtual-surplus representation separately for each
agent, while hidden-action implementation remains embedded in each agent's
nested obedience condition. Joint feasibility nevertheless links allocations
across agents. With a shared capacity constraint, for example, scarce control
is allocated according to marginal dynamic virtual values rather than current
private values alone.

Correlated private information creates a different channel. One agent's type
may predict the reports of the others, so peer information can discipline
false reporting. Section~6.4 shows that the relevant peer-score intensity is
determined by the report--action gain from a deviation relative to the
informational divergence generated by the peer score. Partial verification
matters because only feasible false reports need to be deterred, while
dynamics matter because the gain from a current deviation includes its effect
on continuation value.

The distinction between finite and continuous report spaces is important.
With a finite report set and sufficiently separated predictive distributions,
a finite peer-score scale may discipline reporting. With continuous types,
however, a smooth quadratic score generates only a second-order loss from a
nearby false report. If the corresponding deviation gain is first order, no
finite score scale can eliminate that local deviation. Peer information
therefore does not generally remove the verification and reporting frictions
identified in the single-agent analysis, nor does it by itself resolve the
hidden-action obedience problem.

Coupled incentives operate on the action margin. An agent's privately chosen
action may change the distribution of public outcomes or other observable
signals on which its expected current reward depends, and may also change its
future authorization and continuation value. The marginal incentive for
obedience can therefore have both an induced current-reward component and a
continuation-value component. The continuation mechanism may strengthen or
weaken the current expected-reward incentive needed to sustain a target
action, and the resulting agentwise obedience conditions may be strategically
coupled even when current primitive payoffs are separable. Direct dependence
of rewards on the action profile is only the observable-action special case;
the general formulation requires only that actions affect the distribution of
publicly contractible outcomes or signals.

The comparison can therefore be summarized in the three layers introduced at
the start of the section. \textbf{Inherited} elements include dynamic
revelation, continuation values, recursive incentives, impulse responses,
envelope methods, information rents, virtual surplus, moral-hazard obedience,
and distortion dynamics. \textbf{Modified} results arise because one-sided
verification changes the reporting boundary and envelope logic, while
post-report private information and hidden action require reports and
subsequent actions to be evaluated jointly. \textbf{Additional} design
margins arise from the distinction between execution capability and
authorization, dynamic monitoring and information use, joint feasibility,
peer information, coupled incentives, and cross-agent continuation effects.

The contribution is therefore not that hidden information, hidden action, or
the other individual ingredients are unique to artificial agents. It is that
these ingredients can be analyzed jointly in a dynamic mechanism in which
private information evolves, privately informed agents retain post-report
action discretion, realized actions may remain latent, and the designer
controls technological authorization, information use, and continuation
opportunities. This structure identifies when standard dynamic-screening or
moral-hazard results remain sufficient, when reporting and action
implementation must be treated jointly, and when verification, information,
and control instruments create additional economically consequential margins.

\section{Applications}

We conclude with three tractable applications of the design margins developed
in Sections~4--6. The first applies the dynamic virtual-surplus and
implementation analysis of Section~4 to staged control under persistent
private information and makes the branched benchmark structure explicit. The
second interprets the monitoring technology of Section~5 as an AI evaluation
and shows how information use can change hidden behavior being evaluated. The
third applies the coupled-continuation channel of Section~6 to competition for
a future deployment opportunity.

The examples are deliberately stylized. Their purpose is not to introduce
additional general machinery, but to map the preceding analytical objects into
economically interpretable AI-control problems. They also illustrate different
observability structures. Actions may remain latent and affect rewards or
continuation only through observable consequences, may be observed by a
monitoring technology without otherwise becoming public, or may be publicly
observed as a useful benchmark. In each case, current private information or
behavior affects future control, and the value of that continuation regime
feeds back into current incentives.

\subsection{Staged Autonomy under Persistent Private Type}

Consider an AI agent with scalar private type
\(\eta_t\), interpreted as a capability-related private value of delegated
control. The scalar type is distinct from the execution-capability set
\(A_t\) in the general model. At date \(t\), the mechanism chooses a control
level \(q_t\geq0\) and recommends a utilization rate
\(\alpha_t\in[0,1]\). After observing the mechanism response, the agent
chooses its actual utilization
\(a_t\in[0,1]\). Thus \(q_t\) is the screening allocation, while \(a_t\)
is the post-report action.

Let the agent's and designer's gross period payoffs be
\[
u_t(\eta_t,q_t,a_t)
=
\eta_tq_t
-
\frac{\gamma}{2}(a_t-\bar a)^2,
\qquad
v_t(q_t,a_t)
=
\beta q_ta_t-\frac{1}{2}q_t^2,
\]
where
\(\gamma>0\), \(\beta>0\), and
\(\bar a\in(0,1)\). To allow the utilization action to remain hidden, interpret
the payment below as the expected payment induced by the observable
outcome--reward technology conditional on \(a_t\):
\[
p_t
=
p_t^0+\iota_ta_t,
\qquad
|\iota_t|\leq\bar\iota_t,
\]
where \(p_t^0\) is the report-contingent payment intercept and
\(\iota_t\) is the induced marginal action incentive. Thus the mechanism need
not condition directly on a latent \(a_t\); the linear expression is a
reduced-form expected transfer. The observable-action case is nested as the
special case in which this payment can be conditioned directly on utilization.
Period expected payoffs are therefore \(u_t-p_t\) and \(v_t+p_t\).

Suppose
\[
\eta_{t+1}
=
\rho\eta_t+\varepsilon_{t+1},
\qquad
\rho\in(0,1),
\]
where \(\{\varepsilon_t\}_{t\geq2}\) are i.i.d.\ shocks, and let
\(\eta_1\) be uniformly distributed on \([0,1]\). Then
\(I_{1,t}=\rho^{t-1}\) and
\(\Lambda_1(\eta_1)=1-\eta_1\). Hence the dynamic virtual type of
Section~4.3 is
\begin{equation}
\label{eq:sa_virtual_type}
\eta_t^{\mathrm{vir}}
=
\eta_t
-
(1-\eta_1)\rho^{t-1}.
\end{equation}
We restrict attention to parameters for which the relevant solutions satisfy
\(q_t>0\) and \(\alpha_t\in(0,1)\).

Reports satisfy the withholding restriction
\(C_t(\eta_t)=[\underline\eta_t,\eta_t]\).
To isolate the current report--action interaction, suppose a current report
affects the current control, recommendation, and payment but not the
continuation rule once \(\eta_{t+1}\) is realized.

The application-specific \(4B\) problem is
\[
\begin{aligned}
\sup_{\psi}\quad
&
\int
\sum_{t=1}^{T}
\delta^{t-1}
\bigl[
v_t(q_t,a_t)+p_t
\bigr]
\,d\mathbb P^{P,\psi}
\\
\text{s.t.}\quad
&
V_1^A(\eta_1)\geq0
\qquad
\forall\eta_1,
\\
&
V_t^A(\eta^t)
\geq
\widehat V_t^A(\eta^t;r)
\qquad
\forall t,\ \forall r\in C_t(\eta_t),
\\
&
|\iota_t|\leq\bar\iota_t,
\qquad
\psi
\text{ is technologically and action feasible.}
\end{aligned}
\]
This is the application-specific counterpart of the full problem in
Section~4.2. As there, the dynamic virtual-surplus candidate must be checked
separately against the reporting, verification, and hidden-action
implementation conditions.

Since \(u_{t,\eta}=q_t\), the relaxed period virtual surplus is
\begin{equation}
\label{eq:sa_virtual_surplus}
\Phi_t(q_t,\alpha_t)
=
\eta_t^{\mathrm{vir}}q_t
+
\beta q_t\alpha_t
-
\frac{1}{2}q_t^2
-
\frac{\gamma}{2}(\alpha_t-\bar a)^2.
\end{equation}
With the lowest initial type's participation constraint binding, the
verification-free upper bound is therefore obtained by maximizing
\(\Phi_t\) pointwise before checking implementability.

Suppose
\begin{equation}
\label{eq:sa_concavity}
\gamma>\beta^2.
\end{equation}
Then \(\Phi_t\) is strictly concave in \((q_t,\alpha_t)\), and the unique
relaxed candidate is
\begin{equation}
\label{eq:sa_candidate}
q_t^\circ
=
\frac{
\gamma\bigl(\eta_t^{\mathrm{vir}}+\beta\bar a\bigr)
}{
\gamma-\beta^2
},
\qquad
\alpha_t^\circ
=
\bar a+\frac{\beta}{\gamma}q_t^\circ.
\end{equation}

We next check the two additional implementation margins. On the reporting
margin, the agent's post-report action problem is independent of its private
type, and the current report does not affect the continuation rule after
\(\eta_{t+1}\) is realized. Choosing the payment intercept \(p_t^0\) so that
truthful continuation utility satisfies the ordinary envelope formula reduces
the feasible false-report comparisons to
\begin{equation}
\label{eq:sa_reporting_condition}
\int_r^\eta
\bigl[
q_t(x)-q_t(r)
\bigr]
\,dx
\geq0,
\qquad
r\leq\eta.
\end{equation}
Thus a nondecreasing control schedule satisfies the reporting condition. Since
\(\eta_t^{\mathrm{vir}}\) is increasing in the current private type and
\(q_t^\circ\) is increasing in
\(\eta_t^{\mathrm{vir}}\), the relaxed candidate satisfies
\eqref{eq:sa_reporting_condition}. The mechanism can therefore be made
envelope compatible, and the verification wedge is zero in this benchmark.

The remaining margin is obedience. Conditional on a report, the agent solves
\[
\max_{a_t\in[0,1]}
\left\{
-\frac{\gamma}{2}(a_t-\bar a)^2
-\iota_ta_t
\right\}.
\]
Implementing the recommendation \(\alpha_t\) therefore requires
\begin{equation}
\label{eq:sa_obedience}
\iota_t
=
-\gamma(\alpha_t-\bar a).
\end{equation}
The bound on marginal action incentives is equivalent to
\[
\gamma|\alpha_t-\bar a|
\leq
\bar\iota_t.
\]
At the relaxed candidate,
\[
|\iota_t^\circ|
=
\beta q_t^\circ,
\]
so obedience is nonbinding if and only if
\begin{equation}
\label{eq:sa_obedience_threshold}
\beta q_t^\circ
\leq
\bar\iota_t.
\end{equation}

\begin{proposition}[Staged autonomy]
\label{prop:staged_autonomy}
Suppose \(\gamma>\beta^2\) and the maintained interior conditions hold. Then:

\begin{enumerate}
\item[(i)]
The solution to the full \(4B\) problem is
\begin{equation}
\label{eq:sa_4B_solution}
q_t^{4B}
=
q_t^\circ
-
\frac{\beta}{\gamma}
\left[
\beta q_t^\circ-\bar\iota_t
\right]_+,
\qquad
\alpha_t^{4B}
=
\bar a
+
\frac{1}{\gamma}
\min\{\beta q_t^\circ,\bar\iota_t\}.
\end{equation}

\item[(ii)]
If
\(\beta q_t^\circ\leq\bar\iota_t\),
the relaxed candidate is fully implementable:
\[
q_t^{4B}=q_t^\circ,
\qquad
\alpha_t^{4B}=\alpha_t^\circ.
\]
The action-implementation and verification restrictions are then both
nonbinding, so the \(4B\) allocation coincides with the
\(2B^{\mathrm{HI}}\), \(3B^{\mathrm{HI+V}}\), and
\(3B^{\mathrm{HI+HA}}\) allocations.

\item[(iii)]
If
\(\beta q_t^\circ>\bar\iota_t\),
obedience binds and
\begin{equation}
\label{eq:sa_binding_obedience}
\alpha_t^{4B}
=
\bar a+\frac{\bar\iota_t}{\gamma},
\qquad
q_t^{4B}
=
\eta_t^{\mathrm{vir}}
+
\beta\bar a
+
\frac{\beta\bar\iota_t}{\gamma}.
\end{equation}
The reporting and verification margins remain nonbinding, so
\(q_t^{4B}=q_t^{3B^{\mathrm{HI+HA}}}\), while
\(q_t^{3B^{\mathrm{HI+V}}}=q_t^{2B^{\mathrm{HI}}}=q_t^\circ\).
\end{enumerate}
\end{proposition}

\begin{proof}
Strict concavity under \eqref{eq:sa_concavity} gives
\eqref{eq:sa_candidate}. The candidate \(q_t^\circ\) is increasing in the
current private type, so
\eqref{eq:sa_reporting_condition} holds and an envelope-compatible payment
schedule implements truthful reporting.

Equation~\eqref{eq:sa_obedience} characterizes sequential obedience. At the
relaxed candidate it requires
\(|\iota_t^\circ|=\beta q_t^\circ\), yielding
\eqref{eq:sa_obedience_threshold}. If this condition holds, the relaxed
candidate satisfies the full nested incentive condition and attains the
virtual-surplus upper bound.

If
\(\beta q_t^\circ>\bar\iota_t\),
the action-incentive bound binds. Since
\(q_t^\circ>0\) implies
\(\alpha_t^\circ>\bar a\), the largest implementable recommendation is
\[
\alpha_t
=
\bar a+\frac{\bar\iota_t}{\gamma}.
\]
Conditional on this recommendation, maximization of
\eqref{eq:sa_virtual_surplus} with respect to \(q_t\) gives
\[
\eta_t^{\mathrm{vir}}
+
\beta\alpha_t-q_t
=
0,
\]
which yields \eqref{eq:sa_binding_obedience}. The resulting control remains
increasing in the current private type, so the reporting and verification
conditions remain satisfied. Combining the binding and nonbinding cases gives
\eqref{eq:sa_4B_solution}.
\end{proof}

For comparison, removing private information and the additional verification
and action-implementation restrictions gives
\begin{equation}
\label{eq:sa_1B}
q_t^{1B}
=
\frac{
\gamma(\eta_t+\beta\bar a)
}{
\gamma-\beta^2
}.
\end{equation}
If private information is removed while the action-implementation constraint
is retained, the hidden-action second-best control is
\(q_t^{2B^{\mathrm{HA}}}
=
q_t^{1B}
-
(\beta/\gamma)
[\beta q_t^{1B}-\bar\iota_t]_+\).
Together with Proposition~\ref{prop:staged_autonomy}, the complete allocation
map in this application is therefore
\begin{equation}
\label{eq:sa_branch_map}
\begin{aligned}
q_t^{2B^{\mathrm{HI}}}
&=
q_t^{3B^{\mathrm{HI+V}}}
=
q_t^\circ,
&
q_t^{2B^{\mathrm{HA}}}
&=
q_t^{1B}
-
\frac{\beta}{\gamma}
[\beta q_t^{1B}-\bar\iota_t]_+,
\\
q_t^{3B^{\mathrm{HI+HA}}}
&=
q_t^{4B}
=
q_t^\circ
-
\frac{\beta}{\gamma}
[\beta q_t^\circ-\bar\iota_t]_+.
\end{aligned}
\end{equation}
Thus the upper verification branch collapses because the verification
constraint is nonbinding, while the lower hidden-action branch is generally
distinct whenever the marginal action-incentive bound binds. In particular,
the hidden-action-only distortion is
\(q_t^{1B}-q_t^{2B^{\mathrm{HA}}}
=
(\beta/\gamma)
[\beta q_t^{1B}-\bar\iota_t]_+\).

Combining \eqref{eq:sa_virtual_type} and
\eqref{eq:sa_4B_solution},
\begin{equation}
\label{eq:sa_distortion}
\boxed{
q_t^{1B}-q_t^{4B}
=
\underbrace{
\frac{\gamma}{\gamma-\beta^2}
(1-\eta_1)\rho^{t-1}
}_{\text{dynamic information-rent distortion}}
+
\underbrace{0}_{\text{verification distortion}}
+
\underbrace{
\frac{\beta}{\gamma}
\left[
\beta q_t^\circ-\bar\iota_t
\right]_+
}_{\text{obedience distortion}}.
}
\end{equation}

Equations~\eqref{eq:sa_branch_map}--\eqref{eq:sa_distortion} make the complete
Section~4 benchmark structure concrete. Persistent private information lowers
delegated control through the propagated information-rent term. In this
benchmark, monotone control and flexible payment intercepts make one-sided
verification nonbinding, while bounded action incentives generate an
additional implementation distortion whenever the relaxed utilization target
cannot be sustained. The \(2B^{\mathrm{HI}}\) and
\(2B^{\mathrm{HA}}\) branches therefore isolate the informational and
hidden-action frictions separately, while \(3B^{\mathrm{HI+HA}}\) and \(4B\)
combine them. Because \(\rho\in(0,1)\), the component attributable to the
initial information rent decays geometrically with the horizon. This is the
staged-autonomy effect. It does not imply that the realized path of
\(q_t^{4B}\) is monotone over calendar time, since the current private type and
the obedience margin may also change.

\subsection{Strategic Evaluation and Future Deployment}

Section~5 distinguishes the primitive monitoring technology from the
mechanism's information-use policy. In this application, the monitoring
technology is an AI evaluation. The evaluation itself is fixed, while the
mechanism varies how much of its output enters the public continuation state.
The application therefore isolates the information--incentive interaction
rather than a change in monitoring capacity.

Consider two dates and suppress reporting to isolate the hidden-action
channel. The agent has privately known capability
\(\eta\in[\underline\eta,\overline\eta]\)
and chooses an evaluation-performance level
\(a\in[0,\eta]\). Performance is observed by the primitive monitoring
technology but is not otherwise publicly observed. The technological bound
\(a\leq\eta\) captures the fact that performance can be suppressed but cannot
exceed capability. The agent's current payoff is
\[
-\frac{1}{2}(\eta-a)^2,
\]
so absent continuation consequences it chooses \(a=\eta\).

The primitive monitoring technology perfectly observes performance. Thus the
raw monitoring signal satisfies
\[
x=a.
\]
The action itself need not enter the public history. Instead, the
information-use policy observes the raw signal \(x\) and determines how much
of that evidence becomes public, exactly as in Section~5. Let
\(\omega\in[0,1]\) index the information-use policy
\(\zeta^\omega\): with probability \(\omega\), the public monitoring message
reveals the raw evaluation score,
\(e=x\), while with probability \(1-\omega\) it sends the uninformative
message \(e=\varnothing\). Hence
\[
Q^\omega
=
Q^{\mathrm{mon}}\zeta^\omega.
\]
The primitive monitoring technology is fixed across \(\omega\); only its
public use changes. For any fixed performance \(a\),
\(\omega'>\omega\) implies
\(Q^{\omega'}\succeq_B Q^\omega\).

At date 2, the public monitoring message determines continuation deployment.
If the score is revealed, suppose the agent is deployed with probability
\(d_0-\nu x\); following the uninformative message, deployment occurs with
probability \(d_0\), where \(\nu\geq0\) and parameters ensure that these
probabilities lie in \([0,1]\). Thus a higher revealed performance score
induces a more restrictive continuation regime. The agent values deployment
at \(B>0\). Its date-1 objective is
\[
U(a;\eta,\omega)
=
-\frac{1}{2}(\eta-a)^2
+
\delta B
\bigl(
d_0-\omega\nu a
\bigr).
\]

\begin{proposition}[Strategic evaluation and sandbagging]
\label{prop:dynamic_sandbagging}
For every \(\omega\in[0,1]\) and capability \(\eta\), the agent's optimal
evaluation performance is
\[
a^*(\eta;\omega)
=
\max\left\{
\eta-\delta B\omega\nu,
0
\right\}.
\]
Hence:

\begin{enumerate}
\item[(i)]
If \(\omega=0\) or \(\nu=0\), monitoring information does not affect
continuation deployment and
\(a^*(\eta;\omega)=\eta\).

\item[(ii)]
Whenever the solution is interior,
\[
\eta-a^*(\eta;\omega)
=
\delta B\omega\nu.
\]

\item[(iii)]
Sandbagging increases with the information-use intensity \(\omega\), the
value \(B\) of future deployment, the discount factor \(\delta\), and the
sensitivity \(\nu\) of continuation deployment to revealed performance.

\item[(iv)]
For fixed underlying performance,
\(\omega'>\omega\) implies
\(Q^{\omega'}\succeq_B Q^\omega\).
At the induced behavioral optimum,
\[
a^*(\eta;\omega')
\leq
a^*(\eta;\omega),
\]
with strict inequality whenever both solutions are interior.
\end{enumerate}
\end{proposition}

\begin{proof}
The objective is strictly concave in \(a\). An interior optimum satisfies
\[
\eta-a-\delta B\omega\nu=0,
\]
so
\(a=\eta-\delta B\omega\nu\).
Imposing \(a\geq0\) gives the stated solution, while
\(a\leq\eta\) is automatically satisfied. Parts~(i)--(iii) follow directly.

For part~(iv), fix \(a\). If \(\omega'>\omega\), the public experiment under
\(\omega\) can be obtained from that under \(\omega'\) by independently
replacing each revealed score with
\(\varnothing\) with probability \(1-\omega/\omega'\). Hence
\(Q^{\omega'}\succeq_B Q^\omega\). The expression for
\(a^*(\eta;\omega)\) is weakly decreasing in \(\omega\), and strictly
decreasing whenever the solution remains interior.
\end{proof}

The proposition separates monitoring capacity from information use. The
primitive monitoring technology always observes the same hidden performance
perfectly; increasing \(\omega\) merely makes the induced public experiment
more informative. Yet equilibrium performance falls with \(\omega\). The
value of a more informative public experiment therefore cannot generally be
evaluated while holding fixed the hidden behavior that generates the
monitoring evidence.

The application also makes the nonseparability condition of
Corollary~\ref{cor:information_separation} explicit. Relative to full
performance \(a=\eta\), a downward deviation to \(a<\eta\) increases expected
continuation deployment value by
\[
\delta B\omega\nu(\eta-a).
\]
The continuation incentive difference therefore varies with the
information-use policy. Information use is not incentive separable: changing
\(\omega\) changes the current action incentive itself. When
\(\omega=0\), the continuation channel disappears and the agent chooses
\(a=\eta\); as \(\omega\) rises, the same monitoring technology creates a
stronger incentive to suppress evaluation performance.

The mechanism is a feedback from hidden current behavior to monitoring and
back to behavior through continuation control. Performance generates raw
monitoring evidence, information use determines how that evidence enters the
public continuation state, and the continuation state determines future
deployment. The information-use policy never needs to condition directly on
the latent action: it conditions on the monitoring evidence generated by that
action. Anticipating this dependence, the agent changes the behavior being
evaluated. Thus an AI evaluation is not merely a measurement device when its
output governs future control; its use becomes part of the incentive
mechanism.

\subsection{Competition through Continuation Opportunities}

Section~6.4 shows that current action incentives can be coupled through the
continuation mechanism even when primitive current payoffs and current rewards
are separable. We isolate that channel by suppressing reporting and peer
information and considering two agents whose only interaction comes from the
allocation of a future deployment opportunity. This application deliberately
uses the observable-action benchmark of Section~6.4 so that continuation
authorization can condition directly on current performance.

There are two agents \(i\in\{1,2\}\) and two dates. At date 1, each chooses
a publicly observed performance action
\[
a_i\in[0,\bar a]
\]
at cost \(a_i^2/2\). There is no current performance reward, and agent \(i\)'s
primitive current payoff is independent of \(a_j\).

At date 2, exactly one agent receives a deployment opportunity, which each
agent values at \(B>0\). The designer commits to
\[
\Pr(i\text{ is deployed}\mid a_i,a_j)
=
\frac{1}{2}
+
\mu(a_i-a_j),
\qquad
i\neq j,
\]
where
\(\mu>0\) and
\(\mu\bar a\leq1/2\).
The deployment probabilities therefore sum to one and remain in
\([0,1]\). Agent \(i\)'s date-1 payoff is
\[
U_i(a_i,a_j)
=
-\frac{1}{2}a_i^2
+
\delta B
\left[
\frac{1}{2}
+
\mu(a_i-a_j)
\right].
\]
There is no coupled current reward; the entire strategic linkage operates
through continuation authorization.

\begin{proposition}[Competition through continuation authorization]
\label{prop:continuation_competition}
Suppose
\[
0<\delta B\mu<\bar a.
\]
The continuation-authorization mechanism induces the unique Nash equilibrium
\[
a_1^*=a_2^*=\delta B\mu.
\]
Hence strictly positive current performance can be induced even though
primitive current payoffs are separable and there is no current performance
reward. Moreover, any symmetric target
\(\widehat a\in(0,\bar a)\)
can be implemented by setting
\[
\mu
=
\frac{\widehat a}{\delta B},
\]
provided the resulting authorization rule satisfies the probability-feasibility
condition.
\end{proposition}

\begin{proof}
Holding \(a_j\) fixed, agent \(i\)'s objective is strictly concave in
\(a_i\). The interior first-order condition is
\[
-a_i+\delta B\mu=0.
\]
Hence the unique best response is
\(a_i=\delta B\mu\).
Under the maintained interiority condition, both agents therefore choose
\[
a_1^*=a_2^*=\delta B\mu.
\]
Solving the same condition for \(\mu\) yields the implementation formula for
any feasible interior symmetric target.
\end{proof}

This is a direct observable-action counterpart of the
continuation-incentive channel in
Proposition~\ref{prop:coupled_continuation_obedience}. Here
\[
\gamma_{i1}^A=1,
\qquad
a_{i1}^0=0,
\qquad
\overline R_{i1}=0,
\]
while the continuation value is
\[
\overline V_{i,2}^A(a_i,a_j)
=
B
\left[
\frac{1}{2}
+
\mu(a_i-a_j)
\right].
\]
Hence
\[
\frac{\partial
\overline V_{i,2}^A(a_i,a_j)}
{\partial a_i}
=
B\mu,
\]
and the truthful-report action condition of Section~6.4 reduces to
\[
a_i
=
\delta B\mu.
\]
Future authorization therefore supplies the marginal incentive that would
otherwise have to come from a current reward.

The application also makes the cross-agent continuation externality explicit.
Increasing \(a_i\) raises agent \(i\)'s probability of future deployment and,
because the opportunity is scarce, lowers agent \(j\)'s probability.
Agents whose primitive current payoffs are separable therefore become
strategically linked through the continuation mechanism. The linear
authorization rule makes the marginal continuation return constant, so each
best response is independent of the rival's action. With a nonlinear
allocation rule, the marginal continuation value would generally depend on the
full action profile, generating the richer strategic interaction allowed by
Section~6.

Public observability is useful here only to isolate the continuation channel.
Under hidden actions, the same marginal-incentive logic applies to the
expected continuation value induced by the action profile through public
outcomes or monitoring evidence, as in the general formulation of
Section~6.4. Direct conditioning on \((a_i,a_j)\) is therefore a convenient
benchmark rather than a requirement of the coupled-incentive result.

The same logic extends beyond winner-take-all deployment. Scarce future tool
access, task assignments, compute budgets, transaction limits, communication
rights, or decision authority can all create continuation incentives. Future
control can therefore both motivate current behavior and generate strategic
interaction among agents whose current primitive preferences and rewards are
otherwise separable.

\subsection{Discussion}

The three applications isolate the principal margins developed in
Sections~4--6. In the staged-autonomy application, the complete benchmark map
is visible in closed form: persistent private information generates the
dynamic virtual-type distortion, bounded action incentives generate the
hidden-action implementation distortion, and one-sided verification is
nonbinding. The \(2B^{\mathrm{HI}}\) and \(2B^{\mathrm{HA}}\) branches isolate
the two primitive frictions, while \(3B^{\mathrm{HI+HA}}\) and \(4B\) combine
them. In the strategic-evaluation application, performance is hidden from the
public mechanism but affects the monitoring signal; information use changes
future deployment and thereby feeds back into the behavior that generates the
evidence. In the multi-agent application, observable current actions affect a
scarce future authorization opportunity, creating cross-agent incentives even
though current primitive payoffs and rewards are separable.

The applications therefore deliberately span different action-observability
regimes. Staged autonomy uses the reduced-form expected reward induced by a
post-report action and hence does not require the action itself to be
contractible. Strategic evaluation makes the distinction sharper: performance
is observed by the primitive monitor but is not otherwise public, so only the
monitoring evidence can enter the continuation mechanism. Competition for
future deployment uses publicly observed actions to isolate the continuation
externality. These cases are all nested by the general framework rather than
being alternative assumptions imposed throughout the paper.

The common object is the continuation mechanism. Continuation values transmit
the effect of private information across time, make monitoring information an
incentive instrument, and couple agents through future access and
opportunities. Hidden action adds a further channel because privately chosen
behavior may change observable outcomes, monitoring evidence, physical
transitions, and therefore the continuation state without itself becoming a
public contracting variable. These channels are analytically distinct but
need not operate separately in richer environments: monitoring may determine
future authorization, authorization may be allocated jointly across agents,
and private types may affect both current behavior and continuation control.

The applications therefore provide a decomposition of the main forces in the
general framework rather than three unrelated examples. Their common message
is that a dynamic AI mechanism cannot generally be reduced to a sequence of
independently optimized static mechanisms. Current private information and
privately chosen behavior shape the future control regime, while the value of
that regime feeds back into reporting and action incentives today.

\section{Conclusion}

This paper develops a general framework for dynamic mechanism and information
design in AI-agent principal--agent environments, encompassing optimal contracting and integrating
evolving hidden information and hidden action with partial verification,
technological authorization, and multi-agent interaction. The common object
linking these dimensions is the continuation mechanism: current reports,
privately chosen actions, and information use shape public beliefs, future
states, and future control, while the value of those continuation consequences
feeds back into current reporting and action incentives. Conventional dynamic
screening, hidden-action contracting, and conditional dynamic information
design therefore arise as special cases of the framework when the other
margins are held fixed.

The first set of results concerns implementability. A dynamic revelation
principle continues to reduce the analysis to direct mechanisms under partial
verification, but incentive compatibility must jointly support truthful
reporting and post-report obedience. Because the agent can receive additional
private information and reoptimize its action after any feasible report, the
relevant deviation is a report--action deviation rather than a report evaluated
at a mechanically assigned action. The recursive characterization therefore
links honesty and obedience through continuation values while preserving the
standard dynamic-programming structure.

Optimal design contains three distinct margins. Dynamic hidden information
generates the familiar information rent; one-sided verification can create an
additional continuation-rent wedge; and hidden action restricts which desired
allocations can actually be implemented. The first two enter the rent
representation, whereas the third generally operates through the implementable
set rather than as another additive rent. This yields a
verification-adjusted dynamic virtual-surplus representation and a branched
benchmark structure. Conventional dynamic screening is the
\(2B^{\mathrm{HI}}\) branch, complete-information hidden-action contracting is
the \(2B^{\mathrm{HA}}\) branch, and the full \(4B\) problem combines hidden
information, one-sided verification, and hidden-action implementation. First
best requires both the private-information and action-implementation frictions
to disappear.

The same decomposition clarifies dynamic distortions. Under type-linear
payoffs, a dynamic virtual type summarizes the propagated hidden-information
distortion, while hidden action enters through the obedience and
implementability conditions. In the linear-quadratic benchmark, the total
distortion separates into a dynamic information-rent component and an
additional hidden-action component, with verification nonbinding in that
benchmark. Persistence therefore provides a mechanism-design rationale for
staged control, but the total distortion need not decline monotonically over
time when action incentives or other implementation constraints also bind.

A separate contribution is the dynamic information-design problem. Conditional
on a fixed primitive monitoring technology, the mechanism chooses how the
resulting evidence is transformed into public information and incorporated
into the continuation state. This provides a conditional dynamic
information-design benchmark inside the broader mechanism-design environment.
Holding reporting and action incentives fixed recovers the familiar logic of
dynamic information design; allowing them to respond endogenously makes
information use part of the incentive problem itself.

The distinction between monitoring capacity and information use is central.
More informative monitoring capacity weakly expands the designer's opportunity
set when information can be freely discarded, because less informative
experiments can be reproduced by garbling. Greater disclosure, however, need
not be optimal. Information released today changes future authorization,
rewards, deployment, and other continuation decisions and may therefore alter
current reporting or hidden-action incentives. Information and incentive design
are separable precisely when the relevant continuation-value differences
between truthful-and-obedient behavior and feasible deviations are invariant
to the information-use policy. Otherwise, information use is itself a dynamic
incentive instrument. When monitoring depends on latent behavior, the policy
can influence the incentive to generate the evidence being disclosed without
ever conditioning directly on the hidden action.

The multi-agent extension preserves this structure while adding new channels
of interaction. Under independence, the rent and virtual-surplus analysis
extends agent by agent, hidden-action implementation remains embedded in each
agent's obedience condition, and joint feasibility links allocations across
agents. Correlated private information can make peer reports informative about
an agent's type and thereby discipline false reporting, although smooth peer
scoring cannot generally eliminate first-order local deviations in continuous
environments. Current actions can also affect expected rewards and future
authorization through public outcomes, monitoring evidence, and continuation
states, creating strategically coupled incentives even when primitive current
payoffs are separable.

The applications make these mechanisms concrete. Staged autonomy separates the
dynamic hidden-information and hidden-action distortions within the complete
benchmark map. Strategic evaluation shows how hidden performance can affect
monitoring evidence and how the use of that evidence can feed back into the
behavior being evaluated. Competition for future deployment illustrates how
continuation authorization can generate present incentives and cross-agent
interaction even without coupled current rewards. The applications differ in
action observability, but all operate through the same feedback from current
information and behavior to future control.

The framework deliberately stops short of making information production itself
endogenous. Primitive private-signal and monitoring technologies are taken as
given, while their strategic use and the resulting continuation consequences
are endogenous. Endogenizing information acquisition, experimentation,
learning, posterior evolution, capability development, and dynamic feasibility
would move from conditional dynamic information design to the broader problem
of endogenous informational evolution studied in \citet{Tian2026}. Other
natural extensions include limited commitment, richer participation and
capability dynamics, robustness to misspecified beliefs and transition laws,
and richer forms of multi-agent competition, coordination, and information
sharing.

The framework builds on the static AI-agent honesty--obedience and
one-sided-imitation analysis of \citet{BergemannKohMorris2026}, the dynamic
screening methods of \citet{PavanSegalToikka2014}, and the integrated
dynamic implementation and contracting foundations of
\citet{TianGeneralDynamic2026}. Its distinct contribution is to combine
and develop these elements in an AI-agent principal--agent environment with
a different information and decision timing, partial verification,
authorization, monitoring, and continuation control.

The broader conclusion is therefore not that AI agents require a mechanism
design theory separate from the existing economic tradition. Dynamic
screening, hidden-action contracting, and dynamic information design remain
the principal benchmarks. What changes is that AI-agent environments naturally
bring these margins together with technologically enforceable authorization
and history-dependent continuation control. The framework provides a common
way to identify when the familiar benchmark results remain valid, when the
different margins must be analyzed jointly, and when information and future
control themselves become central instruments of dynamic mechanism design.



\begin{thebibliography}{}

\bibitem[Athey and Segal(2013)]{AtheySegal2013}
Athey, Susan, and Ilya Segal.
``An Efficient Dynamic Mechanism.''
\emph{Econometrica} 81 (2013): 2463--2485.

\bibitem[Baron and Besanko(1984)]{BaronBesanko1984}
Baron, David P., and David Besanko.
``Regulation and Information in a Continuing Relationship.''
\emph{Information Economics and Policy} 1 (1984): 267--302.

\bibitem[Battaglini(2005)]{Battaglini2005}
Battaglini, Marco.
``Long-Term Contracting with Markovian Consumers.''
\emph{American Economic Review} 95 (2005): 637--658.

\bibitem[Ben-Porath, Dekel, and Lipman(2019)]{BenPorathDekelLipman2019}
Ben-Porath, Elchanan, Eddie Dekel, and Barton L. Lipman.
``Mechanisms with Evidence: Commitment and Robustness.''
\emph{Econometrica} 87 (2019): 529--566.

\bibitem[Bergemann, Heumann, and Morris(2026a)]{BergemannHeumannMorris2026}
Bergemann, Dirk, Tibor Heumann, and Stephen Morris.
``Information Design and Mechanism Design: An Integrated Framework.''
Cowles Foundation Discussion Paper 2494 / forthcoming in
\emph{Advances in Economics and Econometrics}, 2026.

\bibitem[Bergemann, Heumann, and Morris(2026b)]{BergemannHeumannMorrisJPE2026}
Bergemann, Dirk, Tibor Heumann, and Stephen Morris.
``Screening with Persuasion.''
\emph{Journal of Political Economy} 134 (2026): 570--625.

\bibitem[Bergemann, Koh, and Morris(2026)]{BergemannKohMorris2026}
Bergemann, Dirk, Andrew Koh, and Stephen Morris.
``Mechanism Design for Alignment and Control.''
Cowles Foundation Discussion Paper 2557 / arXiv:2609.01595, 2026.

\bibitem[Bergemann and Morris(2016)]{BergemannMorris2016}
Bergemann, Dirk, and Stephen Morris.
``Bayes Correlated Equilibrium and the Comparison of Information Structures in Games.''
\emph{Theoretical Economics} 11 (2016): 487--522.

\bibitem[Bergemann and Morris(2019)]{BergemannMorris2019}
Bergemann, Dirk, and Stephen Morris.
``Information Design: A Unified Perspective.''
\emph{Journal of Economic Literature} 57 (2019): 44--95.

\bibitem[Bergemann and V\"alim\"aki(2010)]{BergemannValimaki2010}
Bergemann, Dirk, and Juuso V\"alim\"aki.
``The Dynamic Pivot Mechanism.''
\emph{Econometrica} 78 (2010): 771--789.

\bibitem[Bergemann and V\"alim\"aki(2019)]{BergemannValimaki2019}
Bergemann, Dirk, and Juuso V\"alim\"aki.
``Dynamic Mechanism Design: An Introduction.''
\emph{Journal of Economic Literature} 57 (2019): 235--274.

\bibitem[Caillaud, Guesnerie, and Rey(1992)]{CaillaudGuesnerieRey1992}
Caillaud, Bernard, Roger Guesnerie, and Patrick Rey.
``Noisy Observation in Adverse Selection Models.''
\emph{Review of Economic Studies} 59 (1992): 595--615.


\bibitem[Carroll(2015)]{Carroll2015}
Carroll, Gabriel.
``Robustness and Linear Contracts.''
\emph{American Economic Review} 105 (2015): 536--563.

\bibitem[Casper et al.(2023)]{CasperEtAl2023}
Casper, Stephen, Xander Davies, Claudia Shi, et al.
``Open Problems and Fundamental Limitations of Reinforcement Learning from Human Feedback.''
\emph{Transactions on Machine Learning Research}, 2023.

\bibitem[Chen et al.(2023)]{ChenEtAl2023}
Chen, Yiting, Tracy Xiao Liu, You Shan, and Songfa Zhong.
``The Emergence of Economic Rationality of GPT.''
\emph{Proceedings of the National Academy of Sciences} 120 (2023): e2316205120.

\bibitem[Caillaud, Guesnerie, and Rey(1992)]{CaillaudGuesnerieRey1992}
Caillaud, Bernard, Roger Guesnerie, and Patrick Rey.


\bibitem[Collina et al.(2025)]{CollinaEtAl2025}
Collina, Natalie, Surbhi Goel, Aaron Roth, Emily Ryu, and Mirah Shi.
``Emergent Alignment via Competition.''
Working paper, 2025, arXiv:2509.15090.

\bibitem[Courty and Li(2000)]{CourtyLi2000}
Courty, Pascal, and Hao Li.
``Sequential Screening.''
\emph{Review of Economic Studies} 67 (2000): 697--717.

\bibitem[Debreu(1952)]{Debreu1952}
Debreu, G. (1952).
A social equilibrium existence theorem.
\textit{Proceedings of the National Academy of Sciences}, 38(10), 886--893.

\bibitem[Doval and Ely(2020)]{DovalEly2020}
Doval, Laura, and Jeffrey C. Ely.
``Sequential Information Design.''
\emph{Econometrica} 88 (2020): 2575--2608.

\bibitem[Dworczak and Smolin(2026)]{DworczakSmolin2026}
Dworczak, Piotr, and Alex Smolin.
``Robust Trust.''
Working paper, 2026, arXiv:2602.09490.

\bibitem[Ely, Georgiadis, and Rayo(2025)]{ElyGeorgiadisRayo2025}
Ely, Jeffrey C., George Georgiadis, and Luis Rayo.
``Feedback Design in Dynamic Moral Hazard.''
\emph{Econometrica} 93 (2025): 597--621.

\bibitem[Es\H{o} and Szentes(2017)]{EsoSzentes2017}
Es\H{o}, P\'eter, and Bal\'azs Szentes.
``Dynamic Contracting: An Irrelevance Theorem.''
\emph{Theoretical Economics} 12 (2017): 109--139.

\bibitem[Forges(1986)]{Forges1986}
Forges, Fran\c{c}oise.
``An Approach to Communication Equilibria.''
\emph{Econometrica} 54 (1986): 1375--1385.

\bibitem[Fudenberg and Liang(2026)]{FudenbergLiang2026}
Fudenberg, Drew, and Annie Liang.
``Friend or Foe: Delegating to an AI Whose Alignment Is Unknown.''
Working paper, 2026; accepted for presentation at ACM EC 2026.

\bibitem[Gan and Li(2026)]{GanLi2026}
Gan, Tan, and Yingkai Li.
``Screening for Choice Sets.''
Working paper, 2026, arXiv:2601.15580.

\bibitem[Garrett and Pavan(2015)]{GarrettPavan2015}
Garrett, Daniel F., and Alessandro Pavan.
``Dynamic Managerial Compensation: A Variational Approach.''
\emph{Journal of Economic Theory} 159 (2015): 775--818.

\bibitem[Green and Laffont(1986)]{GreenLaffont1986}
Green, Jerry R., and Jean-Jacques Laffont.
``Partially Verifiable Information and Mechanism Design.''
\emph{Review of Economic Studies} 53 (1986): 447--456.

\bibitem[Greenblatt et al.(2024)]{GreenblattEtAl2024}
Greenblatt, Ryan, Carson Denison, Benjamin Wright, et al.
``Alignment Faking in Large Language Models.''
arXiv:2412.14093, 2024.

\bibitem[Guan et al.(2025)]{GuanEtAl2025}
Guan, Melody Y., Miles Wang, Micah Carroll, et al.
``Monitoring Monitorability.''
arXiv:2512.18311, 2025.


\bibitem[Guesnerie, Picard, and Rey(1989)]{GuesneriePicardRey1989}
Guesnerie, Roger, Pierre Picard, and Patrick Rey.
``Adverse Selection and Moral Hazard with Risk Neutral Agents.''
\emph{European Economic Review} 33 (1989): 807--823.

\bibitem[Hadfield-Menell and Hadfield(2019)]{HadfieldMenellHadfield2019}
Hadfield-Menell, Dylan, and Gillian K. Hadfield.
``Incomplete Contracting and AI Alignment.''
In \emph{Proceedings of the AAAI/ACM Conference on AI, Ethics, and Society}, 2019.

\bibitem[Hammond et al.(2025)]{HammondEtAl2025}
Hammond, Lewis, Alan Chan, Jesse Clifton, et al.
``Multi-Agent Risks from Advanced AI.''
arXiv:2502.14143, 2025.

\bibitem[Herrmann and Levinstein(2026)]{HerrmannLevinstein2026}
Herrmann, Daniel A., and Benjamin A. Levinstein.
``Radical AI Interpretability.''
arXiv:2606.26523, 2026.

\bibitem[Hubinger et al.(2024)]{HubingerEtAl2024}
Hubinger, Evan, Carson Denison, Jesse Mu, et al.
``Sleeper Agents: Training Deceptive LLMs that Persist Through Safety Training.''
arXiv:2401.05566, 2024.

\bibitem[Hurwicz and Uzawa(1971)]{HurwiczUzawa1971}
Hurwicz, Leonid, and Hirofumi Uzawa.
``On the Integrability of Demand Functions.''
In \emph{Preferences, Utility, and Demand}, edited by John S. Chipman,
Leonid Hurwicz, Marcel K. Richter, and Hugo F. Sonnenschein, 114--148.
New York: Harcourt Brace Jovanovich, 1971.

\bibitem[Hurwicz and Shapiro(1978)]{HurwiczShapiro1978}
Hurwicz, Leonid, and Leonard Shapiro.
``Incentive Structures Maximizing Residual Gain under Incomplete Information.''
\emph{Bell Journal of Economics} 9 (1978): 180--191.

\bibitem[Kamenica and Gentzkow(2011)]{KamenicaGentzkow2011}
Kamenica, Emir, and Matthew Gentzkow.
``Bayesian Persuasion.''
\emph{American Economic Review} 101 (2011): 2590--2615.

\bibitem[Korbak et al.(2025)]{KorbakEtAl2025}
Korbak, Tomek, Mikita Balesni, Elizabeth Barnes, et al.
``Chain of Thought Monitorability: A New and Fragile Opportunity for AI Safety.''
arXiv:2507.11473, 2025.

\bibitem[Kwa et al.(2025)]{KwaEtAl2025}
Kwa, Thomas, Ben West, Joel Becker, et al.
``Measuring AI Ability to Complete Long Tasks.''
arXiv:2503.14499, 2025.

\bibitem[Laffont and Tirole(1986)]{LaffontTirole1986}
Laffont, Jean-Jacques, and Jean Tirole.
``Using Cost Observation to Regulate Firms.''
\emph{Journal of Political Economy} 94 (1986): 614--641.

\bibitem[Lipsey and Lancaster(1956)]{LipseyLancaster1956}
Lipsey, R. G., and Kelvin Lancaster.
``The General Theory of Second Best.''
\emph{Review of Economic Studies} 24 (1956): 11--32.

\bibitem[Mazeika et al.(2025)]{MazeikaEtAl2025}
Mazeika, Mantas, Xuwang Yin, Renan Tamirisa, et al.
``Utility Engineering: Analyzing and Controlling Emergent Value Systems in AIs.''
arXiv:2502.08640, 2025.

\bibitem[Melumad and Reichelstein(1989)]{MelumadReichelstein1989}
Melumad, Nahum D., and Stefan Reichelstein.
``Value of Communication in Agencies.''
\emph{Journal of Economic Theory} 47 (1989): 334--368.


\bibitem[Meng and Tian(2013)]{MengTian2013}
Meng, Dawen, and Guoqiang Tian.
``Multi-task Incentive Contract and Performance Measurement with Multidimensional Types.''
\emph{Games and Economic Behavior} 77 (2013): 377--404.


\bibitem[Meng and Tian(2021)]{MengTian2021}
Meng, D. and Tian, G. (2021).
The competitive and welfare effects of long-term contracts with network
externalities and bounded rationality.
\textit{Economic Theory}, 72, 337--375.

\bibitem[Meng et al.(2022)Meng, Sun, and Tian]{MengSunTian2022}
Meng, D., Sun, L., and Tian, G. (2022).
Dynamic mechanism design on social networks.
\textit{Games and Economic Behavior}, 131, 84--120.

\bibitem[Mertens and Zamir(1985)]{MertensZamir1985}
Mertens, Jean-Fran\c{c}ois, and Shmuel Zamir.
``Formulation of Bayesian Analysis for Games with Incomplete Information.''
\emph{International Journal of Game Theory} 14 (1985): 1--29.

\bibitem[Microsoft Security(2026)]{Microsoft2026}
Microsoft Security.
``Least Privilege for AI Agents: Identity, Access, and Tool Binding.''
July 2026.

\bibitem[Myerson(1982)]{Myerson1982}
Myerson, Roger B.
``Optimal Coordination Mechanisms in Generalized Principal--Agent Problems.''
\emph{Journal of Mathematical Economics} 10 (1982): 67--81.

\bibitem[Needham et al.(2025)]{NeedhamEtAl2025}
Needham, Joe, Giles Edkins, Govind Pimpale, Henning Bartsch, and Marius Hobbhahn.
``Large Language Models Often Know When They Are Being Evaluated.''
arXiv:2505.23836, 2025.

\bibitem[Pavan(2026)]{Pavan2026}
Pavan, Alessandro.
``Dynamic Contracting.''
\emph{Annual Review of Economics} 18 (2026): 1--25.

\bibitem[Pavan, Segal, and Toikka(2014)]{PavanSegalToikka2014}
Pavan, Alessandro, Ilya Segal, and Juuso Toikka.
``Dynamic Mechanism Design: A Myersonian Approach.''
\emph{Econometrica} 82 (2014): 601--653.

\bibitem[Picard(1987)]{Picard1987}
Picard, Pierre.
``On the Design of Incentive Schemes under Moral Hazard and Adverse Selection.''
\emph{Journal of Public Economics} 33 (1987): 305--331.


\bibitem[Qiu, Carroll, and Allen(2026)]{QiuCarrollAllen2026}
Qiu, Tianyi A., Micah Carroll, and Colin Allen.
``Truthfulness Despite Weak Supervision: Evaluating and Training LLMs Using Peer Prediction.''
\emph{International Conference on Learning Representations}, 2026.

\bibitem[Saha(2026)]{Saha2026}
Saha, S.
``The Agent Action Decision Protocol (AADP): Per-Action Authorization for AI Agents.''
IETF Internet-Draft, September 2026.

\bibitem[Samuelson(1938)]{Samuelson1938}
Samuelson, P. A. (1938).
A note on the pure theory of consumer's behaviour.
\textit{Economica}, 5(17), 61--71.

\bibitem[Shafer and Sonnenschein(1975)]{ShaferSonnenschein1975}
Shafer, W., and Sonnenschein, H. (1975).
Equilibrium in abstract economies without ordered preferences.
\textit{Journal of Mathematical Economics}, 2(3), 345--348.

\bibitem[Shah et al.(2022)]{ShahEtAl2022}
Shah, Rohin, Vikrant Varma, Ramana Kumar, Mary Phuong, Victoria Krakovna,
Jonathan Uesato, and Zac Kenton.
``Goal Misgeneralization: Why Correct Specifications Aren't Enough for Correct Goals.''
arXiv:2210.01790, 2022.

\bibitem[Sugaya and Wolitzky(2021)]{SugayaWolitzky2021}
Sugaya, Takuo, and Alexander Wolitzky.
``The Revelation Principle in Multistage Games.''
\emph{Review of Economic Studies} 88 (2021): 1503--1540.

\bibitem[Surapani et al.(2026)]{SurapaniEtAl2026}
Surapani, Rakesh Kumar, Pradeep Kumar Dolabehera Kakitapelli,
Arun Morampudi, and Praveena Padi.
``Authorization Architectures for Tool-Using AI Agents.''
arXiv:2609.15906, 2026.

\bibitem[Tian(1989)]{Tian1989}
Tian, G. (1989).
Implementation of the Lindahl correspondence by a single-valued, feasible,
and continuous mechanism.
\textit{Review of Economic Studies}, 56(4), 613--621.

\bibitem[Tian(1993)]{Tian1993}
Tian, Guoqiang.
``Implementing Lindahl Allocations by a Withholding Mechanism.''
\emph{Journal of Mathematical Economics} 22 (1993): 169--179.

\bibitem[Tian(2026a)]{TianGeneralDynamic2026}
Tian, Guoqiang.
``A General Theory of Dynamic Mechanism Design with Endogenous Information
and Hidden Actions.
Working paper, first draft May 2026, revised October 2026.

\bibitem[Tian(2026b)]{Tian2026}
Tian, Guoqiang.
``Endogenous Information Design and Informational Evolution in Dynamic Economic
Systems with Incentives and Learning.''
Working paper, September 2026.

\bibitem[van der Weij et al.(2025)]{VanDerWeijEtAl2025}
van der Weij, Teun, Felix Hofst\"atter, Oliver Jaffe, Samuel E. Brown,
and Francis Rhys Ward.
``AI Sandbagging: Language Models Can Strategically Underperform on Evaluations.''
\emph{International Conference on Learning Representations}, 2025.

\bibitem[Wen et al.(2024)]{WenEtAl2024}
Wen, Jiaxin, et al.
``Adaptive Deployment of Untrusted LLMs Reduces Distributed Threats.''
arXiv:2411.17693, 2024.

\bibitem[World Economic Forum(2026)]{WEF2026}
World Economic Forum.
\emph{AI Agents in Action: A Playbook for Trusted Adoption, Authorization and Scaling}.
May 2026.

\end{thebibliography}
\end{document}